\documentclass[journal,twoside,web]{ieeecolor}
\usepackage{generic}
\usepackage{cite}
\usepackage{amsmath,amssymb,amsfonts}
\usepackage{algorithmic}
\usepackage{graphicx}
\usepackage{algorithm,algorithmic,array,subcaption}
\let\labelindent\relax
\usepackage{hyperref}
\hypersetup{hidelinks=true}
\usepackage{textcomp}
\usepackage{upgreek}
\usepackage{comment}

\usepackage{mathrsfs}
\usepackage{enumitem,boondox-cal}
\usepackage{subfloat}
\usepackage{algorithm,float}  

\usepackage{amsthm}
\usepackage[T1]{fontenc}
\DeclareMathAlphabet{\mathcal}{OMS}{cmsy}{m}{n}

\usepackage[
noabbrev,
capitalise,
nameinlink,
]{cleveref}
\crefname{lemma}{Lemma}{Lemmas}
\crefname{lemmax}{Lemma}{Lemmas}
\crefname{theorem}{Theorem}{Theorems}
\crefname{theoremx}{Theorem}{Theorems}
\crefname{assumption}{Assumption}{Assumptions}
\crefname{proposition}{Proposition}{Propositions}
\crefname{propositionx}{Proposition}{Propositions}
\crefname{corollary}{Corollary}{Corollaries}
\crefname{corollaryx}{Corollary}{Corollaries}
\crefname{appen}{Appendix}{Appendices}
\crefname{figure}{Figure}{Figures}
\crefname{algorithm}{Algorithm}{Algorithms}
\crefname{remark}{Remark}{Remarks}
\crefname{section}{Section}{Sections}
\crefname{subsection}{Subsection}{Subsections}

\theoremstyle{remark} 
\newtheorem{theorem}{Theorem}
\newtheorem{proposition}{Proposition}

\newtheorem{propositionx}{Proposition}[section]
\newtheorem{lemmax}{Lemma}[section]
\newtheorem{corollaryx}{Corollary}[section]
\newtheorem{assumption}{Assumption}
\newtheorem{definition}{Definition}

\theoremstyle{remark} 
\newtheorem{remark}{Remark}
\newtheorem{remarkx}{Remark}[section]

\def\tr{^\top}

\newcommand{\Absl}[1]{\lVert #1 \rVert}
\newcommand{\abs}[1]{\left| #1\right|}
\newcommand{\Abs}[1]{\left\| #1\right\|}
\newcommand{\FRAC}[2]{\left. #1 \middle/ #2 \right.}
\def\mE{\mathbb{E}}
\def\mP{\mathbb{P}}
\def\as{\text{a.s.}}

\def\col{\operatorname{col}}
\def\Bone{\mathbf{1}}

\def\Neighbour{\mathcal{N}}
\def\mL{\mathcal{L}}
\def\mG{\mathcal{G}}

\def\udf{\underline{\mathtt{f}}}
\def\diag{\text{diag}}

\newcommand{\MODIFY}[1]{#1}
\newcommand{\MODIFYY}[1]{#1}
\newcommand{\MODIFYYY}[1]{#1}

\usepackage{tikz}
\usetikzlibrary{shapes.geometric}
\usetikzlibrary{arrows,arrows.meta}

\tikzstyle{ellipse}=[draw, rectangle, minimum width=2.8em, rounded corners=6pt,line width=0.5pt]
\tikzstyle{pxsbx}=[trapezium, trapezium left angle=75, trapezium right angle=105, minimum width=3em, text centered, draw = black, fill=white,line width=0.5pt]
\tikzstyle{lingxing}=[draw,diamond,shape aspect=3,inner sep = 0.4pt,thick,font=\itshape,line width=0.5pt]
\usetikzlibrary{positioning, arrows.meta, calc}
\tikzset{
	arrow1/.style = {
		draw = black, thick, -{Latex[length = 4mm, width = 1.5mm]},
	}
}
\tikzset{
	nonterminal/.style = {
		rectangle,
		align = center,
		minimum size = 6mm,
		very thick,
		draw = red!50!black!50,
		top color = white,
		bottom color = red!50!black!20,
	}
}
\tikzset{
	terminal/.style = {
		rectangle,
		align = center,
		minimum size = 6mm,
		rounded corners = 3mm,
		very thick,
		draw = black!50,
		top color = white,
		bottom color = black!20,
	}
}

\def\BibTeX{{\rm B\kern-.05em{\sc i\kern-.025em b}\kern-.08em
    T\kern-.1667em\lower.7ex\hbox{E}\kern-.125emX}}
\begin{document}
\title{
	Privacy-Preserving Distributed Estimation with Limited Data Rate
}
\author{Jieming Ke, Jimin Wang, \IEEEmembership{Member, IEEE}, and Ji-Feng Zhang, \IEEEmembership{Fellow, IEEE}
\thanks{This work was supported by National Natural Science Foundation of China under Grants 62433020 and T2293772. Corresponding author: Ji-Feng Zhang.}
\thanks{Jieming Ke is with 
	the State Key Laboratory of Mathematical Sciences, Academy of Mathematics and Systems Science, Chinese Academy of Sciences, Beijing 100190, and also with the School of Mathematical Sciences, University of Chinese Academy of Sciences, Beijing 100049, China. (e-mail: kejieming@amss.ac.cn)}
\thanks{Jimin Wang is with 
	the School of Automation and Electrical Engineering, University of Science and Technology Beijing, Beijing 100083; and also with the Key Laboratory of Knowledge Automation for Industrial Processes, Ministry of Education, Beijing
	100083, China (e-mail: jimwang@ustb.edu.cn)}
\thanks{Ji-Feng Zhang is with 
	the School of Automation and Electrical Engineering, Zhongyuan University of Technology, Zhengzhou 450007, Henan Province; and also with the State Key Laboratory of Mathematical Sciences, Academy of Mathematics and	Systems Science, Chinese Academy of Sciences, Beijing 100190, China. (e-mail: jif@iss.ac.cn)}
}

\maketitle

\begin{abstract}
	This paper focuses on the privacy-preserving distributed estimation  problem with a limited data rate, where the observations are the sensitive information. 
	Specifically, a binary-valued quantizer-based privacy-preserving distributed estimation algorithm is developed, which improves the algorithm's privacy-preserving capability and simultaneously reduces the communication costs. The algorithm's privacy-preserving capability, measured by the Fisher information matrix, is dynamically enhanced over time. 
	Notably, the Fisher information matrix of the output signals with respect to the sensitive information converges to zero at a polynomial rate, and \MODIFYY{the improvement in privacy brought by the quantizers is quantitatively characterized as a multiplicative effect.}
	Regarding the communication costs, each sensor transmits only 1 bit of information to its neighbours at each time step. Additionally, the assumption on the negligible quantization error for real-valued messages is not required. While achieving the requirements of privacy preservation and reducing communication costs, 
	the algorithm ensures that its estimates converge almost surely to the true value of the unknown parameter by establishing a co-design guideline for the time-varying privacy noises and step-sizes. 
	 A polynomial almost sure convergence rate is obtained, and then the trade-off between privacy and convergence rate is established. 
	Numerical examples demonstrate the main results.
\end{abstract}


\section{Introduction}
\label{sec:introduction}

\IEEEPARstart{D}{istributed} estimation has received close attention in the past decade due to its extensive applications in various fields, such as biological networks, online machine learning, and smart grids \cite{Sayed2013Diffusion,Kar2012TIT}. Different from traditional centralized estimation, the observations of distributed estimation are collected by different sensors in the communication network. Therefore, network communication is required to fuse the observations from each sensor. However, in actual distributed systems, observations may contain sensitive information, and the network communication may lead to sensitive information leakage. For example, medical research usually requires clinical observation data of patients from different hospitals, which involves the patients' personal data \cite{Li2018DPDOL,Smith2011Privacy}. Motivated by this practical background, this paper investigates how to achieve distributed estimation while preventing leakage of the observations.

The current literature offers several privacy-preserving methods for distributed systems. One of the methods is the homomorphic encryption method \cite{ZhuHM2018PPDO,Hadjicostis2020PPDA,Regueiro2021Privacy-enhancing,Li2024scis}, which provides high-dimensional security while ensuring control accuracy. Another commonly used method is the stochastic obfuscation method \cite{HuGQ2021Nash,Farokhi2019ensuring,Gratton2021TIFS,Wei2024SCIS,Nozari2017DPAC,Jayaraman2018NIPS}, which has the advantages of low computational complexity and high timeliness. 
Other methods include the state decomposition method  \cite{WangYQ2019decomposition} and the privacy mask method  \cite{Altafini2020system-theoretic}. Specifically, for the distributed estimation problem, \cite{Liu2018gossip} proposes an observation perturbation differential privacy method, while \cite{WangJM2022IJRNC,WangJM2023JSSC,WangYQ2023online} give output perturbation differential privacy methods. The methods in \cite{WangJM2022IJRNC,WangJM2023JSSC,Liu2018gossip,WangYQ2023online} provide strong privacy, but their communication relies on the transmission of real-valued messages, which causes quantization errors and high communication costs when applied to digital networks based on quantized communications.

For distributed estimation problem under quantized communications, \cite{Kar2012TIT} proposes a distributed estimation algorithm under infinite-level quantized communications. Under limited data rate, \cite{Alistarh2017QSGD,Wangni2018Gradient,Koloskova2019Decentralized,Kovalev2021linearly,Carpentiero2024Compressed,Michelusi2022Finite,Carpentiero2023distributed} investigate the quantization methods following the biased compression rule \cite{Michelusi2022Finite}. The realization of limited data rate relies on an assumption on the negligible quantization error for real-valued messages. Without such an assumption, \cite{Sayin2014} designs a single-bit diffusion strategy under binary-valued communications, and \cite{Nassif2023Quantization} proposes a distributed estimation algorithm based on variable-rate quantizers. But, the algorithms' estimates in \cite{Sayin2014,Nassif2023Quantization} do not converge to the true value. 

To achieve privacy preservation and quantized communications simultaneously, quantizer-based privacy-preserving methods have recently received significant attention \cite{Chen2023PPDED,GaoL2021DPC,Lang2023federated,WangYQ2023quantization,LiuL2023CDC}. 
For example, \cite{Chen2023PPDED} proposes a dynamic quantization-based homomorphic encryption method. For higher computational efficiency, \cite{Lang2023federated} designs special privacy noises and dither signals in dithered lattice quantizers, and \cite{GaoL2021DPC} proposes a dynamic coding scheme with Laplacian privacy noises. Both of them realize $ \epsilon $-differential privacy. \cite{WangYQ2023quantization,LiuL2023CDC} treat the dither signals in the quantizers as privacy noises, and prove that using the dithered lattice quantizer (i.e., ternary quantizer in \cite{WangYQ2023quantization} and stochastic quantizer in \cite{LiuL2023CDC}) can achieve $ (0,\delta) $-differential privacy. 
Intuitively, the incorporation of quantizers increases the difficulty for adversaries to infer sensitive information, but existing works lack quantitative characterization of improvement in privacy brought by quantizers. 

Building on the above excellent works, this paper answers several key questions. How to simultaneously achieve privacy preservation, ensure a limited data rate, and guarantee the convergence of estimates for distributed estimation problems? How to quantitatively characterize the improvement in privacy brought by quantizers? And, what is the trade-off between privacy and convergence rate under our quantizer-based method for the distributed estimation problem? 

To answer these questions, a novel binary-valued quantizer-based privacy-preserving distributed estimation algorithm is proposed. For quantized communications, our algorithm achieves message transmission at a limited data rate by using the comparison of adjacent binary-valued signals. Based on this technique, the biased compression rule \cite{Michelusi2022Finite} for quantizers can be avoided, and hence, our analysis does not rely on the assumption on the negligible quantization error for real-valued messages as in \cite{Alistarh2017QSGD,Wangni2018Gradient,Koloskova2019Decentralized,Kovalev2021linearly,Carpentiero2024Compressed,Michelusi2022Finite,Carpentiero2023distributed}, and the information receiver is not required to know the upper bounds of the estimate's norms to decode the quantized data as in \cite{WangYQ2023quantization,Lang2023federated}. For the privacy, dither signals in quantized communications \cite{WangY2024JSSC,WangY2024efficient} are also treated as privacy noises. In addition, binary-valued quantizers also make sensitive information more difficult to infer. 

To quantitatively characterize the improvement in privacy of our quantizer-based method, Fisher information is adopted as the privacy metric because of its following advantages. Firstly, Fisher information is related to the Cram{\'e}r-Rao lower bound, and thereby can intuitively quantify the capability of potential adversaries to infer sensitive information. Hence, Fisher information has been adopted as the privacy metric for the privacy-preserving smart meters \cite{Farokhi2017batteries}, the privacy-preserving database query \cite{Farokhi2019ensuring} and privacy-preserving average consensus \cite{Ke2023Differentiated}. 
\MODIFYYY{More importantly, compared to differential privacy, Fisher information provides a more precise measure of privacy. Differential privacy has difficulty in accurately characterizing the improvement of deterministic quantizers on privacy \cite{GaoL2021DPC}, whereas Fisher information can explicitly capture this improvement by reducing adversaries' estimation accuracy for the sensitive information.}
	


By using Fisher information, the binary-valued quantizer-based privacy-preserving distributed estimation algorithm is shown to achieve a dynamically enhanced privacy. The Fisher information matrix of the output signals with respect to the sensitive information converges to zero at a polynomial rate. The dynamic enhanced privacy can be achieved because under our algorithm, the privacy noises can be constant or even increasing, in contrast to the decaying ones in existing privacy-preserving distributed estimation algorithms \cite{WangJM2022IJRNC,WangJM2023JSSC}. 
Furthermore, dynamically enhanced privacy can be used to reveal the trade-off between privacy and convergence rate. When privacy is enhanced at a higher rate, the convergence rate will decrease.

This paper proposes a novel binary-valued quantizer-based privacy-preserving distributed estimation algorithm. 
The main contributions of this paper are summarized as follows.

\begin{enumerate}
	\MODIFYY{
		\item The improvement in privacy brought by the quantizers has been quantitatively characterized. 
		Specifically, under Gaussian privacy noises, the introduction of binary-valued quantizers can improve the privacy-preserving level by at least $\frac{\pi}{2}$ times, which reveals the impact of quantizers on the privacy-preserving level as a multiplicative effect.}
	\item The privacy-preserving capability of the proposed algorithm is dynamically enhanced. The Fisher information matrix of the output signals with respect to the sensitive information converges to zero at a polynomial rate. 
	\MODIFYY{Notably, the privacy analytical framework is unified for general privacy noise types, including Gaussian, Laplacian and even heavy-tailed ones.}
	\item Under the proposed algorithm, each sensor transmits only 1 bit of information to its neighbours at each time step. This is the lowest data rate among existing quantizer-based privacy-preserving distributed algorithms \cite{Chen2023PPDED,Lang2023federated,GaoL2021DPC,WangYQ2023quantization}. Additionally, the assumption on the negligible quantization error for real-valued messages  \cite{Alistarh2017QSGD,Wangni2018Gradient,Koloskova2019Decentralized,Kovalev2021linearly,Carpentiero2024Compressed,Michelusi2022Finite,Carpentiero2023distributed} is not required.
	\item 
	\MODIFYY{A co-design guideline for the time-varying privacy noises and step-sizes under quantized communications is proposed to ensure the almost sure convergence of the algorithm.}
	A polynomial almost sure convergence rate is also obtained. 
	\item The trade-off between privacy and convergence rate is established. Better privacy implies a slower convergence rate, and vice versa. Furthermore, the sensor operators can determine their own preference for the privacy and convergence rate by properly selecting privacy noises and step-sizes. 
\end{enumerate}

The rest of this paper is organized as follows. \cref{sec:prob_form} formulates the problem, and introduces the Fisher information-based privacy metric. \cref{sec:algo} proposes our privacy-preserving distributed estimation algorithm. \cref{sec:privacy} analyzes the privacy-preserving capability of the algorithm. \cref{sec:conv} proves the almost sure convergence of the algorithm, and calculates the almost sure convergence rate. \cref{sec:tradeoff} establishes the trade-off between privacy and convergence rate. \cref{sec:simu} uses numerical examples to demonstrate the main results. \cref{sec:concl} gives a concluding remark for this paper.

\subsection*{Notation}

In the rest of the paper, $ \mathbb{N} $, $\mathbb{R}$, $\mathbb{R}^{n}$, and $ \mathbb{R}^{n\times m} $ are the sets of natural numbers, real numbers, $n$-dimensional real vectors, and $ n \times m $-dimensional real matrices, respectively.  $\|x\|$ is the Euclidean norm for vector $ x $, and $\| A \|$ is the induced matrix norm for matrix $ A $.  
$ A^+ $ is the pseudo-inverse of matrix $ A $. 
$ I_n $ is an $n\times n$ identity matrix.  $\mathbb{I}_{\{\cdot\}}$ denotes the indicator function, whose value is 1 if its argument (a formula) is true; and 0, otherwise. $ \Bone_n $ is the $ n $-dimensional vector whose elements are all ones. $ \diag\{\cdot\} $ denotes the block matrix formed in a diagonal manner
of the corresponding numbers or matrices. $ \col\{\cdot\} $ denotes the column vector stacked by the corresponding vectors. $ \otimes $ denotes the Kronecker product. 
$ \mathcal{N}(0,\sigma^2) $, $ \text{\textit{Lap}}(0,b) $ and $ \text{\textit{Cauchy}}(0,r) $ represent Gaussian distribution with density function $ \frac{1}{\sqrt{2\pi}\sigma} \exp\left( - {x^2}/{2\sigma^2} \right) $, Laplacian distribution with density function $ \frac{1}{2b} \exp\left( - {\abs{x}}/{b} \right) $ and Cauchy distribution with density function $ \left. {1}\middle/{\left(\pi r \left[ 1 + \left( x \middle/ r \right)^2 \right]\right)} \right. $, respectively.


\section{Preliminaries and problem formulation}\label{sec:prob_form}

\subsection{Preliminaries on graph theory}\label{ssec:graph}

In this paper, the communication graph is switching among topology graphs $ \mathcal{G}^{(1)}, \ldots, \mG^{(M)} $, where $ \mG^{(u)} = \left( \mathcal{V}, \mathcal{E}^{(u)}, \right.$  $ \left. \mathcal{A}^{(u)} \right) $ for all $ u = 1,\ldots,M $. $ \mathcal{V} = \{ 1, \ldots, N \} $ is the set of the sensors. $ \mathcal{E}^{(u)} \in \{ (i,j) : i,j \in \mathcal{V} \} $ is the edge set.  $ \mathcal{A}^{(u)} = (a_{ij}^{(u)})_{N\times N} $ represents the symmetric weighted adjacency matrix of the graph whose elements are all non-negative. $ a_{ij}^{(u)} > 0 $ if and only if $ (i,j)\in \mathcal{E}^{(u)} $. Besides, $ \Neighbour_{i}^{(u)} = \{j : (i,j)\in\mathcal{E}^{(u)}\}$ denotes the sensor $ i $'s the neighbour set corresponding to the graph $ \mG^{(u)} $. Define Laplacian matrix as $ \mL^{(u)} = \mathcal{D}^{(u)} - \mathcal{A}^{(u)} $, where $ \mathcal{D}^{(u)} = \diag\left(\sum_{i\in\Neighbour_1}a_{i1}^{(u)},\ldots,\sum_{i\in\Neighbour_N}a_{iN}^{(u)}\right) $. 

The union of $ \mathcal{G}^{(1)}, \ldots, \mG^{(M)} $ is denoted by $ \mG = \left( \mathcal{V}, \mathcal{E}, \mathcal{A} \right) $, where $ \mathcal{E} = \bigcup_{r=1}^{M} \mathcal{E}^{(u)} $, and $ \mathcal{A} = \sum_{u=1}^{M} \mathcal{A}^{(u)} $. Besides, set $ \Neighbour_{i} = \{j : (i,j)\in\mathcal{E}\}$. 

\begin{assumption}\label{assum:connect}
	The union graph $ \mG $ is connected.
\end{assumption}

\MODIFY{
	\begin{remark}
		Instead of requiring instantaneous connectivity at each time step in \cite{Ke2024SCDE,WangYQ2023quantization,Michelusi2022Finite}, \cref{assum:connect} only requires the joint connectivity of the switching topologies $ \mathcal{G}^{(1)}, \ldots, \mG^{(M)} $ over time. 
	\end{remark}
}


The communication graph at time $ k $, denoted by $ \mathtt{G}_k $, is associated with a homogeneous Markovian chain $ \{\mathtt{m}_k:k\in\mathbb{N}\} $ with a state space $ \{1,\ldots,M\} $, transition probability $ p_{uv} = \mP\{\mathtt{m}_k=v|\mathtt{m}_{k-1}=u\} $, and stationary distribution $ \pi_u = \lim_{k\to\infty} \mP\{\mathtt{m}_k = u\} $. If $ \mathtt{m}_k = u $, then $ \mathtt{G}_k = \mathcal{G}^{(u)} $. 
Denote $ q_{ij,k} = \mP\{(i,j)\in\mathcal{E}^{(\mathtt{m}_k)}\} $.
For convenience, $ \mathcal{E}^{(\mathtt{m}_k)} $, $ a_{ij}^{(\mathtt{m}_k)} $, $ \mathcal{N}_i^{(\mathtt{m}_k)} $ and $ \mathcal{L}^{(\mathtt{m}_k)} $ are abbreviated as $ \mathtt{E}_k $, $ \mathtt{a}_{ij,k} $, $ \mathtt{N}_{i,k} $ and $ \mathtt{L}_k $, respectively, in the rest of this paper. 

\begin{remark}
	Markovian switching graphs can be used to model the link failures \cite{Kar2009consensus,ZhangQ2012}. $ \mathtt{a}_{ij,k} > 0 $ implies that the communication link between the sensors $ i $ and $ j $ is normal. $ \mathtt{a}_{ij,k} = 0 $ implies that the communication link fails. 
\end{remark}

\begin{remark}
	Given $ p_{u,1} = \mP\{\mathtt{G}_1 = \mG^{(u)}\} $, $ q_{ij,k} $ can be recur- sively obtained by $q_{ij,k}=\sum_{u\in\mathbb{G}_{ij}} p_{u,k}$, $p_{u,k+1} = \mP\{\mathtt{G}_{k+1} = \mG^{(u)}\} = \sum_{v=1}^{M} p_{v,k} p_{vu}$,
	where $ \mathbb{G}_{ij} = \{u:(i,j)\in\mathcal{E}^{(u)}\} $. By Theorem 1.2 of \cite{Seneta}, we have $ q_{ij,k} = \sum_{u\in\mathbb{G}_{ij}} \pi_{u} + O\left( \lambda_p^k \right) $ for some $ \lambda_p \in (0,1) $. 
	Specifically, when the initial distribution $ \{p_{u,1}: u =1,\ldots,M\} $ is the stationary distribution $ \{\pi_{u}: u =1,\ldots,M\} $, we have $ q_{ij,k}=\sum_{u\in\mathbb{G}_{ij}} \pi_{u} $. 
\end{remark}

\subsection{Observation model}

In the multi-sensor system coupled by the Markovian switching graphs, the sensor $ i $ observes the unknown parameter $ \theta \in \mathbb{R}^n $ from the observation model 
\begin{align}\label{sys:observation}
	\mathtt{y}_{i,k} = \mathtt{H}_{i,k} \theta + \mathtt{w}_{i,k},
	\ i = 1,\ldots,N,
	\ k\in\mathbb{N},
\end{align}
where $ \theta $ is the unknown parameter, $ k $ is the time index, $ \mathtt{w}_{i,k} \in \mathbb{R}^{m_i} $ is the observation noise, and $ \mathtt{y}_{i,k} \in \mathbb{R}^{m_i} $ is the observation.
$ \mathtt{H}_{i,k} \in \mathbb{R}^{m_i\times n} $ is the random measurement matrix. 

Assumptions for the observation model \eqref{sys:observation} are given as follows. 

\begin{assumption}\label{assum:input}
	The random measurement matrix $ \mathtt{H}_{i,k} $ is not necessarily available, but its mean value $ \bar{H}_i $ is known. 
	$ \sum_{i=1}^{N} \bar{H}_i^\top \bar{H}_i $ is invertible.  
\end{assumption}


\begin{remark}
	 \MODIFYY{The invertibility on $ \sum_{i=1}^{N} \bar{H}_i^\top \bar{H}_i $ is the cooperative observability assumption  \cite{Kar2012TIT,Jakovetic2023heavy,Vukovic2024heavy}.  
	 Additionally, \cite{Kar2012TIT} uses the unknown $\mathtt{H}_{i,k}$ to model sensor failure. 
	 Under \cref{assum:input}, the subsystem of each sensor is not necessarily observable. $ \bar{H}_i $ can be even $ 0 $ for some sensor $ i $. Hence, communications between sensors are necessary to fuse data collected by different sensors.}
\end{remark}



\begin{assumption}\label{assum:noise}
	$ \{ \mathtt{w}_{i,k}, \mathtt{H}_{i,k} : i\in\mathcal{V}, k\in\mathbb{N} \} $ is an independent sequence\footnote{\MODIFY{A random variable sequence is said to be independent if any pair of random variables in the sequence are independent of each other.}} such that
	\begin{equation}\label{condi:w_rho}
		\mE \mathtt{w}_{i,k} = 0,\ \sup_{i\in\mathcal{V},\ k\in\mathbb{N}} \mE \Abs{\mathtt{w}_{i,k}}^{\rho} < \infty, 
	\end{equation}
	\begin{equation}\label{condi:H_rho}
		\sup_{i\in\mathcal{V},\ k\in\mathbb{N}} \mE \Abs{\mathtt{H}_{i,k}-\bar{H}_i}^{\rho} < \infty,
	\end{equation} 
	for some $ \rho > 2 $, and independent of $ \{\mathtt{G}_k: k\in\mathbb{N}\} $. 
\end{assumption}

\begin{remark}
	If $ \rho $ in \eqref{condi:w_rho} and \eqref{condi:H_rho} takes different values, for example $ \rho_1 $ and $ \rho_2 $, respectively, then by Lyapunov inequality \cite{Shiryaev}, \eqref{condi:w_rho} and \eqref{condi:H_rho} still hold for $ \rho = \min \{ \rho_1, \rho_2 \} $. 
\end{remark}

\subsection{Dynamically enhanced privacy}

This section will formulate the privacy-preserving distributed estimation problem. 
\MODIFY{Notably, in some medical research \cite{Smith2011Privacy}, the observation $ \mathtt{y}_{i,k} $ is the private clinical observation data held by different hospitals. 
Such privacy scenarios motivate us to protect the observation $ \mathtt{y}_{i,k} $. }

The set containing all the information transmitted in network is denoted as $ \mathtt{S} = \left\{ \mathtt{s}_{ij,k} : (i,j)\in\mathtt{E}_k, k\in\mathbb{N} \right\} $, where $ \mathtt{s}_{ij,k} $ is the signal that the sensor $ i $ transmits to the sensor $ j $ at time $ k $. 
Then, we introduce Fisher information as a privacy metric to quantify the privacy-preserving capability. 

\begin{definition}[Fisher information, \cite{zamir1998Fisher}]
	Fisher information of $ \mathtt{S} $ with respect to sensitive information $ \MODIFYYY{\mathtt{y}} $ is defined as 
	\begin{align*}
		\mathcal{I}_\mathtt{S}(\mathtt{y})
		= \mE \left[ \left[\frac{\partial \ln(\mP(\mathtt{S}|\mathtt{y}))}{\partial \mathtt{y}}\right]\left[\frac{\partial \ln(\mP(\mathtt{S}|\mathtt{y}))}{\partial \mathtt{y}}\right]\tr \middle| \mathtt{y} \right]. 
	\end{align*}
	Given a random variable $ \MODIFYYY{\mathtt{x}} $, the conditional Fisher information is defined as  
	\begin{align*}
		\mathcal{I}_\mathtt{S}(\mathtt{y}|\mathtt{x})
		= \mE \left[ \left[\frac{\partial \ln(\mP(\mathtt{S}|\mathtt{x},\mathtt{y}))}{\partial \mathtt{y}}\right]\left[\frac{\partial \ln(\mP(\mathtt{S}|\mathtt{x},\mathtt{y}))}{\partial \mathtt{y}}\right]\tr \middle| \mathtt{y} \right].
	\end{align*}
\end{definition}

Fisher information can be used to quantify the privacy-preserving capability because of the following proposition.

\begin{proposition}[Cram{\'e}r-Rao lower bound, \cite{zamir1998Fisher}]\label{prop:crbound}
	If $ \mathcal{I}_\mathtt{S}(\mathtt{y}) $ is invertible, then for any unbiased estimator $ \hat{\mathtt{y}} = \hat{\mathtt{y}}(\mathtt{S}) $ of $ \mathtt{y} $, $ \mE(\hat{\mathtt{y}}-\mathtt{y})(\hat{\mathtt{y}}-\mathtt{y})\tr \geq \mathcal{I}_\mathtt{S}^{-1}(\mathtt{y}) $. 
\end{proposition}

\begin{remark}
	Fisher information is a natural privacy metric \cite{Farokhi2017batteries,Farokhi2019ensuring,Ke2023Differentiated}, because by \cref{prop:crbound}, smaller $ \mathcal{I}_\mathtt{S}(\mathtt{y}) $ implies less information leaks, and vice versa. 
	Besides, Fisher information is closely related to other common privacy metrics. For example, 
	\cite{Barnes2020Fisher} reveals the positive correlation between $ \epsilon $-differential privacy and upper bounds of Fisher information. There are other advantages for Fisher information as the privacy metric. Firstly, compared to mutual information \cite{MIP2016,Entropy2022}, Fisher information is unrelated to the \textit{a priori} knowledge on the sensitive information, and hence, it suitable for privacy-preserving distributed estimation where the distribution of the sensitive information $ \mathtt{y}_{i,k} $ contains the unknown parameter $ \theta $. Secondly, compared to maximal leakage \cite{MaxLeak2020}, Fisher information allows $ \mathtt{y}_{i,k} $ to be both continuous and discrete. 
	\MODIFYYY{Thirdly, compared to differential privacy \cite{GaoL2021DPC}, Fisher information has advantage in characterizing the improvement in privacy brought by deterministic quantizers.} 
\end{remark}

Under the above privacy metric, we aim to achieve
the dynamically enhanced privacy as defined below. 

\begin{definition}
	If the privacy-preserving capability of an algorithm is said to be dynamically enhanced, then given any $ i \in \mathcal{V} $ and $ k $ with $ \mE\mathcal{I}_\mathtt{S}(\mathtt{y}_{i,k}) > 0 $, there exists $ T > k $ such that $ \mE\mathcal{I}_\mathtt{S}(\mathtt{y}_{i,t}) < \mE\mathcal{I}_\mathtt{S}(\mathtt{y}_{i,k}) $ for all $ t \geq T $. 
\end{definition}

\begin{remark}
	By \cref{lemma:enhance} in \cref{appen}, $ \lim_{k\to\infty} \mE \mathcal{I}_{\mathtt{S}} (\mathtt{y}_{i,k}) $  $ = 0 $ is sufficient for the dynamically enhanced privacy. 
\end{remark}

\subsection{Problem of interest}

This paper mainly seeks to develop a new privacy-preserving distributed estimation algorithm which can simultaneously achieve
\begin{enumerate}[leftmargin=1.4em]
	\item The privacy is dynamically enhanced over time; 
	\item The sensor $ i $ transmits only 1 bit of information to its neighbour $ j $ at each time step;
	\item And, the estimates for all sensors converge to the true value of the unknown parameter almost surely. 
\end{enumerate}

\section{Privacy-preserving algorithm design}\label{sec:algo}

This subsection will firstly give the binary-valued quantizer-based method, and then propose a binary-valued quantizer-based privacy-preserving distributed estimation algorithm. 

The traditional consensus+innovations type distributed estimation algorithms \cite{Kar2012TIT,Kar2011rate} fuse the observations through the transmission of estimates $ \hat{\uptheta}_{i,k-1} $, which would lead to sensitive information leakage. For the privacy issue, the following binary-valued quantizer-based method is designed to transform them into binary-valued signals before transmission. Firstly, if $ k = n q+l $ for some $ q \in \mathbb{N} $ and $ l \in \{1,\ldots,n\} $, then the sensor $ i $ generates $ \varphi_k $ as the $ n $-dimensional vector whose $ l $-th element is $ 1 $ and the others are $ 0 $. The sensor $ i $ uses $ \varphi_{k} $ to compress the previous local estimate $ \hat{\uptheta}_{i,k-1} $ into the scalar $\mathtt{x}_{i,k} = \varphi_{k}\tr \hat{\uptheta}_{i,k-1}$. 
Secondly, the sensor $ i $ generates the privacy noise $ \mathtt{d}_{ij,k} $ with distribution $ F_{ij,k}(\cdot) $ for all $ j\in\mathtt{N}_{i,k} $. Then, given the threshold $ C_{ij} $, the sensor $ i $ generates the binary-valued signal
\begin{equation}\label{step:quantize}
	\mathtt{s}_{ij,k}
 	= 
	\begin{cases}
		1,& \text{if}\ \mathtt{x}_{i,k} + \mathtt{d}_{ij,k} \leq C_{ij};\\
		-1,& \text{otherwise}.
	\end{cases}
\end{equation}

\begin{remark}
	The threshold $ C_{ij} $ can be any real number. From the communication perspective, \MODIFYYY{when $\mathtt{d}_{ij,k}$ is unbiased, the communication efficiency of the quantizer \eqref{step:quantize} will be better if $ C_{ij} $ gets closer to $ \mathtt{x}_{i,k} $ \cite{WangY2024JSSC}. Since $ \mathtt{x}_{i,k} $ should converge to $\varphi_k\tr\theta$, we can set $C_{ij} = \arg\min_{C} \mE_{\theta\sim p(\theta)} \Abs{\theta - C \mathbf{1}_n}^2 = \sum_{i=1}^{n} \mu_{\theta,i} $, where $p(\theta)$ is the \textit{a priori} distribution for $ \theta $, $\mu_{\theta,i}$ is the $i$-th element of $\mE_{\theta\sim p(\theta)} \theta$, and $ \mathbf{1}_n $ is the $n$ dimensional vector with all elements being 1.} 
\end{remark}

By using the binary-valued quantizer-based method \eqref{step:quantize}, a novel privacy-preserving distributed estimation algorithm is proposed in \cref{algo:DE}. 

\begin{algorithm}[H]
	\caption{Binary-valued quantizer-based privacy-preserving distributed estimation algorithm.}
	\label{algo:DE}
	\begin{algorithmic}
		\STATE \textbf{Input:}
		initial estimate sequence $ \{\hat{\theta}_{i,0}\} $, threshold sequence $ \{C_{ij}\} $ with $ C_{ij} = C_{ji} $, noise distribution sequence $ \{F_{ij,k}(\cdot)\} $ with $ F_{ij,k}(\cdot) = F_{ji,k}(\cdot) $, step-size sequences $ \{\alpha_{ij,k}\} $ with $ \alpha_{ij,k} = \alpha_{ji,k} > 0 $ and $ \{\beta_{i,k}\} $ with $ \beta_{i,k} > 0 $.
		\STATE \textbf{Output:}
		estimate sequence $ \{\hat{\uptheta}_{i,k}\} $.
		\STATE \textbf{for} $ k = 1,2,\ldots, $ \textbf{do}
		\STATE \quad \textbf{Privacy preservation:} Use the binary-valued quantizer-based method \eqref{step:quantize} to transform the previous local estimate $ \hat{\uptheta}_{i,k-1} $ into the binary-valued signal $ \mathtt{s}_{ij,k} $, and send the binary-valued signal $ \mathtt{s}_{ij,k} $ to the neighbour $ j $.
		\STATE \quad \textbf{Information fusion:} Fuse neighbourhood information. 
		\begin{equation*}
			\check{\uptheta}_{i,k} = \hat{\uptheta}_{i,k-1} + \varphi_k \sum_{j\in \mathtt{N}_{i,k}} \alpha_{ij,k} \mathtt{a}_{ij,k} \left( \mathtt{s}_{ij,k} - \mathtt{s}_{ji,k} \right).
		\end{equation*}
		\vspace{-0.6em}
		\STATE \quad \textbf{Estimate update:} Use $ \mathtt{y}_{i,k} $ to update the local estimate.
		\begin{equation*}
			\hat{\uptheta}_{i,k} = \check{\uptheta}_{i,k} + \beta_{i,k}\bar{H}_{i}\tr \left(\mathtt{y}_{i,k}-\bar{H}_{i} \hat{\uptheta}_{i,k-1}\right), 
		\end{equation*}
		where $ \bar{H}_i \MODIFYY{= \mE \mathtt{H}_{i,k}} $ as in \cref{assum:input}. 
		\STATE \textbf{end for}
	\end{algorithmic}
\end{algorithm}

\begin{remark}\label{remark:quan}
	\MODIFYYY{
	In \cref{algo:DE}, the expectation of $\mathtt{s}_{ij,k}-\mathtt{s}_{ji,k}$ is $ F\!\big(C_{ij}-\varphi_k^\top \hat{\uptheta}_{i,k-1}\big)
	-
	F\!\big(C_{ij}-\varphi_k^\top \hat{\uptheta}_{j,k-1}\big) $, 	which captures the difference in the statistical properties of the quantized data under different estimates. Therefore, when $F(\cdot)$ is strictly monotonically increasing, this communication mechanism is able to achieve information fusion of quantized data over a distributed network without relying on controlling an upper bound on the quantization error. This communication mechanism has the following two advantages:}
	\begin{enumerate}[leftmargin=1.4em]
		\item \MODIFYYY{It does not rely on the assumption that certain real-valued messages can be transmitted with negligible quantization error as required by many existing stochastic compression methods \cite{Alistarh2017QSGD,Wangni2018Gradient,Koloskova2019Decentralized,Kovalev2021linearly,Carpentiero2024Compressed,Michelusi2022Finite,Carpentiero2023distributed}. Compression methods in these works satisfy} the biased compression rule
		\begin{align}\label{condi:biasc}
			\sqrt{\mE \left[ \lVert \mathtt{Q}(\mathtt{x}) - \mathtt{x} \rVert^2 \middle| \mathtt{x} \right]} \leq \kappa \lVert \mathtt{x} \rVert + \iota, 
		\end{align}
		where $ \mathtt{Q}(\cdot) $ is the stochastic compression operator, and $ \kappa\in[0,1) $, $ \iota \geq 0 $. However, when the decoder does not know the \textit {a priori} upper bound of $ \lVert \MODIFYYY{\mathtt{x}} \rVert $, under \eqref{condi:biasc}, $ \mathtt{Q}(\cdot) $ cannot compress a real-valued vector into finite bits. This is because under finite-level quantizer $ \mathtt{Q}(\cdot) $, $ \mathtt{Q}(\MODIFYYY{\mathtt{x}}) $ is uniformly bounded, leading to 
		\begin{align*}
			\lim_{\lVert \mathtt{x} \rVert \to \infty} \frac{\sqrt{\mE \left[ \lVert \mathtt{Q}(\mathtt{x}) - \mathtt{x} \rVert^2 \middle| \mathtt{x} \right]}}{\lVert \mathtt{x} \rVert} = 1 > \kappa,
		\end{align*}
		which is contradictory to \eqref{condi:biasc}. Due to this limitation,  \cite{Alistarh2017QSGD,Wangni2018Gradient,Koloskova2019Decentralized,Kovalev2021linearly,Carpentiero2024Compressed,Michelusi2022Finite,Carpentiero2023distributed} assume that certain real-valued messages can be transmitted with negligible quantization error. \MODIFYYY{The quantizer \eqref{step:quantize} does not follow \eqref{condi:biasc}, and hence, does not require this assumption.}
		\MODIFYYY{\item It does not require a decoder, thereby reducing the computational complexity. For comparison, dynamic quantization  method \cite{Chen2023PPDED,GaoL2021DPC}, empirical measurement method \cite{ZhaoYL2019consensus} and recursive projection method \cite{WangY2026distributed} require decoding the quantized data into approximate values of the original estimates before they can be used for information fusion. In contrast, \cref{algo:DE} avoids the decoder using $\mathtt{s}_{ij,k}-\mathtt{s}_{ji,k}$, and hence, achieves a similar computational complexity with the traditional ones without quantization \cite{ZhangQ2012,WangJM2022IJRNC}.}
	\end{enumerate}
\end{remark}


	\begin{remark}
		\cref{algo:DE} does not rely on the value of $ \mathtt{H}_{i,k} $ due to \cref{assum:input}. 
		Therefore, under \cref{algo:DE},
		preserving the privacy of $\mathtt{y}_{i,k}$ inherently ensures the privacy of $ \mathtt{H}_{i,k} $. 
	\end{remark}
	
	\begin{remark}\label{remark:C_phi}
		\MODIFYYY{\cref{algo:DE} achieves a data rate of 1 bit per time. Such an extremely low data rate can be achieved because it does not rely on transmission of real-valued messages as shown in \cref{remark:quan}, and the introduction of $\varphi_k$ makes the data rate irrelevant of the dimension of the estimates. For comparison, \cite{Lang2023federated,WangYQ2023quantization} requires no less than $ n $ and $\lceil n \log_2 3 \rceil$ bits per time, respectively.}
		When not pursuing such an extremely low data rate, $\varphi_k$ can be \MODIFYYY{replaced by a periodic selection matrix $\Psi_k \in \mathbb{R}^{\psi\times n}$. Under this replacement, the sensor generates a $\psi$-bit $s_{ij,k}$ by applying independent privacy noise and binary-valued quantization to each element of the vector $\mathtt{x}_{i,k} = \Psi_{k}\tr \hat{\uptheta}_{i,k-1}$. This approach improves the estimation accuracy with data rate of $\psi$ bits per time.}		
	\end{remark}


Assumptions for privacy noises and step-sizes in \cref{algo:DE} are given as follows. 

\begin{assumption}\label{assum:privacy_noise}
	The privacy noise sequence $ \{\mathtt{d}_{ij,k}: (i,j)\in\mathcal{E}, k\in\mathbb{N}\} $ satisfies
	\begin{enumerate}[label={\roman*)},leftmargin=1.4em]
		\item The density function $ f_{ij,k}(\cdot) $ of $ \mathtt{d}_{ij,k} $ exists;
		\item $ \eta_{ij,k} = \sup_{x\in\mathbb{R}} \frac{f_{ij,k}^2\left(x\right)}{F_{ij,k}(x)\left( 1 - F_{ij,k}(x) \right)} < \infty $;
		\item There exists a sequence $ \{\zeta_{ij,k}\} $ such that for all compact set $ \mathcal{X} $, $ \inf_{(i,j)\in\mathcal{E},k\in\mathbb{N},x\in\mathcal{X}} \frac{f_{ij,k}(x)}{\zeta_{ij,k}} > 0 $;
		\item $ \{\mathtt{d}_{ij,k}:(i,j) \in\mathcal{E}, k\in\mathbb{N}\} $ is an independent sequence, and independent of $ \{\mathtt{w}_{i,k},\mathtt{G}_k,\mathtt{H}_{i,k}:i\in \mathcal{V},\ k\in\mathbb{N}\} $. 
	\end{enumerate}
\end{assumption}

\begin{remark}
	The \cref{assum:privacy_noise} ii) is for the privacy analysis, and iii) is for the convergence analysis.
\end{remark}

\begin{assumption}\label{assum:step}
	The step-size sequences $ \{\alpha_{ij,k}:(i,j)\in\mathcal{E}, k\in\mathbb{N}\} $ and $ \{\beta_{i,k}:i\in\mathcal{V},k\in\mathbb{N}\} $ satisfy
	\begin{enumerate}[label={\roman*)},leftmargin=1.4em]
		\item $ \sum_{k=1}^{\infty} \alpha^2_{ij,k} < \infty $ and $ 
		\alpha_{ij,k} = O\left( \alpha_{ij,k+1} \right) $ for all $ (i,j)\in\mathcal{E} $;
		\item $ \sum_{k=1}^{\infty} \beta^2_{i,k} < \infty $ and $
		\beta_{i,k} = O\left( \beta_{i,k+1} \right) $ for all $ i\in\mathcal{V} $; 
		\item $ \sum_{k=1}^{\infty} z_k = \infty $ for $ z_k = \min\{ \alpha_{ij,k}\zeta_{ij,k}: (i,j)\in\mathcal{E}\} \cup \{ \beta_{i,k}: i\in\mathcal{V} \} $ \MODIFYY{and  $\zeta_{ij,k}$ follows \cref{assum:privacy_noise} iii).} 
	\end{enumerate}
\end{assumption}

\begin{remark}
	\cref{assum:step} is the stochastic approximation condition for distributed estimation. \MODIFYY{Such step-sizes are typically set as polynomial functions of $k$ \cite{Ke2024SCDE}.} When the step-sizes are all polynomial, \cref{assum:step} iii) is equivalent to $ \sum_{k=1}^{\infty} \alpha_{ij,k}\zeta_{ij,k} = \infty $ for all $ (i,j)\in\mathcal{E} $ and $ \sum_{k=1}^{\infty} \beta_{i,k} = \infty $ for all $ i\in\mathcal{V} $. 
	In \cref{assum:step}, the step-sizes are not necessarily the same for all sensors, in contrast to the centralized step-sizes adopted in many distributed algorithms \cite{Kar2012TIT,HuGQ2021Nash,ZhuHM2018PPDO}. Therefore, the sensor operators can properly select their step-sizes based on their own requirements. 
\end{remark}


\MODIFYY{
	\begin{remark}
		\cref{assum:privacy_noise,assum:step} indicate that \cref{algo:DE} establishes a time-varying design method for privacy noises and step-sizes under quantized communications upon the algorithm framework of \cite{Ke2024SCDE}: 
		\begin{enumerate}[label={\roman*)},leftmargin=1.4em]
			\item Diversified design of privacy noises is allowed to meet different privacy-preserving requirements.  By Lemma 5.3 of \cite{WangY2024JSSC} and Lemmas \ref{lemma:zeta}-\ref{lemma:Cauchy} in \cref{appen2}, \cref{assum:privacy_noise} accommodates not only standard differential privacy noises like Laplacian and Gaussian ones but also heavy-tailed Cauchy noise for outlier protection \cite{Ito2021heavy}; 
			\item Privacy noises are allowed to be time-varying for a better privacy-preserving level. By \cref{prop:assum4} in \cref{appen2}, \cref{assum:step} permits polynomially increaing privacy noises under a maximum allowable growth rate.
			\item  A co-design guideline for the privacy noises and step-sizes is presented to ensure the convergence of \cref{algo:DE} under quantized communications.   
			Crucially, this guideline differs fundamentally from the non-quantized case \cite{WangYQ2023online,WangJM2024bipartite,WangYQ2024Tailoring}. In non-quantized settings, larger noise increases communication signal variance, requiring smaller step-sizes to mitigate its impact on estimation accuracy. Under quantized communications, however, the communication signal variance remains uniformly bounded regardless of noise magnitude.  Instead, larger noise reduces the previous estimate information carried by the communication signal, necessitating larger step-sizes to ensure efficient information utilization. 
		\end{enumerate}
	\end{remark}
}

\section{Privacy analysis}\label{sec:privacy}

The section will analyze the privacy-preserving capability of \cref{algo:DE}. 
\cref{thm:privacy/finite} below proves that privacy-preserving capability of \cref{algo:DE} is dynamically enhanced over time. \cref{thm:priv/improve}  quantify the improvement of the privacy-preserving capability brought by the binary-valued quantizers. 


\begin{theorem}\label{thm:privacy/finite}
	Suppose Assumptions \ref{assum:input}, \ref{assum:noise}, 
	\ref{assum:privacy_noise} i), ii), iv) and \ref{assum:step} ii), iii) hold, and
	\begin{enumerate}[label={\roman*)},leftmargin=1.4em]
		\item $ \beta_{i,k} \lambda_{\max} (Q_i) < 1 $, where $ Q_i = \bar{H}_i\tr\bar{H}_i $ and $ \lambda_{\max} (Q_i) $ is the maximum eigenvalue of $ Q_i $;
		\item $ \sum_{t=1}^\infty\! \prod_{l=1}^{t} \eta_{ij,t} \! \left( 1 \! - \! \lambda_{\min}^+(Q_i)\beta_{i,l} \right)^2 \! <\! \infty $, where $ \lambda_{\min}^+(Q_i) $ is the minimum positive eigenvalue of $ Q_i $. 
	\end{enumerate}
	Then,  
	\begin{align}\label{concl:privacy}
		& \mE \mathcal{I}_{\mathtt{S}}(\mathtt{y}_{i,k}) \nonumber\\
		\leq & \sum_{j\in\mathcal{N}_i} \sum_{t=k+1}^\infty \beta_{i,k}^2  q_{ij,t} \eta_{ij,t} \left( \prod_{l=k+1}^{t-1} \left( 1 - \lambda^+_{\min}(Q_i) \beta_{i,l} \right) \! \right)^2 \!\!\! \bar{H}_i\bar{H}_i\tr \nonumber\\
		< & \infty, 
	\end{align} 
	where $ q_{ij,k} $ is given in \cref{ssec:graph}. 
	\MODIFY{Furthermore, if
	\begin{enumerate}[label={\roman*)},leftmargin=1.4em]
		\addtocounter{enumi}{+2}
		\item $ p_{u,1}=\mP\{\mathtt{G}_1=\mG^{(u)}\} = \pi_u $; 
		\item $ \eta_{ij,k} \leq \frac{\eta_{ij,1}}{k^{2 \epsilon_{ij}}} $ with $ \eta_{ij,1} > 0  $ and $ \epsilon_{ij} \geq 0 $;
		\item $ \beta_{i,k} = \frac{\beta_{i,1}}{k^{\delta_i}} $ if $ k \geq k_{i,0} $; and $ 0 $, otherwise, where $ \delta_i \in (1/2,1] $ , $ \beta_{i,1} \in (0, k_{i,0}^{\delta_i}) $ and $ 2 \lambda_{\min}^+(Q_i) \beta_{i,1} + 2 \epsilon_{ij} > 1 $; 
	\end{enumerate}
	then 
	\begin{align}\label{concl:privacy_rate}
		\mE\mathcal{I}_{\mathtt{S}} (\mathtt{y}_{i,k}) 
		\leq& \sum_{j\in\mathcal{N}_i} \sum_{u\in\mathbb{G}_{ij}} \pi_{u} R_{ij,k} \beta_{i,k}\eta_{ij,k} \bar{H}_{i} \bar{H}_{i}\tr \nonumber\\
		=& O\left( \sum_{j\in\mathcal{N}_i} \frac{1}{k^{\delta_i+2\epsilon_{ij}}} \right), 
	\end{align}
	where 
	\begin{align*}
		R_{ij,k}\! = \!
		\begin{cases}
			\frac{\beta_{i,1}}{2\lambda^+_{\min}(Q_i)\beta_{i,1}+2 \epsilon_{ij}-1} 
			\frac{(k+1)^{2\lambda^+_{\min}(Q_i)\beta_{i,1}} k^{2\epsilon_{ij}}}{(k-1)^{2\lambda^+_{\min}(Q_i)\beta_{i,1}+2\epsilon_{ij}}}, \text{if}\ \delta_i = 1 ;\\
			\frac{\beta_{i,1}}{ 2 \lambda_{\min}^+(Q_i) \beta_{i,1} - (\delta_i-2 \epsilon_{ij}) k^{\delta_i-1}}, \qquad\quad\  \text{if}\ \delta_i \in (1/2,1).\\
		\end{cases}
	\end{align*}
	Therefore, \cref{algo:DE} achieves the dynamically enhanced privacy.}
\end{theorem}

\begin{proof}
	Firstly, we expand the sequence $ \mathtt{S} = \{ \mathtt{s}_{ij,k} : (i,j)\in\mathtt{E}_k, k\in\mathbb{N} \} $ to $ \breve{\mathtt{S}} = \left\{ \mathtt{s}_{ij,k} : (i,j)\in\mathcal{E}, k\in\mathbb{N} \right\} $. Note that we have expanded the noise sequence $ \{\mathtt{d}_{ij,k}: (i,j)\in\mathtt{E}_k, k\in\mathbb{N}\} $ to $ \{\mathtt{d}_{ij,k}: (i,j)\in\mathcal{E}, k\in\mathbb{N}\} $ in \cref{assum:privacy_noise}. Then, for all $ (i,j) \in \mathcal{E} $, define 
	\begin{align}
		\mathtt{a}_{ij,k}^\prime & = \begin{cases}
			1, & \text{if}\ (i,j)\in\mathtt{E}_k;\\
			0, & \text{otherwise},
		\end{cases} \label{def:a'} \\
		\mathtt{s}_{ij,k}^\prime & = \begin{cases}
			1, & \text{if}\ \mathtt{x}_{i,k} + \mathtt{d}_{ij,k} \leq C_{ij};\\
			-1, & \text{otherwise}. 
		\end{cases} \nonumber
	\end{align}
	For $ (i,j)\in\mathcal{E}\setminus \mathtt{E}_k $, define $ \mathtt{s}_{ij,k} = 0 $. Then, $ \mathtt{s}_{ij,k} = \mathtt{a}_{ij,k}^\prime \mathtt{s}_{ij,k}^\prime $ and $ \mE\mathcal{I}_\mathtt{S}(\mathtt{y}_{i,k}) = \mE \mathcal{I}_{\breve{\mathtt{S}}} (\mathtt{y}_{i,k}) $. 
	
	Note that $ \mathcal{I}_{\{\mathtt{y}_{i,l}:{l \neq k}\}}(\mathtt{y}_{i,k}) = 0 $. 
	Then, by \cref{coro:chain} in \cref{appen}, 
	\begin{align}\label{ineq:FI_trans}
		\mathcal{I}_{\breve{\mathtt{S}}} (\mathtt{y}_{i,k}) 
		\leq  \mathcal{I}_{\breve{\mathtt{S}}} (\mathtt{y}_{i,k}|\{\mathtt{y}_{i,l}:{l \neq k}\}) 
	\end{align}
	Note that for any $ (u,v) \in \mathcal{E} $, $ \mathtt{d}_{uv,t} $ is independent of $ \mathtt{M}^{-}_{i,t-1,k} $ and $ \mathtt{y}_{i,k} $,  and $ \mathtt{x}_{u,t} $ is $ \sigma(\mathtt{M}^{-}_{i,t-1,k}\cup\{\mathtt{y}_{i,k}\}) $-measurable, where $ \sigma(\cdot) $ is the minimum $ \sigma $-algebra containing the corresponding set, and $ \mathtt{M}^{-}_{i,t,k} = \{\mathtt{y}_{i,l}:{l \neq k}\} \cup \{ \mathtt{s}_{uv,l} : (u,v)\in\mathcal{E},l\leq t \} $. Then, given $ \mathtt{M}^{-}_{i,t-1,k} $ and $ \mathtt{y}_{i,k} $, one can get $ \{\mathtt{s}_{uv,t}^\prime:(u,v)\in\mathcal{E}\} $ is independent. Besides, given $ \mathtt{M}^{-}_{i,t-1,k} $ and $ \mathtt{y}_{i,k} $, we have $ \mathtt{s}_{uv,t}^\prime $ is uniquely determined by $ \mathtt{d}_{uv,t} $, and $ \mathtt{a}_{uv,t}^\prime $ is uniquely determined by $ \mathtt{G}_k $. Then, by \cref{assum:privacy_noise}, given $ \mathtt{M}^{-}_{i,t-1,k} $ and $ \mathtt{y}_{i,k} $, one can get $ \{\mathtt{s}_{uv,t}^\prime:(u,v)\in\mathcal{E}\} $ is independent of $ \{a_{uv,t}^\prime:(u,v)\in\mathcal{E}\} $. Therefore, by \cref{lemma:Fisher_independent} in \cref{appen},
	\begin{align}\label{eq:I[lemma]}
		\mathcal{I}_{\breve{\mathtt{S}}} (\mathtt{y}_{i,k}|\{\mathtt{y}_{i,l}:{l \neq k}\}) 
		= & \sum_{t=1}^{\infty} \sum_{(u,v)\in\mathcal{E}} \mathcal{I}_{\mathtt{s}_{uv,t+1}}(\mathtt{y}_{i,k}|\mathtt{M}^{-}_{i,t,k}) \nonumber\\
		= & \sum_{t=1}^{\infty} \sum_{j\in\mathcal{N}_i} \mathcal{I}_{\mathtt{s}_{ij,t+1}}(\mathtt{y}_{i,k}|\mathtt{M}^{-}_{i,t,k}),
	\end{align} 
	
	Denote $ \bar{\mathtt{q}}_{ij,t} = \mP\{(i,j)\in\mathcal{E}|\mathtt{G}_{t-1}\} $, and note that $ \{\mathtt{d}_{ij,k}:k\in\mathbb{N}\} $ is independent. Then, we have 
	\begin{align*}
		& \ln \mP\left\{\mathtt{s}_{ij,t} \middle| \mathtt{y}_{i,k},\mathtt{M}^{-}_{i,t,k-1}\right\}\cr
		= & \ln \left(\bar{\mathtt{q}}_{ij,t} F_{ij,t}(C_{ij}-\mathtt{x}_{i,t})\right)\mathbb{I}_{\{\mathtt{s}_{ij,t}=1\}}
		+ \ln (1-\bar{\mathtt{q}}_{ij,t}) \mathbb{I}_{\{\mathtt{s}_{ij,t}=0\}}  \nonumber\\
		& + \ln \left(\bar{\mathtt{q}}_{ij,t} \left(1- F_{ij,t}(C_{ij}-\mathtt{x}_{i,t})\right) \right) \mathbb{I}_{\{\mathtt{s}_{ij,t}=-1\}},
	\end{align*}
	which implies
	\begin{align}\label{eq:partial_lnP}
		& \frac{\partial}{\partial \mathtt{y}_{i,k}} \ln \left( \mP\left\{\mathtt{s}_{ij,t}\middle|\mathtt{y}_{i,k},\mathtt{M}^{-}_{i,t-1,k} \right\} \right) \nonumber\\
		= & \frac{\partial}{\partial \mathtt{y}_{i,k}} \ln \left( \bar{\mathtt{q}}_{ij,t} F_{ij,t}(C_{ij} - \mathtt{x}_{i,t}) \right) \mathbb{I}_{\{\mathtt{s}_{ij,t}=1\}} \nonumber\\
		& + \frac{\partial}{\partial \mathtt{y}_{i,k}} \ln \left(\bar{\mathtt{q}}_{ij,t} \left(1- F_{ij,t}(C_{ij} - \mathtt{x}_{i,t})\right) \right) \mathbb{I}_{\{\mathtt{s}_{ij,t}=-1\}} \nonumber\\
		= & - \frac{f_{ij,t}\left( C_{ij} - \mathtt{x}_{i,t} \right)}{F_{ij,t}(C_{ij}-\mathtt{x}_{i,t})} \frac{\partial \mathtt{x}_{i,t}}{\partial \mathtt{y}_{i,k}} \mathbb{I}_{\{\mathtt{s}_{ij,t}=1\}} \nonumber\\
		& + \frac{f_{ij,t}\left( C_{ij} - \mathtt{x}_{i,t} \right)}{1 - F_{ij,t}(C_{ij}-\mathtt{x}_{i,t})} \frac{\partial \mathtt{x}_{i,t}}{\partial \mathtt{y}_{i,k}} \mathbb{I}_{\{\mathtt{s}_{ij,t}=-1\}}. 
	\end{align}
	
	Now, we calculate $ \frac{\partial \mathtt{x}_{i,t}}{\partial \mathtt{y}_{i,k}} $. If $ k \geq t $, then $ \frac{\partial \mathtt{x}_{i,t}}{\partial \mathtt{y}_{i,k}} = 0 $. If $ k < t $, then 
	by \cref{lemma:+} in \cref{appen},  
	\begin{align}\label{eq:x/y}
		\frac{\partial \mathtt{x}_{i,t}}{\partial \mathtt{y}_{i,k}}
		= & \beta_{i,k} \bar{H}_i J_i \left( \prod_{l=k+1}^{t-1} (I_n - \beta_{i,l} Q_i) \right)\tr \varphi_{t} \nonumber\\
		= & \beta_{i,k} \bar{H}_i \left( \prod_{l=k+1}^{t-1} (J_i - \beta_{i,l} Q_i ) \right)\tr \varphi_{t}, 
	\end{align}
	where $ J_i = Q_i^+ Q_i $. 
	Hence, by \eqref{ineq:FI_trans}-\eqref{eq:x/y} and \cref{lemma:+eigen,lemma:sum_prod} in \cref{appen}, 
	\begin{align}\label{ineq:EI<infty}
		& \mE\mathcal{I}_{\mathtt{S}} (\mathtt{y}_{i,k})
		= \mE\mathcal{I}_{\breve{\mathtt{S}}} (\mathtt{y}_{i,k}) \nonumber\\
		= & \sum_{t=1}^{\infty} \sum_{j\in\mathcal{N}_i} \mE \left[ \left(\frac{\partial}{\partial \mathtt{y}_{i,k}} \ln \left( \mP \left\{\mathtt{s}_{ij,t}\middle|\mathtt{y}_{i,k},\mathtt{M}^{-}_{i,t-1,k}  \right\}\right)\right)\right. \nonumber\\
		& \cdot \left. \left(\frac{\partial}{\partial \mathtt{y}_{i,k}} \ln \left( \mP \left\{\mathtt{s}_{ij,t} \middle| \mathtt{y}_{i,k},\mathtt{M}^{-}_{i,t-1,k}  \right\}\right)\right)\tr \right] \nonumber\\
		= & \!\sum_{j\in\mathcal{N}_i}\! \sum_{t=k+1}^\infty \beta_{i,k}^2  \mE \left[ \frac{\bar{\mathtt{q}}_{ij,t} f_{ij,t}^2\left( C_{ij} - \mathtt{x}_{i,t} \right)}{F_{ij,t}(C_{ij}-\mathtt{x}_{i,t})\left( 1 - F_{ij,t}(C_{ij}-\mathtt{x}_{i,t}) \right)} \right]\nonumber\\
		\cdot & \bar{H}_i \left( \prod_{l=k+1}^{t-1} (J_i - \beta_{i,l} Q_i) \right)\tr \!\! \varphi_{t} \varphi_{t}\tr  \left( \prod_{l=k+1}^{t-1} (J_i - \beta_{i,l} Q_i) \right) \bar{H}_i\tr \nonumber\\
		\leq & \! \sum_{j\in\mathcal{N}_i} \sum_{t=k+1}^\infty \beta_{i,k}^2  q_{ij,t} \eta_{ij,t} \left( \prod_{l=k+1}^{t-1} \left( 1 - \lambda^+_{\min}(Q_i) \beta_{i,l} \right)\! \right)^2 \! \bar{H}_i\bar{H}_i\tr \nonumber\\
		< & \infty. 
	\end{align}
	Now, we prove \eqref{concl:privacy_rate}. If $ k < k_{i,0} $, then $ \beta_{i,k} = 0 $, which together with \eqref{concl:privacy} implies $ \mE \mathcal{I}_{\mathtt{S}}(\mathtt{y}_{i,k}) = 0 $. 
	
	If $ k \geq k_{i,0} $, then by Lemma A.2 of \cite{WangJM2024bipartite}, one can get
	\begin{align}\label{ineq:H*prod*phi*prod*H}
		& \bar{H}_i \left( \prod_{l=k+1}^{t-1} (J_i - \beta_{i,l} Q_{i}) \right)\tr \varphi_{t} \varphi_{t}\tr  \left( \prod_{l=k+1}^{t-1} (J_i - \beta_{i,l} Q_{i}) \right) \bar{H}_i\tr \nonumber\\
		\leq & \left( \prod_{l=k+1}^{t-1} \left( 1 - \frac{\lambda^+_{\min}(Q_i)\beta_{i,1}}{l^{\delta_i}} \right) \right)^2 \bar{H}_i\bar{H}_i\tr \nonumber\\
		\leq &
		\begin{cases}
			\left( \frac{k+1}{t-1} \right)^{2\lambda^+_{\min}(Q_i)\beta_{i,1}}\bar{H}_i\bar{H}_i\tr,\\
			\qquad\qquad\qquad\qquad\qquad\qquad\qquad\qquad \text{if}\ \delta_i = 1 ; \\
			\exp\left( \frac{2\lambda^+_{\min}(Q_i)\beta_{i,1}}{1-\delta_i} \left( (k+1)^{1-\delta_i} - t^{1-\delta_i} \right) \right) \bar{H}_i\bar{H}_i\tr, \\
			\qquad\qquad\qquad\qquad\qquad\qquad\qquad\qquad \text{if}\ \delta_i < 1.  
		\end{cases}
	\end{align}
	Therefore, if $ \delta_i = 1 $, then
	\begin{align}\label{ineq:Fisher/delta=1}
		\mE\mathcal{I}_{\mathtt{S}}(\mathtt{y}_{i,k}) 
		\leq & \sum_{j\in\mathcal{N}_i} \! \sum_{t=k+1}^\infty \!\! \beta_{i,k}^2 q_{ij,t} \eta_{ij,t}\left( \! \frac{k+1}{t-1} \! \right)^{2\lambda^+_{\min}(Q_i)\beta_{i,1}}\!\!\!\bar{H}_i\bar{H}_i\tr \nonumber\\
		\leq &  \sum_{j\in\mathcal{N}_i} \beta_{i,1}^2 \eta_{ij,1} \left( \sum_{u\in\mathbb{G}_{ij}} \pi_{u} \right) \frac{(k+1)^{2\lambda^+_{\min}(Q_i)\beta_{i,1}}}{k^{2}} \nonumber\\ 
		&\cdot \sum_{t=k+1}^\infty \frac{\bar{H}_i\bar{H}_i\tr}{(t-1)^{2\lambda^+_{\min}(Q_i)\beta_{i,1} + 2 \epsilon_{ij}}} \nonumber\\
		\leq & \sum_{j\in\mathcal{N}_i} \beta_{i,1}^2 \eta_{ij,1} \left( \sum_{u\in\mathbb{G}_{ij}} \pi_{u} \right) \frac{(k+1)^{2\lambda^+_{\min}(Q_i)\beta_{i,1}}}{k^{2}} \nonumber\\
		&\cdot \frac{(k-1)^{1-2\lambda^+_{\min}(Q_i)\beta_{i,1}-2 \epsilon_{ij}}}{2\lambda^+_{\min}(Q_i)\beta_{i,1}+2 \epsilon_{ij}-1} \bar{H}_i\bar{H}_i\tr \nonumber\\
		\leq & \sum_{j\in\mathcal{N}_i} \left( \sum_{u\in\mathbb{G}_{ij}} \pi_{u} \right) \frac{\beta_{i,1}}{2\lambda^+_{\min}(Q_i)\beta_{i,1}+2 \epsilon_{ij}-1}  
		 \nonumber\\
		&\cdot  \frac{(k+1)^{2\lambda^+_{\min}(Q_i)\beta_{i,1}} k^{2\epsilon_{ij}}}{(k-1)^{2\lambda^+_{\min}(Q_i)\beta_{i,1}+2\epsilon_{ij}}} \beta_{i,k}\eta_{ij,k}   \bar{H}_i\bar{H}_i\tr. 
	\end{align}
	If $ \delta_i < 1 $, then $ 2 \lambda_{\min}^+(Q_i) \beta_{i,1} k^{1-\delta_i} > 1-2 \epsilon_{ij} > \delta_{i}-2 \epsilon_{ij} $, which together with \cref{lemma:sum_exp} in \cref{appen} implies 
	\begin{align}\label{ineq:Fisher/delta<1}
		& \mE\mathcal{I}_{\{\mathtt{s}_{ij,t}:j\in\mathcal{N}_i,t\in\mathbb{N}\}} (\mathtt{y}_{i,k}) \nonumber\\
		\leq & \! \sum_{j\in\mathcal{N}_i} \! \sum_{t=k+1}^\infty \!\! \beta_{i,k}^2 \eta_{ij,t} q_{ij,t} \frac{\exp\left( \frac{2\lambda^+_{\min}(Q_i)\beta_{i,1}}{1-\delta_i}  (k+1)^{1-\delta_i} \right)}{\exp\left( \frac{2\lambda^+_{\min}(Q_i)\beta_{i,1}}{1-\delta_i}  t^{1-\delta_i}  \right)} \bar{H}_i\bar{H}_i\tr
		\nonumber\\
		\leq & \! \sum_{j\in\mathcal{N}_i} \left(\sum_{u\in\mathbb{G}_{ij}} \pi_{u}\right)  \frac{\beta_{i,1}^2 \eta_{ij,1}}{k^{2\delta_i}} \exp\left( \frac{2\lambda^+_{\min}(Q_i)\beta_{i,1}}{1-\delta_i}  (k+1)^{1-\delta_i} \right)
		\nonumber\\
		& \cdot \sum_{t=k+1}^\infty \frac{\exp\left( - \frac{2\lambda^+_{\min}(Q_i)\beta_{i,1}}{1-\delta_i}  t^{1-\delta_i}  \right)}{t^{2\epsilon_{ij}}}  \bar{H}_i\bar{H}_i\tr
		\nonumber\\
		\leq & \! \sum_{j\in\mathcal{N}_i} \left(\sum_{u\in\mathbb{G}_{ij}} \pi_{u}\right) \frac{\beta_{i,1}}{ 2 \lambda_{\min}^+(Q_i) \beta_{i,1} - (\delta_{i}-2 \epsilon_{ij}) k^{\delta_i-1}} \nonumber\\
		&\cdot \beta_{i,k} \eta_{ij,k} \bar{H}_i\bar{H}_i\tr. 
	\end{align}
	Hence by $ \beta_{i,k} \eta_{ij,k} = O\left( \frac{1}{k^{\delta_i+2\epsilon_i}} \right) $,  \eqref{concl:privacy_rate} is obtained. Then, the theorem can be proof by \cref{lemma:enhance}.	
\end{proof}

\MODIFYYY{
\begin{remark}
	In \cref{thm:privacy/finite}, i) is easy to be achieved since $ \beta_{i,k} $ converges to 0. ii) provides a general condition on minimum privacy noises required to ensure finite $\mE \mathcal{I}_{\mathtt{S}}(\mathtt{y}_{i,k})$, with iv) and v)  specifying polynomial step-size and noise design schemes satisfying this condition. iii) implies that the initial distribution of the Markovian chain is its stationary distribution. This condition is naturally satisfied if the chain has been running from the infinite past.
\end{remark}
}

\begin{remark}
	By \eqref{concl:privacy_rate}, \MODIFYYY{$ \mE\mathcal{I}_{\mathtt{S}} (\mathtt{y}_{i,k}) = O \left( \sum_{j\in\Neighbour_i} \beta_{i,k} \eta_{ij,k}\right) $, which indicates that the privacy can be enhanced over time.} 
	The sensor $ i $'s operator can control the convergence rate of $ \mE\mathcal{I}_{\mathtt{S}} (\mathtt{y}_{i,k}) $ by properly selecting the step-size $ \beta_{i,k} $ and the privacy noise distributions. 
	\MODIFYYY{Firstly, smaller $\beta_{i,k}$ implies a better privacy. Besides, privacy noises can be designed as increasing ones to enhance privacy. In contrast, the constant dithered lattice quantizer methods \cite{Lang2023federated,WangYQ2023quantization,LiuL2023CDC} and the decaying noise method \cite{GaoL2021DPC}  are difficult to utilize for improving the convergence rate of $\mE\mathcal{I}_{\mathtt{S}} (\mathtt{y}_{i,k}) $.
	Additionally, \eqref{concl:privacy} reveals that the privacy-preserving level measured by Fisher information is proportional to communication frequency. Therefore, the Markovian switching topology can improve privacy by reducing communication frequency.}
\end{remark}

\begin{remark}
	By \cref{prop:assum3,prop:assum4}, the privacy noise distributions satisfying the condition of \cref{thm:privacy/finite} include 
	$ \mathcal{N}(0,\sigma_{ij,k}^2) $ with $ \sigma_{ij,k} = \sigma_{ij,1} k ^{\epsilon_{ij}} $ and $ \text{\textit{Lap}}(0,b_{ij,k}) $ with $ b_{ij,k} = b_{ij,1} k ^{\epsilon_{ij}}$. Under such a choice of noise distribution, by \cref{thm:privacy/finite}, we have $\mE\mathcal{I}_{\mathtt{S}} (\mathtt{y}_{i,k}) = O\left( \sum_{j\in\mathcal{N}_i} \frac{\beta_{i,k}}{\mE \mathtt{d}_{ij,k}^2} \right)$. 
	Besides, the variances of the privacy noises are not necessarily finite. For example, by \cref{prop:assum3,prop:assum4}, the privacy noises can obey Cauchy distribution, which is heavy-tailed with infinite variance.  
\end{remark}

The following theorem takes Gaussian privacy noise as an example to quantify the improvement of privacy brought by the binary-valued quantizers. In the theorem, the conditional Fisher information given $\{\mathtt{y}_{i,t}: t \neq k\}$ is considered as the privacy metric to eliminate privacy-preserving effects between different observations $ \mathtt{y}_{i,k} $. 

\begin{theorem}\label{thm:priv/improve}
	Under the condition of Theorem \ref{thm:privacy/finite}, when the noise $d_{ij,k}$ is Gaussian distributed, we have
	\begin{align*}
		\mathcal{I}_{\mathtt{S}} \left(\mathtt{y}_{i,k}\middle|\left\{\mathtt{y}_{i,t} :{t \neq k}\right\}\right) \leq \frac{2}{\pi} \mathcal{I}_{\bar{\mathtt{X}}} \left(\mathtt{y}_{i,k}\middle|\left\{\mathtt{y}_{i,t}\ :{t \neq k}\right\}\right), 
	\end{align*}
	where $ \bar{\mathtt{x}}_{ij,k} = \mathtt{x}_{i,k} + \mathtt{d}_{ij,k}$ and $\bar{\mathtt{X}} = \{ \bar{\mathtt{x}}_{ij,k} : (i,j) \in \mathtt{E}_k, k\in\mathbb{N}\}$. 
\end{theorem}

\begin{proof}
	Set $\MODIFYYY{\mathtt{d}_{ij,k}} \sim \mathcal{N} (\mu_{ij,k},\sigma_{ij,k}^2)$. 
	Similar to \eqref{ineq:EI<infty} and by Lemma 5.3 of \cite{WangY2024JSSC}, we have
	\begin{align*}
		& \mathcal{I}_{\mathtt{S}} \left(\mathtt{y}_{i,k}\middle|\left\{\mathtt{y}_{i,t}:{t \neq k}\right\}\right) =  \mathcal{I}_{\breve{\mathtt{S}}} \left(\mathtt{y}_{i,k}\middle|\left\{\mathtt{y}_{i,t}:{t \neq k}\right\}\right) \nonumber\\
		\leq & \sum_{t=k+1}^{\infty} \sum_{j\in\mathcal{N}_i} \frac{2}{\pi \sigma_{ij,k}^2} q_{ij,t} \bar{\varphi}_{i,k,t} \bar{\varphi}_{i,k,t}\tr,
	\end{align*}
	where $\bar{\varphi}_{i,k,t} = \beta_{i,k} \bar{H}_i \left( \prod_{l=k+1}^{t-1} (I_n - \beta_{i,l} Q_i) \right)\tr \varphi_{t}$.

	Similarly, one can get
	\begin{align*}
		& \mathcal{I}_{\bar{\mathtt{X}}} \left(\mathtt{y}_{i,k}\middle|\left\{\mathtt{y}_{i,t}:{t \neq k}\right\}\right) =  \mathcal{I}_{\breve{\mathtt{X}}} \left(\mathtt{y}_{i,k}\middle|\left\{\mathtt{y}_{i,t}:{t \neq k}\right\}\right) \nonumber\\
		 = & \sum_{t=k+1}^{\infty} \sum_{j\in\mathcal{N}_i} \frac{1}{\sigma_{ij,k}^2} q_{ij,t} \bar{\varphi}_{i,k,t} \bar{\varphi}_{i,k,t}\tr,
	\end{align*}
	where $ \breve{\mathtt{X}} = \{ \bar{\mathtt{x}}_{ij,k}\mathtt{a}_{ij,k}^\prime : (i,j) \in \mathtt{E}_k, k\in\mathbb{N}\}$ and $\mathtt{a}_{ij,k}^\prime$ is defined in \eqref{def:a'}. 
	Thus, the theorem is proved. 
%
\end{proof}

\begin{remark}
	\cref{thm:priv/improve} proves that in the Gaussian privacy noise case, the introduction of quantizers improves the privacy-preserving capability of the algorithm by at least $ \frac{\pi}{2} $ times. Similarly, according to \cref{lemma:Cauchy}, in the Cauchy noise case, the improvement is at least $ \frac{\pi^2}{8} $ times.
	Therefore, the impact of quantizers in the privacy-preserving level is revealed as a multiplicative effect.
	And, by \cref{lemma:Lap}, in the Laplacian noise case, the introduction of the quantizers also improves the privacy-preserving capability of the algorithm, except for the case that $ \mathtt{x}_{i,k} = C $ for all $k$, which will not happen almost surely due to the randomness of $\mathtt{x}_{i,k}$. 
\end{remark}

\section{Convergence analysis}\label{sec:conv}

This section will focus on the convergence properties of \cref{algo:DE}. 
Firstly, the almost sure convergence will be proved. 
Then, the almost sure convergence rate will be obtained. 

For convenience, denote
\begin{align*}
	\tilde{\uptheta}_{i,k} &= \hat{\uptheta}_{i,k} - \theta,\ 
	\tilde{\Theta}_k = \col\{\tilde{\uptheta}_{1,k},\ldots,\tilde{\uptheta}_{N,k}\},\
	\bar{a}_{ij} = \sum_{r=1}^{M} \pi_{r} a_{ij}^{(r)}, \\ 
	\bar{\mathbb{H}} &= \diag\{ \bar{H}_{1}\tr \bar{H}_{1},\ldots, \bar{H}_{N}\tr \bar{H}_{N}\}, \
	\hat{\mathtt{F}}_{ij,k} = F_{ij,k} (C_{ij}-\mathtt{x}_{i,k}),\\
	\bar{\mathbb{H}}_{\beta,k} &= \diag\{\beta_{1,k} \bar{H}_{1}\tr \bar{H}_{1},\ldots, \beta_{N,k} \bar{H}_{N}\tr \bar{H}_{N}\}, \\
	\Phi_{i,k} & =   \varphi_{k}\sum_{j\in \mathcal{N}_i}\alpha_{ij,k} ( \mathtt{a}_{ij,k} - \bar{a}_{ij} )\left( \mathtt{s}_{ij,k} - \mathtt{s}_{ji,k} \right),\\
	\Phi^\prime_{i,k} & =  \varphi_{k}\sum_{j\in \mathcal{N}_i}\alpha_{ij,k} \bar{a}_{ij}\left( ( \mathtt{s}_{ij,k} - \mathtt{s}_{ji,k} ) - 2 ( \hat{\mathtt{F}}_{ij,k} - \hat{\mathtt{F}}_{ji,k} )\right),\\
	\mathtt{W}_k &= \col\{\beta_{1,k} \left(\mathtt{y}_{1,k}-\bar{H}_1 \theta\right), \ldots, \beta_{N,k} \left( \mathtt{y}_{N,k} - \bar{H}_N \theta \right)\},\\
	&\qquad\ + \col\{\Phi_{1,k},\ldots, \Phi_{N,k} \} + \col\{\Phi^\prime_{1,k},\ldots, \Phi^\prime_{N,k} \},\\
	\mathcal{F}_{k} &= \sigma(\{\mathtt{w}_{i,t},\mathtt{G}_t,\mathtt{H}_{i,t},\mathtt{d}_{ij,t}:i\in \mathcal{V},(i,j)\in\mathtt{E}_t,1\leq t \leq k\}) . 
\end{align*}
%
Then, $ \tilde{\Theta}_k $ is $ \mathcal{F}_{k} $-measurable. 

The following theorem proves the almost sure convergence of \cref{algo:DE}.

\begin{theorem}\label{thm:conv}
	Suppose Assumptions \ref{assum:connect}, \ref{assum:input}, \ref{assum:noise}, \ref{assum:privacy_noise} i), iii), iv) and \ref{assum:step} hold.   
	Then, the estimate $ \hat{\uptheta}_{i,k} $ in \cref{algo:DE} converges to the true value $ \theta $ almost surely. 
\end{theorem}

\begin{proof}
	By Theorem 1.2 of \cite{Seneta}, there exists $ \lambda_a \in (0,1) $ such that $ \mE \mathtt{a}_{ij,k} = \bar{a}_{ij} + O\left( \lambda_a^k \right) $.
	Then, by Assumptions \ref{assum:noise} and \ref{assum:privacy_noise} iv), we have $ \mE\left[ \mathtt{a}_{ij,k} \mathtt{s}_{ij,k} \middle| \mathcal{F}_{k-1} \right] = \bar{a}_{ij} F(C_{ij}-\mathtt{x}_{i,k}) + O\left( \lambda_a^k \right) $. Therefore, one can get 
	\begin{align*}
		& \mE\left[ \Absl{\tilde{\uptheta}_{i,k}}^2 \middle| \mathcal{F}_{k-1} \right] \\
		= & \Absl{\tilde{\uptheta}_{i,k-1}}^2 - 2 \beta_{i,k} \left( \bar{H}_i \tilde{\uptheta}_{i,k} \right)^2 \\
		& + 2 \varphi_k\tr \tilde{\uptheta}_{i,k-1} \sum_{j\in\Neighbour_{i}} \alpha_{ij,k} \bar{a}_{ij} \left(\hat{\mathtt{F}}_{ij,k} -\hat{\mathtt{F}}_{ji,k}\right) \\
		& + O\left( \beta_{i,k}^2 \left( \Absl{\tilde{\uptheta}_{i,k-1}}^2 + 1 \right) + \sum_{j\in\Neighbour_{i}} \alpha_{ij,k}^2 + \lambda_a^k \right). 
	\end{align*}
	Define $ \tilde{\mathtt{x}}_{i,k} = \varphi_k\tr \tilde{\uptheta}_{i,k-1} $. 
	By $ \mathtt{x}_{i,k} = \varphi_k\tr \hat{\uptheta}_{i,k-1} = \tilde{\mathtt{x}}_{i,k} + \varphi_k\tr \theta $, we have $ \mathtt{x}_{i,k} - \mathtt{x}_{j,k} = \tilde{\mathtt{x}}_{i,k} - \tilde{\mathtt{x}}_{j,k} $.  Then, 
	\begin{align*}
		& \sum_{i\in\mathcal{V}} \varphi_k\tr \tilde{\uptheta}_{i,k-1} \sum_{j\in\Neighbour_{i}} \alpha_{ij,k} \bar{a}_{ij} \left(\hat{\mathtt{F}}_{ij,k} -\hat{\mathtt{F}}_{ji,k}\right) \\
		= & 2 \sum_{(i,j)\in\mathcal{E}} \alpha_{ij,k} \bar{a}_{ij} \left( \mathtt{x}_{i,k} - \mathtt{x}_{j,k} \right) \left(\hat{\mathtt{F}}_{ij,k} -\hat{\mathtt{F}}_{ji,k}\right) \leq 0,  
	\end{align*}
	which implies 
	\begin{align*}
		& \mE\left[ \sum_{i\in\mathcal{V}} \Absl{\tilde{\uptheta}_{i,k}}^2 \middle| \mathcal{F}_{k-1} \right]
		\leq \sum_{i\in\mathcal{V}} \Absl{\tilde{\uptheta}_{i,k-1}}^2 \\
		& + O\left( \sum_{i\in\mathcal{V}} \beta_{i,k}^2 \left( \Absl{\tilde{\uptheta}_{i,k-1}}^2 + 1 \right) + \sum_{(i,j)\in\mathcal{E}} \alpha_{ij,k}^2 + \lambda_a^k \right).
	\end{align*}
	Hence, by Theorem 1 of \cite{robbins1971convergence}, $ \sum_{i\in\mathcal{V}} \Absl{\tilde{\uptheta}_{i,k}}^2 $ converges to a finite value almost surely. Therefore, $ \tilde{\uptheta}_{i,k} $, $ \hat{\uptheta}_{i,k} $, and $ \mathtt{x}_{i,k} $ are all bounded almost surely. 
	
	By the Lagrange mean value theorem \cite{zorich}, there exists $ \upxi_{ij,k} $ between $ C_{ij} - \mathtt{x}_{i,k} $ and $ C_{ij} - \mathtt{x}_{j,k} $ such that
	\begin{align*}
		\hat{\mathtt{F}}_{ij,k} - \hat{\mathtt{F}}_{ji,k} 
		= & f_{ij,k}(\upxi_{ij,k}) \left( \mathtt{x}_{j,k} - \mathtt{x}_{i,k} \right) \\
		= & f_{ij,k}(\upxi_{ij,k}) \left( \tilde{\mathtt{x}}_{j,k} - \tilde{\mathtt{x}}_{i,k} \right). 
	\end{align*}
	For convenience, set $ \check{\mathtt{f}}_{ij,k} = f_{ij,k}(\upxi_{ij,k}) $. By the almost sure boundedness of $ \mathtt{x}_{i,k} $ and \cref{assum:privacy_noise}, there exists $ \underline{\mathtt{f}} > 0 $ such that $ \check{\mathtt{f}}_{ij,k} \geq \udf \zeta_{ij,k} $ almost surely. 
	
	Define $ \mathtt{L}_{F,k} $ as a Laplacian matrix whose element in the $ i $-th row and $ j $-th column is $ - \alpha_{ij,k} \bar{a}_{ij} \check{\mathtt{f}}_{ij,k} $ if $ i \neq j $, and $ \sum_{l\in\Neighbour_{i}} \alpha_{1j,k}\bar{a}_{il} \check{\mathtt{f}}_{il,k} $ if $ i = j $. Then, 
	\begin{align}\label{eq:recur_Theta1}
		\tilde{\Theta}_k = \left( I_{N\times n} - \mathbb{H}_{\beta,k} - \mathtt{L}_{F,k}\otimes \varphi_k\varphi_k\tr \right)  \tilde{\Theta}_{k-1} + \mathtt{W}_k, 
	\end{align}
	and $ \mathtt{L}_{F,k} \geq z_k \udf \bar{\mL} $, where $ z_k $ is given in \cref{assum:step}  and $ \bar{\mL} = \sum_{r=1}^{M} \pi_{r} \mL^{(r)} $. 
	In addition, by Lemma 5.4 in \cite{xie2018analysis}, one can get
	\begin{align}\label{ineq:H+L}
		& \sum_{t=k-n+1}^{k} \frac{1}{z_t} \left(\bar{\mathbb{H}}_{\beta,t} + \mathtt{L}_{F,t}\otimes \varphi_t\varphi_t\tr\right)  \nonumber\\
		\geq & \sum_{t=k-n+1}^{k} \left( \bar{\mathbb{H}} + \udf \bar{\mL} \otimes \varphi_t \varphi_t\tr \right)
		\geq n \bar{\mathbb{H}} + \udf \bar{\mL} \otimes I_n > 0. 
	\end{align}
	Hence, by \cref{coro:conv} in \cref{appen}, $ \tilde{\Theta}_k $ and then $ \tilde{\uptheta}_{i,k} $ converge to $ 0 $ almost surely. 
\end{proof}

\begin{remark}
	Note that in \cref{algo:DE}, each sensor transmits 1 bit of information to its neighbours at each time step, and as analyzed in \cref{prop:assum4}, the privacy noises are allowed to be increasing.  
	Then, by \cref{thm:conv}, the estimates of \cref{algo:DE} can converge to the true value $ \theta $ even under 1 communication data rate and increasing privacy noises, which is the first to be achieved among existing privacy-preserving distributed algorithms \cite{HuGQ2021Nash,Gratton2021TIFS,WangYQ2024Tailoring}. 
\end{remark}

\begin{remark}
	In \cref{assum:privacy_noise}, the privacy noise can be heavy-tailed. Therefore, the results in \cref{thm:conv} can also be applied to the heavy-tailed communication noise case \cite{Vukovic2024heavy,Jakovetic2023heavy}.
	For \cref{algo:DE}, the key to achieving convergence with heavy-tailed noises lies in the binary-valued quantizer, which transmits noisy signals with probably infinite variances to binary-valued signals with uniformly bounded variances. 
\end{remark}

Then, the following theorem calculates the almost sure convergence rate of \cref{algo:DE}. 

\begin{theorem}\label{thm:conv_rate}
	Suppose Assumptions \ref{assum:connect}-\ref{assum:privacy_noise} hold, $ \rho >4 $ and the distribution of privacy noise $ \mathtt{d}_{ij,k} $ is $ \mathcal{N}(0,\sigma_{ij,k}^2) $ \MODIFY{(resp., $ \text{\textit{Lap}}(0,b_{ij,k}) $, $ \text{\textit{Cauchy}}(0,r_{ij,k}) $) with $ \sigma_{ij,k} = \sigma_{ij,1} k^{\epsilon_{ij}} $ (resp., $ b_{ij,k} = b_{ij,1} k^{\epsilon_{ij}} $, $ r_{ij,k} = r_{ij,1} k^{\epsilon_{ij}} $) and $ \sigma_{ij,1} = \sigma_{ji,1} > 0 $ (resp., $ b_{ij,1} = b_{ji,1} > 0 $, $ r_{ij,1} = r_{ji,1} > 0 $)}.  Given $ k_{i,0} $, set $ \alpha_{ij,k} = \frac{\alpha_{ij,1}}{k^{\gamma_{ij}}} $, $ \beta_{i,k} = \frac{\beta_{i,1}}{k^{\delta_{i}}} $ if $ k \geq k_{i,0} $; and $ 0 $, otherwise, where
	\begin{enumerate}[label={\roman*)},leftmargin=1.4em]
		\item $ \alpha_{ij,1} = \alpha_{ji,1} > 0 $, $ \gamma_{ij} = \gamma_{ji} > \frac{1}{2} $ and $ \epsilon_{ij} = \epsilon_{ji} \geq 0 $ for all $ (i,j)\in\mathcal{E} $, and $ \beta_{i,1} > 0 $ for all $ i \in \mathcal{V} $;
		\item $ \max_{(i,j)\in\mathcal{E}} \gamma_{ij} + \epsilon_{ij} < \min_{i\in\mathcal{V}}\delta_{i} \leq \max_{i\in\mathcal{V}}\delta_{i} \leq 1 $. 
	\end{enumerate}
	Then, the almost sure convergence rate of the estimation error for the sensor $ i $ is
	\begin{align*}
		\tilde{\uptheta}_{i,k} = 
		\begin{cases}
			O\left( \left. 1 \middle/ k^{\frac{\lambda_H\min_{i\in\mathcal{V}}\beta_{i,1}}{N}} \right. \right), \\
			\qquad\qquad \text{if}\ \bar{b} = 1,\ 2\underline{b}-\frac{2\lambda_H\min_{i\in\mathcal{V}} \beta_{i,1}}{N}>1;\\
			O\left( \FRAC{\ln k}{k^{\underline{b}-1/2}} \right), \\ \qquad\qquad\text{if}\ \bar{b} = 1,\ 2\underline{b}-\frac{2\lambda_H\min_{i\in\mathcal{V}} \beta_{i,1}}{N}\leq1;\\
			O\left( \left. 1 \middle/ k^{\underline{b}-\bar{b}/2 } \right. \right), \\
			\qquad\qquad\text{if}\ \bar{b} < 1,
		\end{cases}\ \as,
	\end{align*}
	where $ \lambda_H = \lambda_{\min} \left(\sum_{i=1}^{N} \bar{H}_i\tr \bar{H}_i \right) $, $ \underline{b} = \min_{(i,j)\in\mathcal{E}}\gamma_{ij} $ and  $ \bar{b} = \max_{i\in\mathcal{V}} \delta_i $. 
\end{theorem}

\begin{proof}
	By \cref{lemma:zeta}, $ \zeta_{ij,k} $ in \cref{assum:privacy_noise} can be $ \frac{1}{k^{\epsilon_{ij}}} $. In this case, $ z_k $ in \cref{assum:step} is $ \min\left\{ \frac{\alpha_{ij,1}}{k^{\gamma_{ij}+\epsilon_{ij}}} : (i,j)\in\mathcal{E}\right\}\cup\left\{\frac{\beta_{i,1}}{k^{\delta_{i}}}: i\in\mathcal{V} \right\} $. 
	
	If $ \bar{b} < 1 $, then the theorem can be proved by \eqref{eq:recur_Theta1}, \eqref{ineq:H+L} and \cref{coro:rate} in \cref{appen}. 
	
	If $ \bar{b} = 1 $, then $ k\sum_{i=1}^{N} \beta_{i,k} \bar{H}_i\tr \bar{H}_i \geq \lambda_H\min_{i\in\mathcal{V}}\beta_{i,1} $. Hence, by Lemma 5.4 of \cite{xie2018analysis}, 
	\begin{align*}
		& \frac{1}{n}\sum_{t=k-n+1}^{k} t \left(\bar{\mathbb{H}}_{\beta,t} + \mathtt{L}_{F,t}\otimes \varphi_t\varphi_t\tr\right) \\
		\geq & \frac{\lambda_H\min_{i\in\mathcal{V}}\beta_{i,1}}{N} I_{nN} + O\left( \frac{1}{k^{\tau}} \right)
	\end{align*}
	for some $ \tau > 0 $, which together with \eqref{eq:recur_Theta1} and \cref{coro:rate} implies the theorem. 
\end{proof}

\begin{remark}
	For all $ \upsilon\in(0,\frac{1}{2}) $, when $ \delta_i = 1 $, $ \gamma_{ij} > \upsilon + \frac{1}{2} $ and $ \beta_{i,1} $ is sufficiently large, by \cref{thm:conv_rate}, \cref{algo:DE} can achieve an almost sure convergence rate of $ o(1/k^\upsilon) $. The convergence rate is consistent with the classical one \cite{Kar2011rate} of distributed estimation without considering the quantized communications and privacy issues. 
\end{remark}

\begin{remark}
	\MODIFYYY{\cref{thm:conv_rate} i) and ii) provide a design guideline for polynomial step-sizes and noise levels that satisfy \cref{assum:step}. In particular, a larger noise restricts the admissible choice of $\alpha_{ij,k}$, which in turn slows down the convergence rate of the algorithm, and vice versa. Additionally, }
	by \cref{thm:conv_rate,thm:privacy/finite}, when $ \delta_i = 1 $, the best privacy level and convergence rate will be achieved simultaneously. 
\end{remark}



%
%

\section{Trade-off between privacy and convergence rate}\label{sec:tradeoff}

Based on the privacy and convergence analysis in Theorems \ref{thm:privacy/finite}-\ref{thm:conv_rate}, this section will establish the trade-off between the privacy level and the convergence rate of \cref{algo:DE}.

\begin{theorem}\label{thm:tradeoff/independent}
	Suppose Assumptions \ref{assum:connect}-\ref{assum:step} hold. Then, given $ \nu \in (\frac{1}{2},1) $, there exist step-size sequences $ \{\alpha_{ij,k}:(i,j)\in\mathtt{E}_k,k\in\mathbb{N}\} $, $ \{\beta_{i,k}:i\in\mathcal{V},k\in\mathbb{N}\} $ and the privacy noise distribution sequence $ \{F_{ij,k}(\cdot):(i,j)\in\mathtt{E}_k,k\in\mathbb{N}\} $ such that $ \mE\mathcal{I}_{\mathtt{S}} (\mathtt{y}_{i,k}) = O\left( \frac{1}{k^{\chi}} \right) $ and $ \tilde{\uptheta}_{i,k} = O\left( \frac{1}{k^{\nu-\chi/2}} \right) $ almost surely for all $ i\in\mathcal{V} $ and $ \chi \in [1,2\nu) $. 
\end{theorem}

\begin{proof}
	Consider the privacy noises obeying the Gaussian distribution $ \mathcal{N}(0,\sigma_{ij,k}^2) $ with $ \sigma_{ij,k} = \sigma_{ij,1} k^{\epsilon_{ij}} $, $ \sigma_{ij,1} = \sigma_{ji,1} > 0 $ and $ \epsilon_{ij} = \epsilon_{ji} \geq 0 $ as \cref{thm:conv_rate} and \cref{prop:assum3} . 
	
	Set $ k_{i,0} = \exp\left( \left\lfloor \frac{1}{\delta_{i}} \ln \beta_{i,1} \right\rfloor +1 \right) $, $ \delta_i = 1 $, $ \epsilon_{ij} = \frac{\chi-1}{2} $, $ \gamma_{ij} = \frac{2+\nu-\chi}{2} $, and $ \beta_{i,1} $ be any number bigger than $ \frac{2-\chi}{2\lambda_{\min}^+(Q_i)} $, where $ \lfloor \cdot \rfloor $ is the floor function.  
	The step-size $ \alpha_{ij,k} = \frac{\alpha_{ij,1}}{k^{\gamma_{ij}}} $, $ \beta_{i,k} = \frac{\beta_{i,1}}{k^{\delta_{i}}} $ if $ k \geq k_{i,0} $; and $ 0 $, otherwise. 	
	Then, the step-size conditions in \cref{thm:privacy/finite,thm:conv_rate} are achieved simultaneously. 
	By \cref{thm:privacy/finite}, $ \mE\mathcal{I}_{\mathtt{S}} (\mathtt{y}_{i,k}) = O\left( \frac{1}{k^{\delta_i+2\epsilon_{ij}}} \right) = O\left( \frac{1}{k^{\chi}} \right) $.
	By \cref{thm:conv_rate}, $ \tilde{\uptheta}_{i,k} = O\left( \ln k/k^{(1+\nu-\chi)/2} \right) =  O\left( \frac{1}{k^{\nu-\chi/2}} \right) $ almost surely. The theorem is proved. 
\end{proof}

\begin{remark}
	\MODIFYYY{In \cref{thm:tradeoff/independent}, $\chi$ decides the convergence rate of  $\mE\mathcal{I}_{\mathtt{S}} (\mathtt{y}_{i,k})$, and hence, can be regarded as the coefficient representing the desired privacy level. $\nu$ can be any scalar smaller than but very close to 1.} By \cref{thm:tradeoff/independent}, better privacy implies a slower convergence rate, and vice versa. 
	\MODIFYYY{The proof of \cref{thm:tradeoff/independent} provides a practical selection for privacy noises and step-sizes to achieve the trade-off under given $\chi$,  where the privacy noises are set to be increasing at a rate of $ k^{\frac{\chi-1}{2}} $, and the step-sizes $\alpha_{ij,k} $ and $\beta_{i,k}$ are set to be convergent at rates of $\frac{1}{k^{(2+\nu-\chi)/2}}$ and $\frac{1}{k}$, respectively.} 		
\end{remark}


\section{Simulations}\label{sec:simu}

This section will demonstrate the main results of the paper by simulation examples. 

\subsection{Numerical examples}

Consider an $ 8 $ sensor system. The communication graph sequence $ \{\mathtt{G}_k:k\in\mathbb{N}\} $ is switching among $ \mG^{(1)},\ \mG^{(2)},\ \mG^{(3)} $ and $ \mG^{(4)} $ as shown in \cref{fig:graph}.  For all $ u = 1,2,3,4 $, $ a_{ij}^{(u)} = 1 $ if $ (i,j) \in \mathcal{E}^{(u)} $; and $ 0 $, otherwise. 
The communication graph sequence $ \{\mathtt{G}_k:k\in\mathbb{N}\} $  is associated with a Markovian chain $ \{\mathtt{m}_k:k\in\mathbb{N}\} $. 
The initial probability $ p_{u,1} = \mP\{\mathtt{G}_1 = \mG^{(u)}\} = \frac{1}{4} $. 
The transition probability matrix
\begin{align*}
	P = (p_{uv})_{4\times 4} = \begin{bmatrix}
		\frac{1}{2} & \frac{1}{2} & 0 & 0 \\
		0 & \frac{1}{2} & \frac{1}{2} & 0 \\
		0 & 0 & \frac{1}{2} & \frac{1}{2} \\
		\frac{1}{2} & 0 & 0 & \frac{1}{2} \\
	\end{bmatrix},
\end{align*}
where $ p_{uv} = \mP\{m_k=v|m_{k-1}=u\} $. Therefore, the stationary distribution $ \pi_u = \frac{1}{4} $ for all $ u = 1,2,3,4 $. 

\begin{figure}[!t]
	\centering
	\subfloat[Graph $ \mG^{(1)} $]{
		\includegraphics[scale=0.5,width=0.4\linewidth]{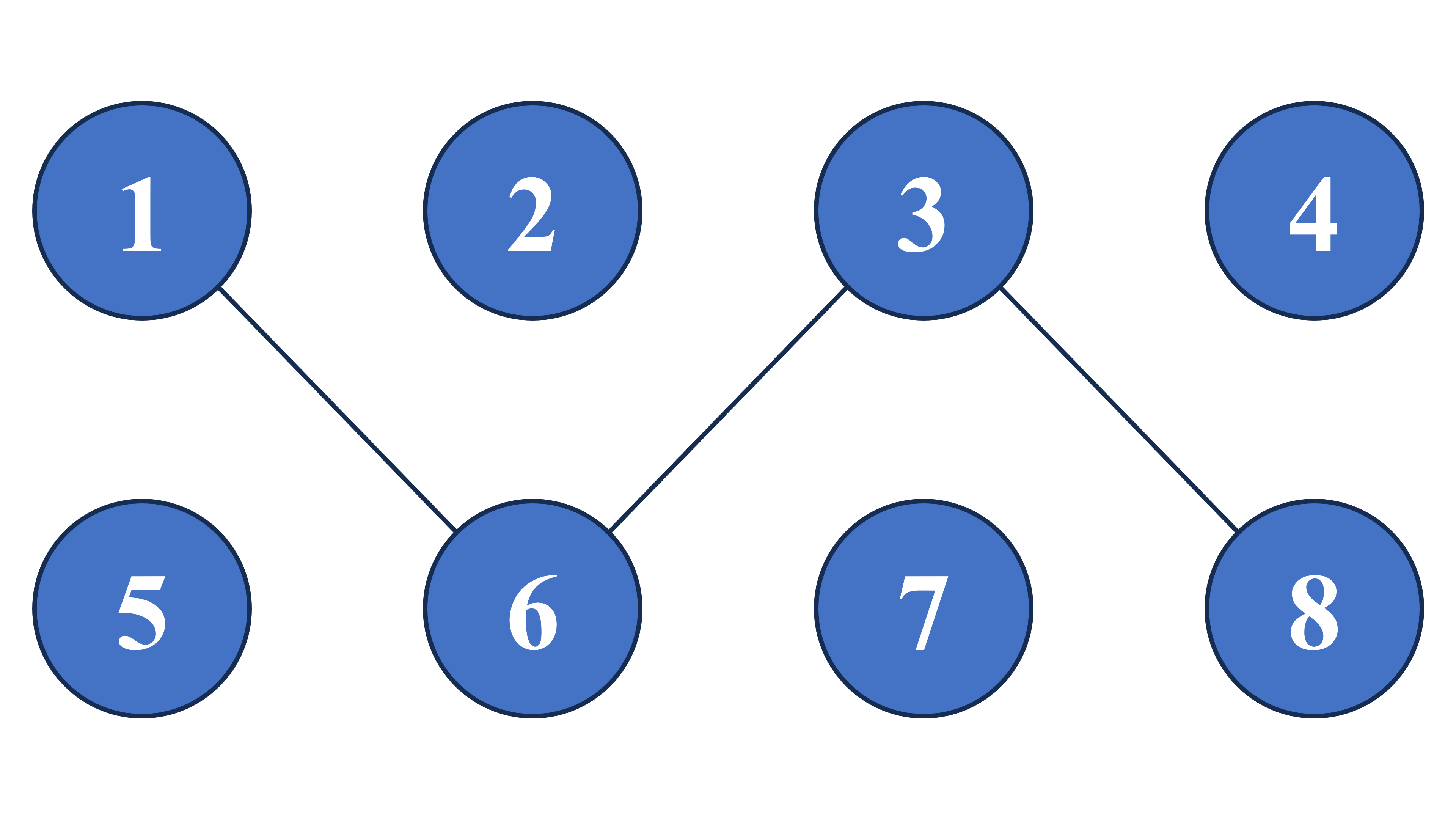}}
	\quad
	\subfloat[Graph $ \mG^{(2)} $]{
		\includegraphics[scale=0.5,width=0.4\linewidth]{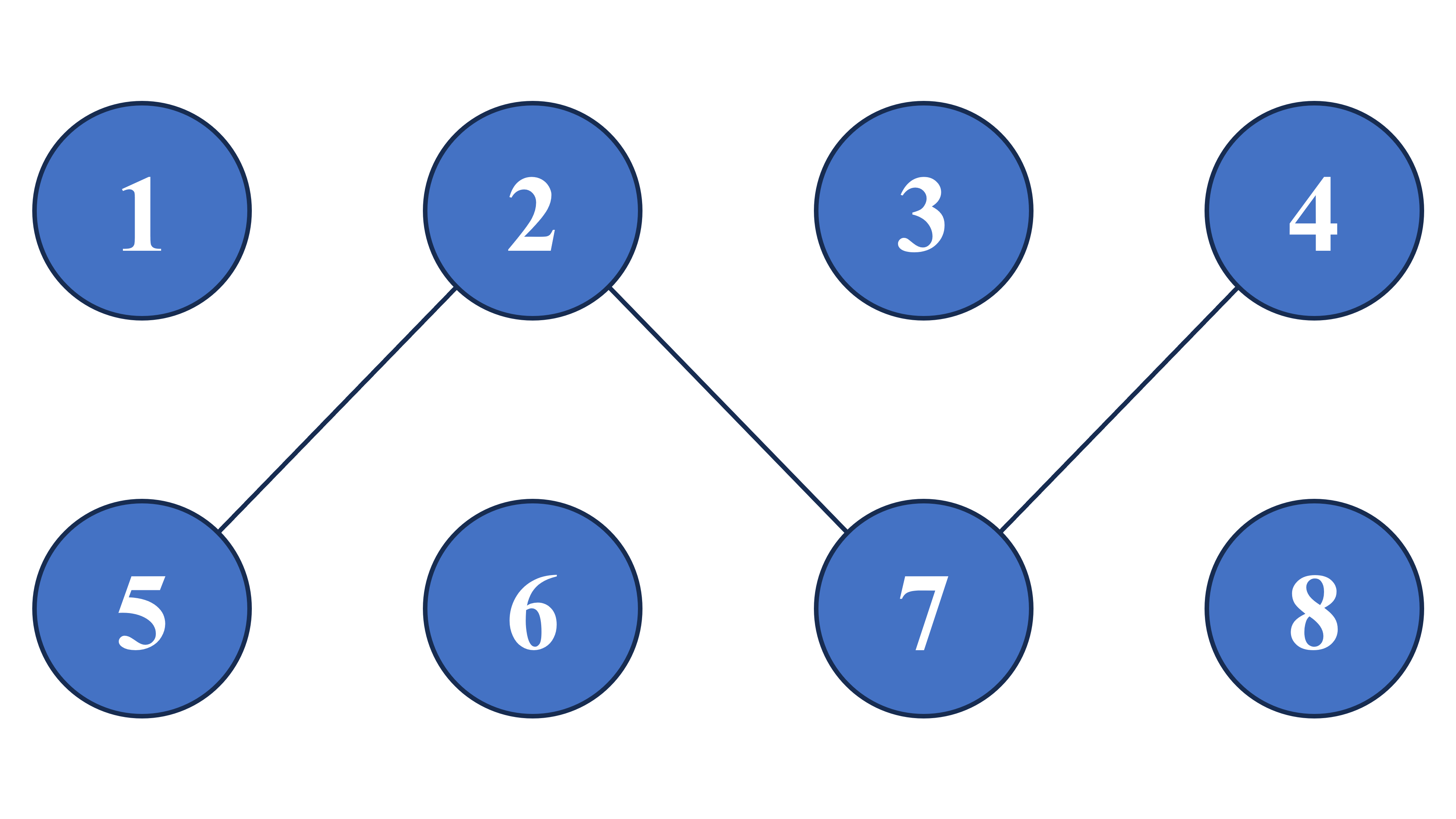}}
	\\
	\subfloat[Graph $ \mG^{(3)} $]{
		\includegraphics[scale=0.5,width=0.4\linewidth]{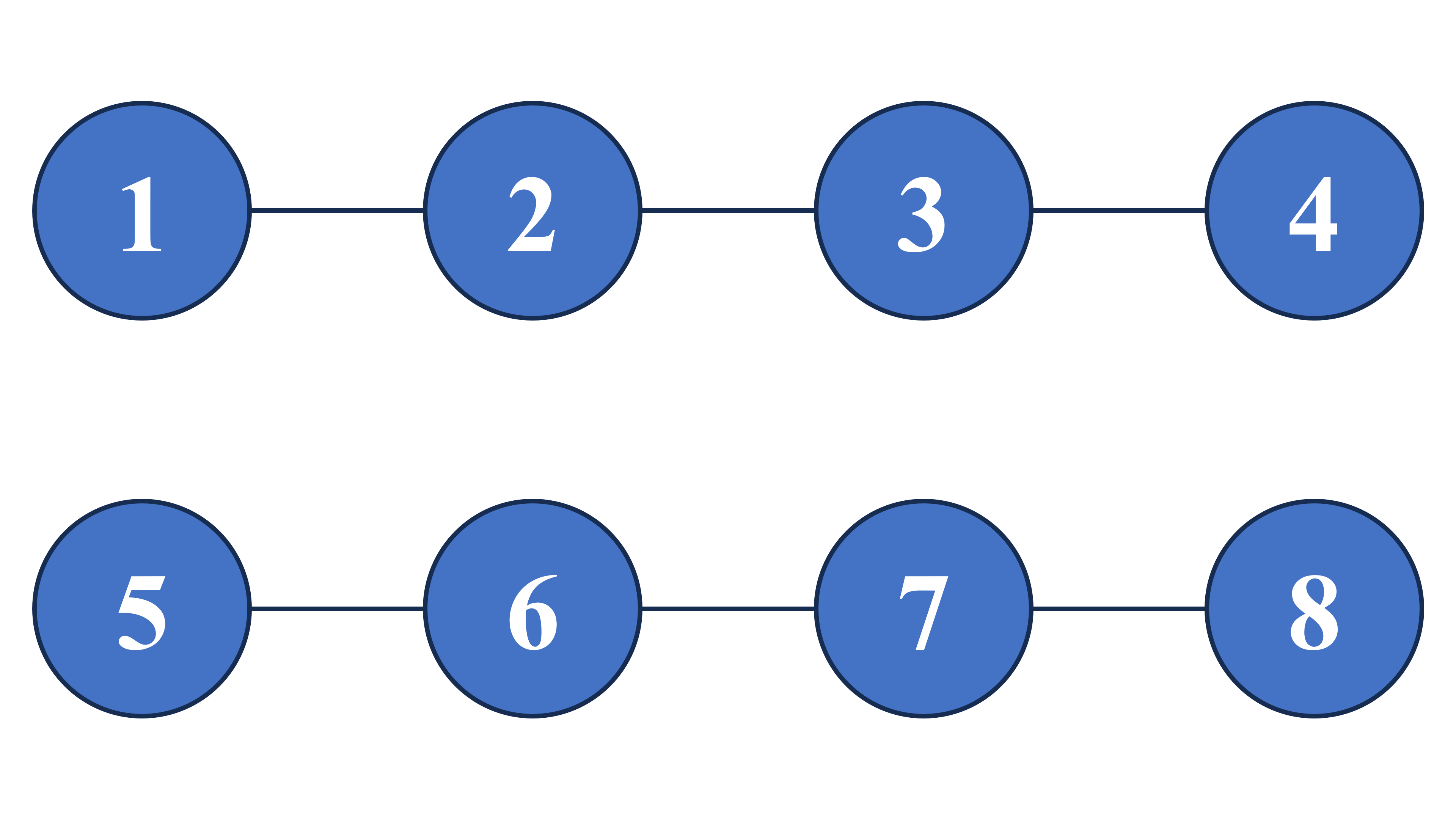}}
	\quad
	\subfloat[Graph $ \mG^{(4)} $]{
		\includegraphics[scale=0.5,width=0.4\linewidth]{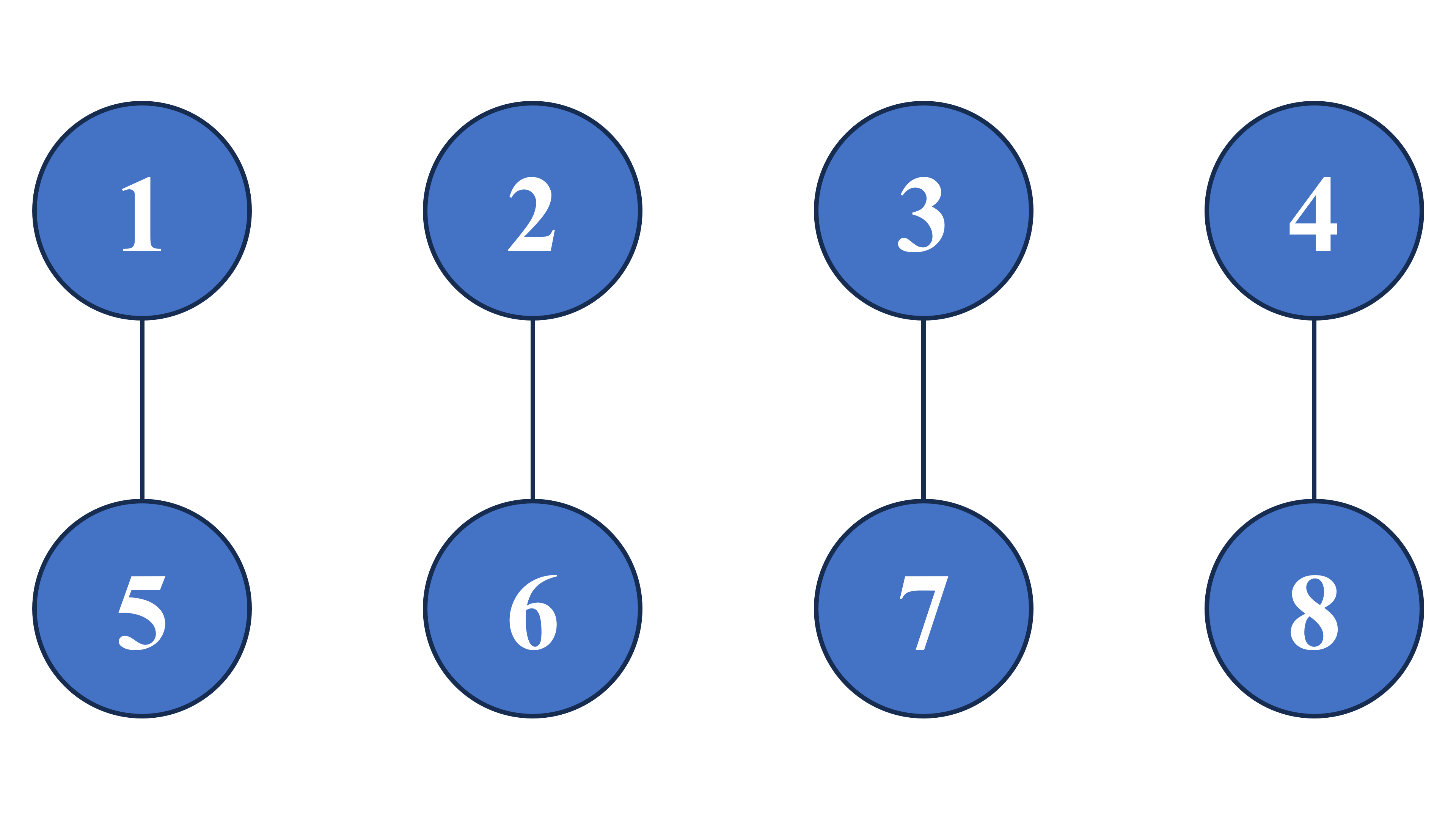}}
	\caption{Communication graphs}
	\label{fig:graph}
\end{figure}

In the observation model, the unknown parameter $ \theta = \begin{bmatrix}
	1 & -1
\end{bmatrix}^\top $. Sensors fail with probability $ \frac{1}{2} $. When the sensor $ i $ does not fail at time $ k $, the measurement matrix $ \mathtt{H}_{i,k} = 
\begin{bmatrix}
	2 & 0
\end{bmatrix} $ if $ i $ is odd, and  $ 
\begin{bmatrix}
	0 & 2
\end{bmatrix} $ if $ i $ is even. 
When the sensor $ i $ fails, $ \mathtt{H}_{i,k} = 0 $. Therefore, $ \bar{H}_i = \begin{bmatrix}
	1 & 0
\end{bmatrix} $ if $ i $ is odd, and  $ 
\begin{bmatrix}
0 & 1
\end{bmatrix} $ if $ i $ is even.
The observation noise $ \mathtt{w}_{i,k} $ is i.i.d. Gaussian with zero mean and standard deviation $ 0.1 $. 

In \cref{algo:DE}, the threshold $ C_{ij} = 0 $. The step-sizes $ \MODIFYY{\alpha_{ij,k}} = \frac{3}{k^{0.8}} $, and $ \beta_{i,k} = \frac{3}{k} $ if $ k \geq 8 $; and $ 0 $, otherwise. Three types of privacy noise distributions are considered, including Gaussian distribution $ \mathcal{N}(0,\sigma_{ij,k}^2) $ with $ \sigma_{ij,k} = k^{0.15} $, Laplacian distribution $ \text{\textit{Lap}}(0,b_{ij,k}) $ with $ b_{ij,k} = k^{0.15} $ and Cauchy distribution $ \text{\textit{Cauchy}}(0,r_{ij,k}) $ with $ r_{ij,k} = k^{0.15} $.

We repeat the simulation 100 times, and \cref{fig:conv} illustrates the trajectories of $ \frac{1}{100N} \sum_{i=1}^{N} \sum_{\varsigma = 1}^{100} \Absl{\tilde{\uptheta}_{i,k}^{\varsigma}}^2 $, where $ \tilde{\uptheta}_{i,k}^\varsigma $ is the estimate of $ \theta $ by sensor $ i $ at time $ k $ in the $ \varsigma $-th run. The figure demonstrates that the estimates can converge the true value $ \theta $ even under increasing noises and 1 communication data rate. In addition, \cref{fig:conv} shows that when sensors do not communicate with each other, the estimates do not converge to the true value. Therefore, the communication is necessary for the distributed estimation.

\begin{figure}[!htbp]
	\centering
	\includegraphics[width=1\linewidth]{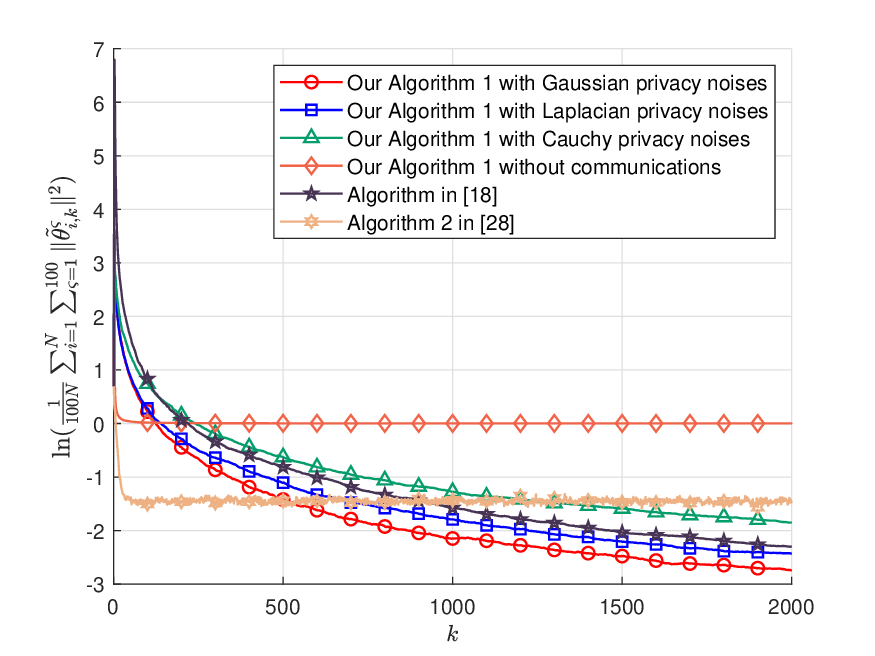}
	\caption{The trajectories of \MODIFYY{ $ \ln\left( \frac{1}{100N} \sum_{i=1}^{N} \sum_{\varsigma = 1}^{100} \Absl{\tilde{\uptheta}_{i,k}^{\varsigma}}^2 \right) $ }}
	\label{fig:conv}
\end{figure}

\cref{fig:priv} draws the upper bounds of the non-zero elements in $ \mE \mathcal{I}_{\mathtt{S}}(\mathtt{y}_{i,k}) $ given by \cref{thm:privacy/finite}. \MODIFY{To avoid duplicate presentation of similar figures, \cref{fig:priv} only takes the sensors 1 and 2 as representative examples.} The figure indicates that the privacy-preserving capability of \cref{algo:DE} is dynamically enhanced under the three types of privacy noise distributions. 

\begin{remark}
	Under our setting, 
	$ \bar{H}_{i} \bar{H}_{i}\tr = \begin{bmatrix}
		1 & 0 \\ 0 & 0
	\end{bmatrix} $ if $ i $ is odd; and $ \bar{H}_{i} \bar{H}_{i}\tr = \begin{bmatrix}
		0 & 0 \\ 0 & 1
	\end{bmatrix} $ if $ i $ is even. Then, by \cref{thm:privacy/finite}, there is only one element in the matrix $ \mathcal{I}_{\mathtt{S}}(\mathtt{y}_{i,k}) $ is non-zero. 
	Therefore, it is sufficient to depict the trajectory of non-zero element in the matrix $ \mE\mathcal{I}_{\mathtt{S}}(\mathtt{y}_{i,k}) $ in \cref{fig:priv}. 
\end{remark}

\begin{figure}[!htbp]
	\centering
	\subfloat[The boundaries of (1,1) element in $ \ln \mathbb{E} \mathcal{I}_{\mathtt{S}}(\mathtt{y}_{1,k}) $]{
		\includegraphics[scale=1,width=1\linewidth]{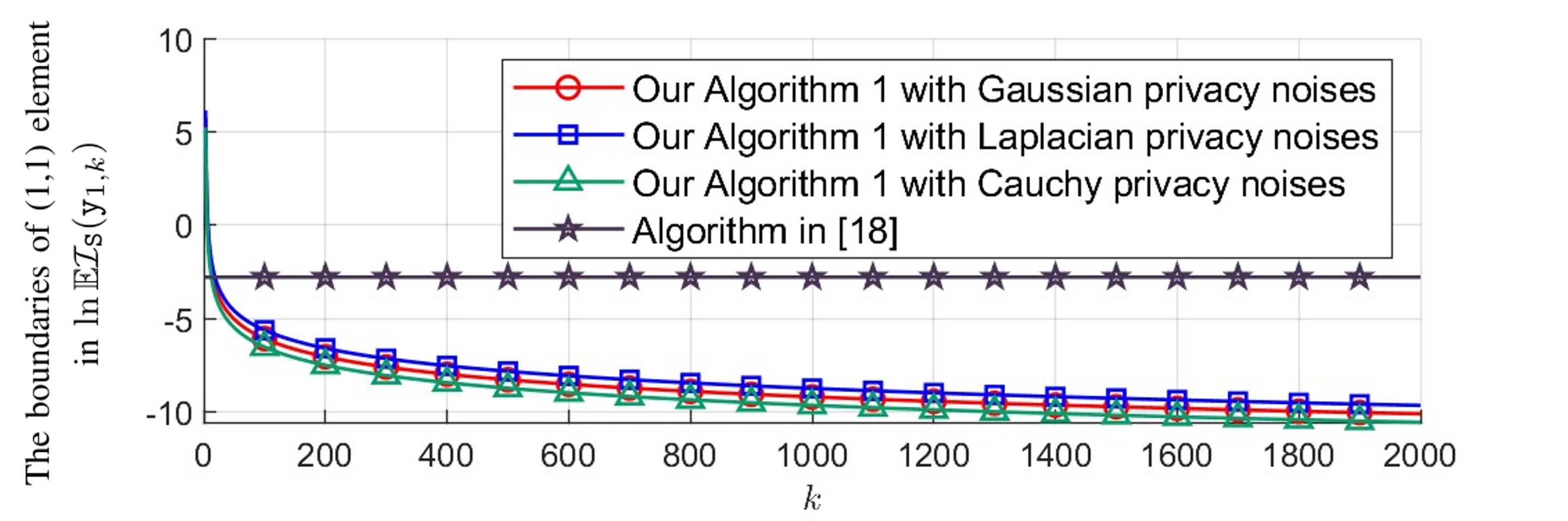}}\\
	\subfloat[The boundaries of (2,2) element in $ \ln \mathbb{E} \mathcal{I}_{\mathtt{S}}(\mathtt{y}_{2,k}) $]{
		\includegraphics[scale=1,width=1\linewidth]{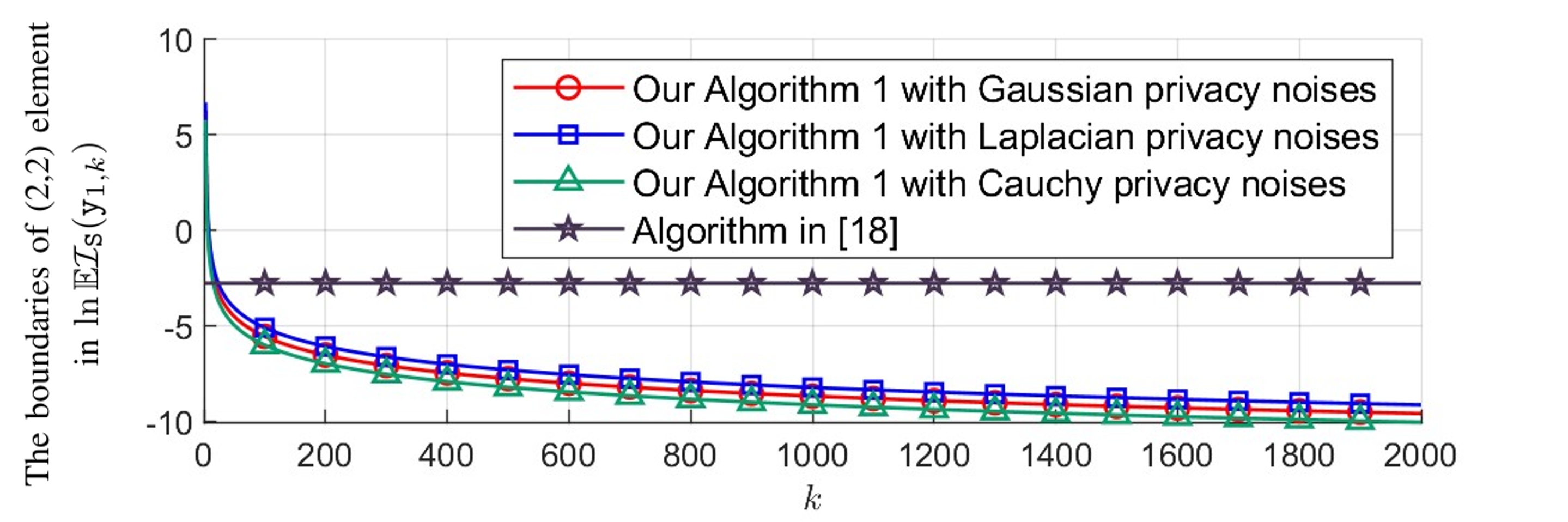}}
	\caption{\MODIFYYY{The upper boundaries of the non-zero elements in $ \mE \mathcal{I}_{\mathtt{S}}(\mathtt{y}_{i,k}) $ for the sensors $ 1 $ and $ 2 $}}
	\label{fig:priv}
\end{figure}

\MODIFY{\cref{fig:conv,fig:priv} also compare \cref{algo:DE} with existing ones in \cite{WangJM2022IJRNC,Sayin2014}. From \cref{fig:conv,fig:priv}, one can get that \cref{algo:DE} can achieve similar estimation error and much better privacy simultaneously compared with the algorithm in \cite{WangJM2022IJRNC}. 
	Besides, the algorithm in \cite{WangJM2022IJRNC} requires sensors to transmit real-valued information to each other, in contrast to the binary-valued communications of our \cref{algo:DE}. Algorithm 2 in \cite{Sayin2014} also requires binary-valued communications. The mean square errors of its estimates quickly decrease to a certain value, but do not converge to 0. Therefore, after about 1000 iterations, the estimation error of our \cref{algo:DE} is smaller than that of Algorithm 2 in \cite{Sayin2014}. Besides, \cite{Sayin2014} does not consider the privacy-preserving issue. 
}

\cref{fig:trade} demonstrates the trade-off between privacy and convergence rate for \cref{algo:DE}. In \cref{algo:DE}, the step-size $ \alpha_{ij,k} = \frac{3}{k^{(2.9-\chi)/{2}}} $, and the privacy noises is Cauchy distributed with $ r_{ij,k} = k^{\frac{\chi-1}{2}} $, where $ \chi = 1.3,\ 1.6 $ and $ 1.9 $. \cref{fig:trade} (a) depicts the log-log plot for the boundaries of $ \mE \mathcal{I}_{\mathtt{S}}(\mathtt{y}_{1,k}) $. It is observed that a better privacy level is achieved with a larger $ \chi $. \cref{fig:trade} (b) shows the log-log plot for the  trajectories of $ \frac{1}{100N} \sum_{i=1}^{N} \sum_{\varsigma = 1}^{100} \Absl{\tilde{\uptheta}_{i,k}^{\varsigma}}^2 $. It is observed that a better convergence rate is achieved with a smaller $ \chi $. \MODIFY{Therefore, the trade-off can be shown under different $ \chi $. }

\begin{figure}[!htbp]
	\centering
	\subfloat[\MODIFYYY{The boundaries of (1,1) element in $ \ln \mE \mathcal{I}_{\mathtt{S}}(\mathtt{y}_{1,k}) $ with different $ \chi $}]{
		\includegraphics[scale=1,width=1\linewidth]{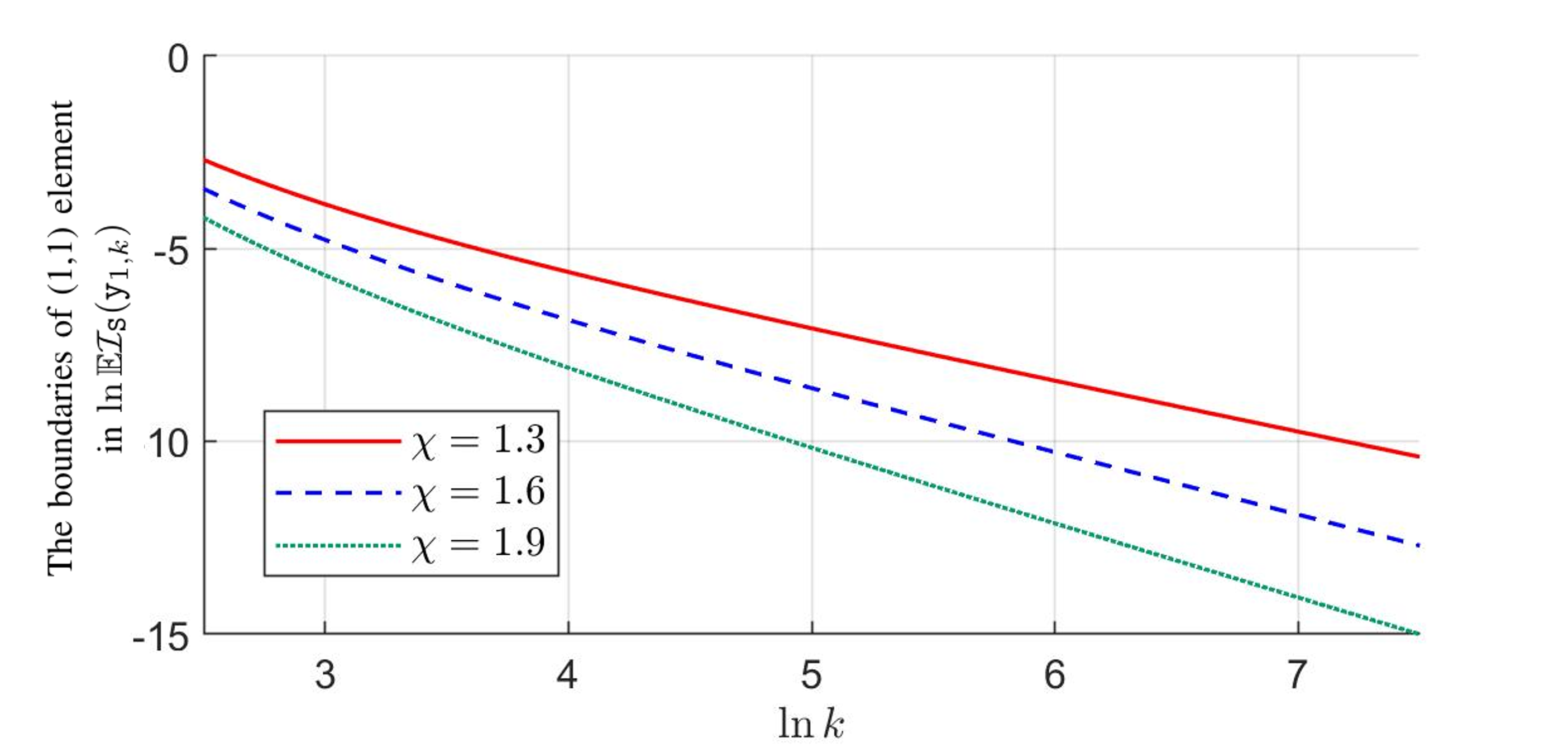}}\
	\subfloat[The log-log plot for $ \frac{1}{100N} \sum_{i=1}^{N} \sum_{\varsigma = 1}^{100} \Absl{\tilde{\uptheta}_{i,k}^{\varsigma}}^2 $ with different $ \chi $]{
		\includegraphics[scale=1,width=1\linewidth]{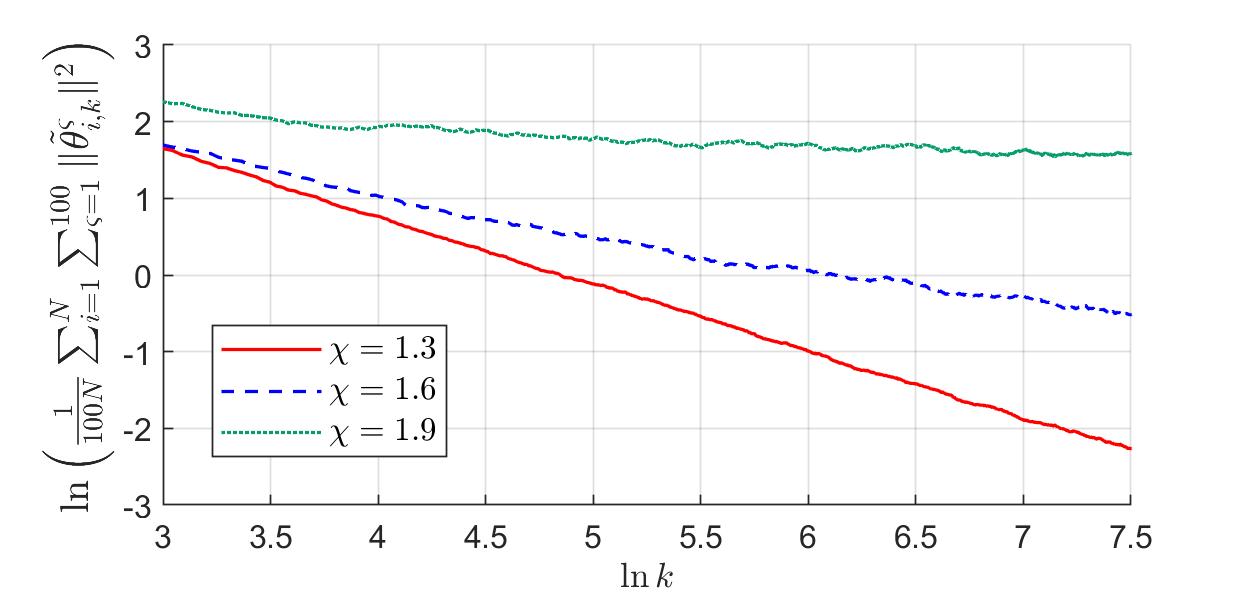}}
	\caption{\MODIFY{The trade-off between privacy and convergence rate}}
	\label{fig:trade}
\end{figure}

By \cref{remark:C_phi}, when not pursuing 1 bit communication data rate,  
$ \varphi_k $ in \cref{algo:DE} can be \MODIFYYY{replaced by a higher dimensional selection matrix} to make the modified algorithm perform better in high-dimensional settings. To show this improvement, consider the case of $ n = 12 $. The unknown parameter $\theta$ is uniformly generated within $ [-1,1]^{12} $. $ \bar{H}_i $ is expanded to $ \begin{bmatrix}
	I_6 & O_6
\end{bmatrix} $ if $ i $ is odd, and  $ 
\begin{bmatrix}
	O_6 & I_6
\end{bmatrix} $ if $ i $ is even, which ensures \cref{assum:input} in the high dimensional $\theta$ case. Under this settings, from \cref{fig:multiplebitscompare}, one can see that the modified algorithm converges faster than the original \cref{algo:DE}. Additionally, since the influence of $ \varphi_k $ was neglected in the analysis of \cref{thm:privacy/finite}, the upper bound of privacy-preserving level obtained from Theorem 1 still holds after removing $ \varphi_k $, which implies that the modified algorithm still has strong privacy-preserving capability.

\begin{figure}[!htbp]
	\centering
	\includegraphics[width=1\linewidth]{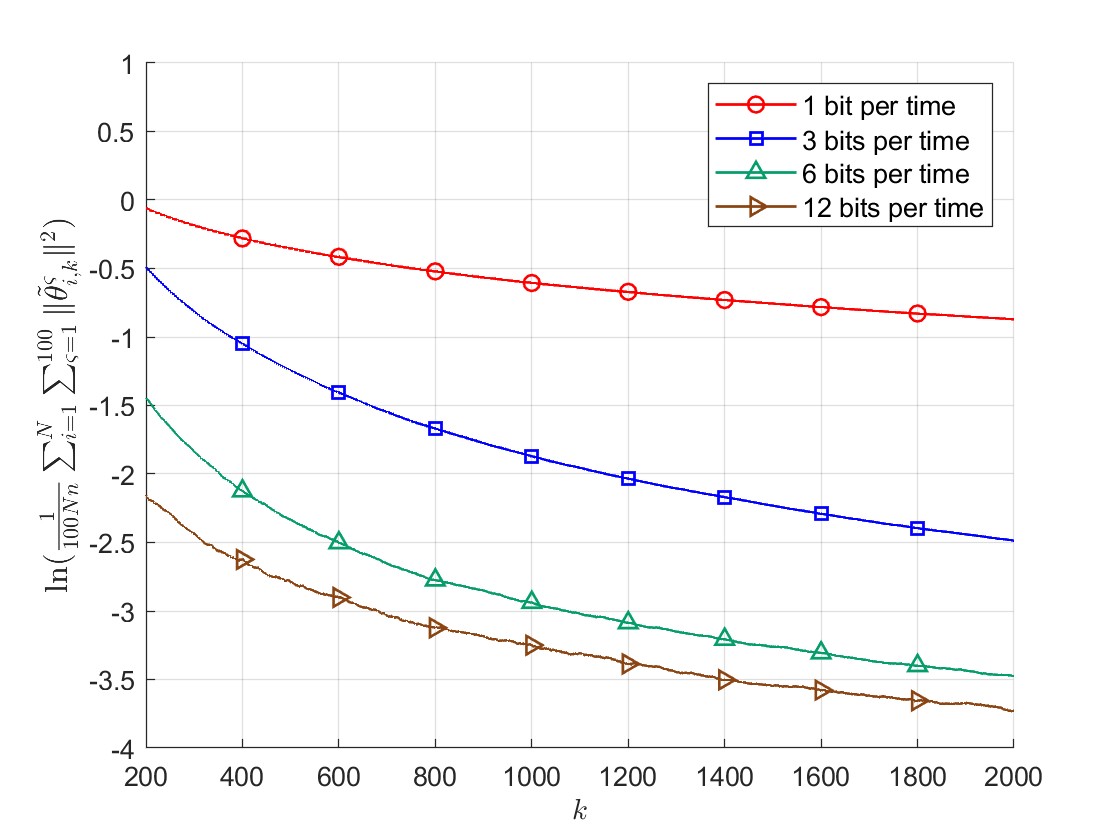}
	\caption{\MODIFYYY{The trajectories of $ \ln\left(\frac{1}{100Nn} \sum_{i=1}^{N} \sum_{\varsigma = 1}^{100} \Absl{\tilde{\uptheta}_{i,k}^{\varsigma}}^2 \right)$ with different data rate}}
	\label{fig:multiplebitscompare}
\end{figure}

%

\subsection{An experiment on the event rate analysis of essential hypertension}

In this subsection, \cref{algo:DE} is applied in the event rate analysis of essential hypertension in 281299 white British participants\footnote{\MODIFY{The data comes from UK Biobank (Application: 78793).}}. 
In the experiment, $ \mathtt{H}_{i,k} = 1 $ if there is a participant for the sensor $ i $ at time $ k $; and $ \mathtt{H}_{i,k} = 0 $, otherwise. $ \mathtt{H}_{i,k} = 1 $ with probability $ 0.7 $. 
The observation $ \mathtt{y}_{i,k} = 1 $  if $ \mathtt{H}_{i,k} = 1 $ and the participant suffers from the essential hypertension; and $ \mathtt{y}_{i,k}  = 0 $, otherwise. Such clinical information $ \mathtt{y}_{i,k} $ is private, and needs to be protected in practical scenarios. 

About 4/5 of the database is used as the training set, while the rest is the test set. From the test set, we have the event rate $ \theta \approx 0.2699 $. 
Data in the training set is distributed in a 20 sensor network. In the network, $ \mathtt{a}_{ij,k} = 1 $ if $ (i,j)\in\mathtt{E}_k $; and $ 0 $, otherwise. The initial probability $ \mP\{ \mathtt{a}_{ij,1} = 1 \} = 0.5 $, and the transition probability $ \mP\{ \mathtt{a}_{ij,k} = 1 | \mathtt{a}_{ij,k-1} = 1 \} = \mP\{ \mathtt{a}_{ij,k} = 0 | \mathtt{a}_{ij,k-1} = 0 \} = 0.7 $. 

In \cref{algo:DE}, the threshold $ C_{ij} = 0 $. The step-sizes $ \MODIFYY{\alpha_{ij,k} = \frac{0.2}{k^{(2.9-\chi)/{2}}}} $, $ \beta_{i,k} = \frac{0.4}{k} $, and the privacy noise is Gaussian $ \mathcal{N}(0,\sigma_{ij,k}^2) $ with $ \MODIFYY{\sigma_{ij,k} = k^{\frac{\chi-1}{2}}} $, \MODIFYY{where $ \chi = 1.3,\ 1.6 $ and $ 1.9 $. } Under the settings, \cref{fig:real} (a) shows the dynamically enhanced privacy of our algorithm, and \cref{fig:real} (b) demonstrates the convergence.

\begin{figure}[!htbp]
	\centering
	\subfloat[\MODIFYYY{The boundary of $ \ln \mathbb{E} \mathcal{I}_{\mathtt{S}}(\mathtt{y}_{1,k}) $}]{
		\includegraphics[scale=0.4,width=0.48\linewidth]{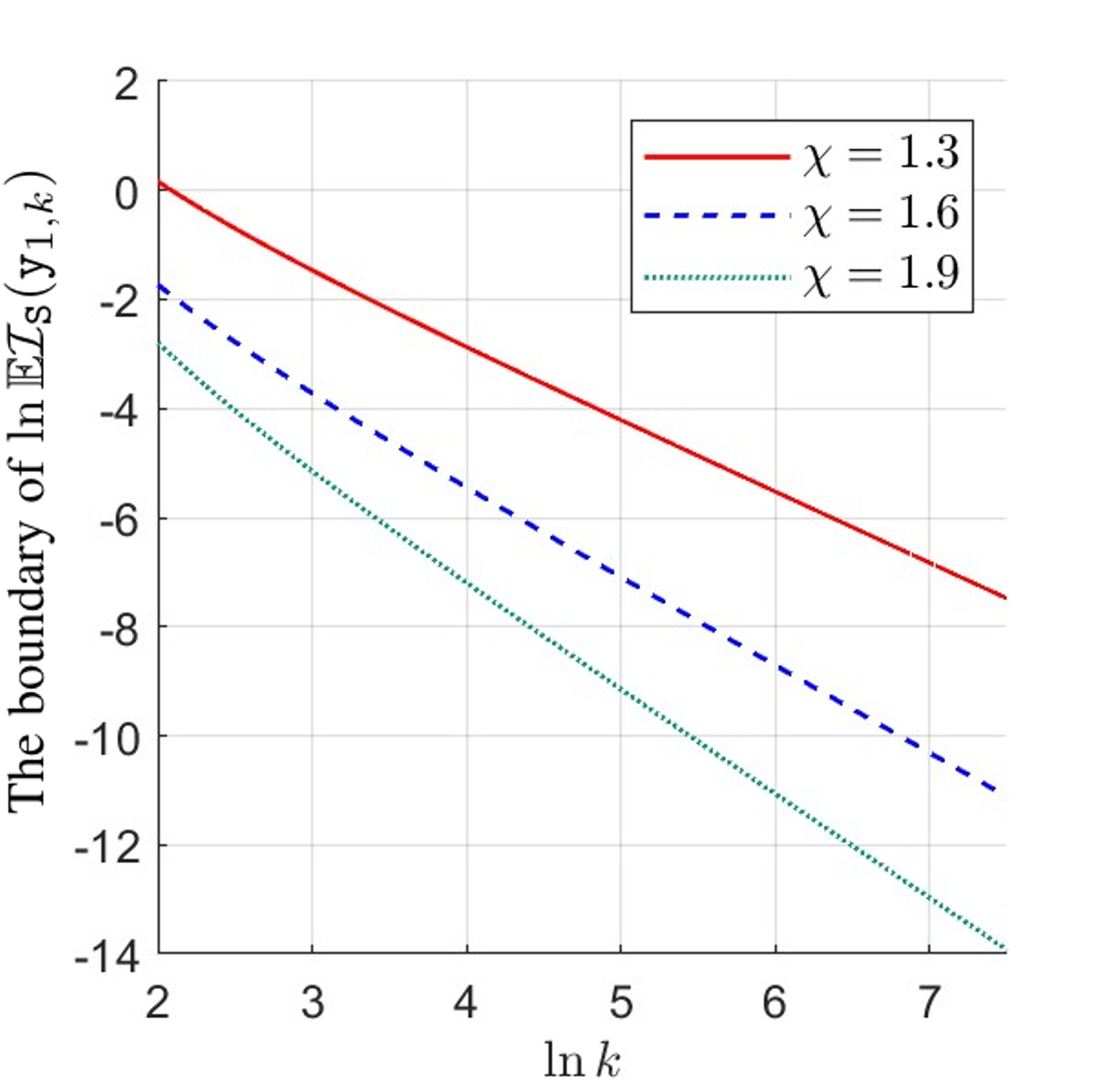}}
	\subfloat[Estimation error]{
		\includegraphics[scale=0.4,width=0.48\linewidth]{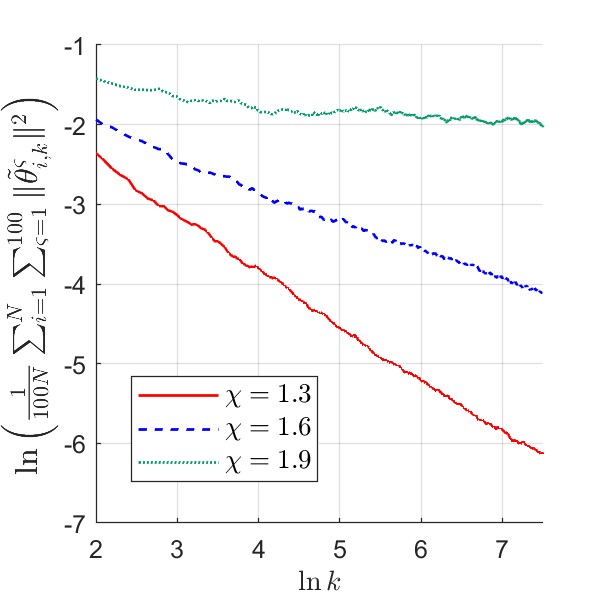}}
	\caption{\MODIFY{Privacy and convergence of \cref{algo:DE} for the event rate analysis of essential hypertension}}
	\label{fig:real}
\end{figure}

%

\section{Conclusion}\label{sec:concl}

This paper proposes a binary-valued quantizer-based privacy-preserving distributed estimation algorithm with multiple advantages. 
In terms of privacy, the proposed algorithm achieves the dynamically enhanced privacy, and the Fisher information-based privacy metric $ \mE \mathcal{I}_{\mathtt{S}}(\mathtt{y}_{i,k}) $ is proved to converge to $ 0 $ at a polynomial rate. 
In terms of communication costs, each sensor transmits only 1 bit of information to its neighbours at each time step. Besides, \MODIFY{the assumption on the negligible quantization error for real-valued messages is not required.} 
In terms of effectiveness, the proposed algorithm can achieve almost sure convergence even with increasing privacy noises. A polynomial convergence rate is also obtained. 
Besides, the trade-off between privacy and convergence rate is established. When the step-sizes and privacy noise distributions are properly selected, a better privacy-preserving capability implies a slower convergence rate, and vice versa. 

There are still many interesting topics worth further investigation.  For example, how to apply the proposed method to distributed optimization problems to achieve the dynamically enhanced privacy and a limited data rate, and how to \MODIFY{protect the observation matrices}. 

\useRomanappendicesfalse

\appendices

\section{Lemmas and corollaries}\label[appen]{appen}

\setcounter{equation}{0}
\renewcommand{\theequation}{A.\arabic{equation}}

\begin{lemmax}\label{lemma:enhance}
	If $ \lim_{k\to\infty} \mE \mathcal{I}_{\mathtt{S}} (\mathtt{y}_{i,k}) = 0 $, then the privacy-preserving capability is dynamically enhanced.
\end{lemmax}

\begin{proof}
	Since $ \lim_{k\to\infty} \mE \mathcal{I}_{\mathtt{S}} (\mathtt{y}_{i,k}) = 0 $, for any $ A > 0 $, there exists $ T \in \mathbb{N} $ such that $ \mE \mathcal{I}_{\mathtt{S}} (\mathtt{y}_{i,t}) \leq A $ for all $ t \geq T $. Then, the lemma can be proved by setting $ A = \mE \mathcal{I}_{\mathtt{S}} (\mathtt{y}_{i,k}) $. 
\end{proof}

\begin{lemmax}[Chain rule for Fisher information, \cite{zamir1998Fisher}]\label{lemma:chain}
	For random variables $ \mathtt{X}, \mathtt{Y}, \uptheta $, $ \mathcal{I}_{\mathtt{X},\mathtt{Y}}(\uptheta) = \mathcal{I}_{\mathtt{X}}(\uptheta) + \mathcal{I}_{\mathtt{Y}}(\uptheta|\mathtt{X}) \geq \mathcal{I}_{\mathtt{X}}(\uptheta) $. 
\end{lemmax}


\begin{corollaryx}\label{coro:chain}
	For random variables $ \mathtt{X}, \mathtt{Y}, \mathtt{Z}, \uptheta $, we have
	\begin{enumerate}[label={\alph*)},leftmargin=1.4em]
		\item $ \mathcal{I}_{\mathtt{X},\mathtt{Y}}(\uptheta|\mathtt{Z}) = \mathcal{I}_{\mathtt{X}}(\uptheta|\mathtt{Z}) + \mathcal{I}_{\mathtt{Y}}(\uptheta|\mathtt{X},\mathtt{Z}) $;
		\item If $ \mathcal{I}_{\mathtt{Y}}(\uptheta|\mathtt{X}) = 0 $, then $ \mathcal{I}_{\mathtt{X}}(\uptheta|\mathtt{Y}) \leq \mathcal{I}_{\mathtt{X},\mathtt{Y}}(\uptheta) = \mathcal{I}_{\mathtt{X}}(\uptheta) $;
		\item If $ \mathcal{I}_{\mathtt{X}}(\uptheta|\mathtt{Z}) = 0 $, then $ \mathcal{I}_{\mathtt{Y}}(\uptheta|\mathtt{Z}) \leq \mathcal{I}_{\mathtt{Y}}(\uptheta|\mathtt{X},\mathtt{Z}) $.
	\end{enumerate} 
\end{corollaryx}

\begin{proof}
	a) By \cref{lemma:chain}, we have
	\begin{align*}
		&\mathcal{I}_{\mathtt{X},\mathtt{Y}}(\uptheta|\mathtt{Z})
		=  \mathcal{I}_{\mathtt{X},\mathtt{Y},\mathtt{Z}}(\uptheta) - \mathcal{I}_{\mathtt{Z}}(\uptheta) \\
		= & \mathcal{I}_{\mathtt{X}}(\uptheta|\mathtt{Y},\mathtt{Z}) + \mathcal{I}_{\mathtt{X},\mathtt{Z}}(\uptheta|\mathtt{Y}) - \mathcal{I}_{\mathtt{Z}}(\uptheta) 
		= \mathcal{I}_{\mathtt{X}}(\uptheta|\mathtt{Z}) + \mathcal{I}_{\mathtt{Y}}(\uptheta|\mathtt{X},\mathtt{Z}). 
	\end{align*}
	
	b) By \cref{lemma:chain}, we have
	\begin{align*}
		\mathcal{I}_{\mathtt{X}}(\uptheta|\mathtt{Y}) 
		\leq \mathcal{I}_{\mathtt{X},\mathtt{Y}}(\uptheta)
		=  \mathcal{I}_{\mathtt{X}}(\uptheta) + \mathcal{I}_{\mathtt{Y}}(\uptheta|\mathtt{X}) 
		= \mathcal{I}_{\mathtt{X}}(\uptheta).
	\end{align*}

	c) By a), we have
	\begin{align*}
		\mathcal{I}_{\mathtt{Y}}(\uptheta|\mathtt{Z})
		= & \mathcal{I}_{\mathtt{X},\mathtt{Y}}(\uptheta|\mathtt{Z}) - \mathcal{I}_{\mathtt{X}}(\uptheta|\mathtt{Y},\mathtt{Z}) 
		\leq \mathcal{I}_{\mathtt{X},\mathtt{Y}}(\uptheta|\mathtt{Z}) \\
		= & \mathcal{I}_{\mathtt{X}}(\uptheta|\mathtt{Z}) + \mathcal{I}_{\mathtt{Y}}(\uptheta|\mathtt{X},\mathtt{Z}) 
		= \mathcal{I}_{\mathtt{Y}}(\uptheta|\mathtt{X},\mathtt{Z}). \qedhere
	\end{align*}

\end{proof}

\begin{lemmax}\label{lemma:Fisher_independent}
	For random variables $ \mathtt{X}, \uptheta $, and random variable sequences $ \mathtt{Y}_k = \{\mathtt{Y}_{i,k} : i = 1,\ldots,N\}, \mathtt{Z}_k = \{\mathtt{Z}_{i,k} : i = 1,\ldots,N\} $ for all $ k \in \mathbb{N} $, if
	\begin{enumerate}[label={\roman*)},leftmargin=1.4em]
		\item $ \mathtt{Y}_{1,k}, \ldots, \mathtt{Y}_{N,k} \neq 0 $, $ \mathtt{Z}_{1,k}, \ldots, \mathtt{Z}_{N,k} \in \{0,1\} $;
		\item Given $ \uptheta $, $ \mathtt{X} $ and $ \breve{\mathtt{Z}}_{k-1} $, the sequence $ \mathtt{Y}_k $ is independent, and independent of $ \mathtt{Z}_k $, where $ \breve{\mathtt{Z}}_k = \bigcup_{t=1}^{k} \hat{\mathtt{Z}}_t $ and $ \hat{\mathtt{Z}}_k = \{ \mathtt{Z}_{i,k} \mathtt{Y}_{i,k}, i = 1,\ldots, N\} $; 
		\item $ \mathcal{I}_{\mathtt{Z}_k}(\uptheta|\mathtt{X},\breve{\mathtt{Z}}_{k-1}) = 0 $, 
	\end{enumerate}
	then $ \mathcal{I}_{\breve{\mathtt{Z}}_{\infty}} (\uptheta|\mathtt{X})
	= \sum_{k=1}^\infty \sum_{i=1}^{N} \mathcal{I}_{\mathtt{Z}_{i,k} \mathtt{Y}_{i,k}} (\uptheta|\mathtt{X},\breve{\mathtt{Z}}_{k-1}) $.
\end{lemmax}

\begin{proof}
	Note that by i), we have $ \mathtt{Z}_{i,k} $ can be uniquely determined by $ \mathtt{Z}_{i,k} \mathtt{Y}_{i,k} $. Then, by ii), given $ \uptheta $, $ \mathtt{X} $, $ \breve{\mathtt{Z}}_{k-1} $ and $ \mathtt{Z}_k $, we have $ \hat{\mathtt{Z}}_k $ is independent. Hence, by \cref{coro:chain}, 
	\begin{align}\label{eq:I_k=sum}
		\mathcal{I}_{\breve{\mathtt{Z}}_{\infty}} (\uptheta|\mathtt{X})
		= & \sum_{k=1}^\infty\mathcal{I}_{\hat{\mathtt{Z}}_{k}} (\uptheta|\mathtt{X}, \breve{\mathtt{Z}}_{k-1}) 
		= \sum_{k=1}^\infty\mathcal{I}_{\hat{\mathtt{Z}}_{k},\mathtt{Z}_k} (\uptheta|\mathtt{X}, \breve{\mathtt{Z}}_{k-1}) \nonumber\\
		= & \sum_{k=1}^\infty\mathcal{I}_{\hat{\mathtt{Z}}_{k}} (\uptheta|\mathtt{X}, \breve{\mathtt{Z}}_{k-1},\mathtt{Z}_k) + \mathcal{I}_{\mathtt{Z}_k}(\uptheta|\mathtt{X},\breve{\mathtt{Z}}_{k-1}) \nonumber\\
		= & \sum_{k=1}^\infty\sum_{i = 1}^{N} \mathcal{I}_{\mathtt{Z}_{i,k} \mathtt{Y}_{i,k}} (\uptheta|\mathtt{X},\breve{\mathtt{Z}}_{k-1},\mathtt{Z}_k). 
	\end{align}
	By iii), given $ \uptheta $, $ \mathtt{X} $, $ \breve{\mathtt{Z}}_{k-1} $ and $ \mathtt{Z}_{i,k} $, we have $ \mathtt{Y}_{i,k} $ is independent of $ \mathtt{Z}_{j,k} $ for all $ j \neq i $. Therefore, by \cref{coro:chain}, 
	\begin{align*}
		& \mathcal{I}_{\mathtt{Z}_{i,k} \mathtt{Y}_{i,k}} (\uptheta|\mathtt{X},\breve{\mathtt{Z}}_{k-1},\mathtt{Z}_k)
		= \mathcal{I}_{\mathtt{Z}_{i,k} \mathtt{Y}_{i,k}} (\uptheta|\mathtt{X},\breve{\mathtt{Z}}_{k-1},\mathtt{Z}_{i,k}) \\
		= & \mathcal{I}_{\mathtt{Z}_{i,k} \mathtt{Y}_{i,k},\mathtt{Z}_{i,k}} (\uptheta|\mathtt{X},\breve{\mathtt{Z}}_{k-1}) - \mathcal{I}_{\mathtt{Z}_{i,k}} (\uptheta|\mathtt{X},\breve{\mathtt{Z}}_{k-1}) \\
		= & \mathcal{I}_{\mathtt{Z}_{i,k} \mathtt{Y}_{i,k}} (\uptheta|\mathtt{X},\breve{\mathtt{Z}}_{k-1}),
	\end{align*}
	which together with \eqref{eq:I_k=sum} implies the lemma. 
\end{proof}

\begin{lemmax}\label{lemma:+}
	For a matrix $ H $, set $ Q = H^\top H $ and $ J = Q^+ Q $. Then, $ HJ=H $. 
\end{lemmax}

\begin{proof}
	By Theorem 1 of \cite{Greville1966generailized}, 
	\begin{align*}
		HJ
		= & (H^\top)^+ H^\top H Q^+ Q
		= (H^\top)^+ Q Q^+ Q \\
		= & (H^\top)^+ Q 
		= (H^\top)^+ H^\top H
		= H. \qedhere
	\end{align*}
%
\end{proof}

\begin{lemmax}\label{lemma:+eigen}
	For a positive semi-definite matrix $ Q $, set $ J = Q^+ Q $. Then, $ \lambda_{\max}(J-\beta Q) = 1 - \beta \lambda_{\min}^+(Q) $, where $ \beta \in \left[0,\frac{1}{\lambda_{\min}^+(Q)}\right] $, and $ \lambda_{\max}(\cdot) $, $ \lambda_{\min}^+(\cdot) $ are defined in \cref{thm:privacy/finite}. 
\end{lemmax}

\begin{proof}
	By Theorem 5 of \cite{Scroggs1966alternate}, all the eigenvectors $ v $ for $ Q $ are eigenvectors for $ J-\beta Q $. If $ Q v = 0 $, then $ (J - \beta Q) v = 0 $. If $ Q v = \lambda v $ for some $ \lambda > 0 $, then $ (J - \beta Q) v = (1-\beta \lambda) v $. The lemma is thereby proved. 
\end{proof}

\begin{lemmax}\label{lemma:sum_prod}
	If sequences $ \{a_k:k\in\mathbb{N}\} $, $ \{b_k:k\in\mathbb{N}\} $ and $ \{\eta_k:k\in\mathbb{N}\} $ satisfy 
	\begin{enumerate}[label={\roman*)},leftmargin=1.4em]
		\item $ a_k \in [0,\bar{a}]$ for some $ \bar{a} < 1 $, and $ \eta_k > 0 $;
		\item $ \sum_{t=1}^{\infty}\prod_{l=1}^{t} \eta_t (1-a_l)^p < \infty $ for some positive integer $ p $;
		\item $ b_k > 0 $ and $ \sum_{k=1}^{\infty} b_k < \infty $,
	\end{enumerate}
	then $ \sum_{t=k}^{\infty} \prod_{l=k}^{t} \eta_t (1-a_l+b_l)^p < \infty $. 
\end{lemmax}

\begin{proof}
	Firstly, we have 
	\begin{align*}
		 \sum_{t=k}^{\infty}\prod_{l=k}^{t} \eta_t (1-a_l)^p
		= & \frac{\sum_{t=k}^{\infty}\prod_{l=1}^{t} \eta_t (1-a_l)^p}{\prod_{l=1}^{k-1} (1-a_l)^p}\\
		\leq & \frac{\sum_{t=1}^{\infty}\prod_{l=1}^{t}\eta_t (1-a_l)^p}{\prod_{l=1}^{k-1}(1-a_l)^p} < \infty. 
	\end{align*}
	
	Then, one can get
	\begin{align*}
		& \sum_{t=k}^{\infty} \prod_{l=k}^{t}\eta_t (1-a_l+b_l)^p \\
		\leq & \sum_{t=k}^{\infty}\prod_{l=k}^{t} \eta_t (1-a_l)^p\left( 1+\frac{b_l}{1-\bar{a}} \right)^p \\
		\leq & \left(\sum_{t=k}^{\infty}\prod_{l=k}^{t} \eta_t (1-a_l)^p\right)\left( \prod_{t=1}^{\infty} \left( 1+\frac{b_t}{1-\bar{a}} \right) \right)^p
		< \infty. \qedhere
	\end{align*}
\end{proof}

\begin{lemmax}\label{lemma:sum_exp}
	If $ c, k_0 > 0 $, $ g \geq 0 $ and $ p \in (0,1] $ satisfy $ cpk^{p}_{0} \geq 1-p-g $, then
	\begin{align*}
		& \sum_{t=1}^{k} \frac{\exp\left( -c(t+k_0)^p \right)}{(t+k_0)^g} \\ 
		\leq & \frac{k_0^{1-p-g} \exp(-ck_0^p) - (k+k_0)^{1-p-g} \exp(-c(k+k_0)^p )}{cp-(1-p-g)k_0^{-p}}. 
	\end{align*}
\end{lemmax}

\begin{proof}
	From the condition of the lemma, we have
	\begin{align*}
		& \frac{\sum_{t=1}^{k} \exp\left( -c(t+k_0)^p \right)}{(t+k_0)^g}
		\leq \int_{k_0}^{k+k_0} \frac{\exp\left( -ct^p \right)}{t^g} \text{d}t \\
		\leq & \int_{k_0}^{k+k_0} \frac{cp - (1-p-g)t^{-p}}{cp - (1-p-g)k_0^{-p}} \frac{\exp\left( -ct^p \right)}{t^g} \text{d}t \\
		= & \frac{\int_{k_0}^{k+k_0} \! cp t^{-g} \exp\left( -c t^p \right) - (1-p-g) t^{-p-g} \exp\left( -c t^p \right) \text{d}t}{cp-(1-p-g)k_0^{-p}} \\
		= & \frac{k_0^{1-p-g} \exp(-ck_0^p) - (k+k_0)^{1-p-g} \exp(-c(k+k_0)^p )}{cp-(1-p-g)k_0^{-p}}. 
	\end{align*}
	The lemma is thereby proved. 
\end{proof}

\begin{lemmax}\label{lemma:conv}
	Assume that
	\begin{enumerate}[label={\roman*)},leftmargin=1.4em]
		\item $ \{\alpha_k:k\in\mathbb{N}\} $, $ \{\beta_k:k\in\mathbb{N}\} $ and $ \{\gamma_k:k\in\mathbb{N}\} $ are posi- tive sequences  satisfying $ \sum_{k=1}^{\infty} \alpha_k = \infty $, $ \sum_{k=1}^{\infty} \beta_k^2 < \infty $ and $ \sum_{k=1}^{\infty} \gamma_k^2 < \infty $;
		\item $ \{\mathcal{F}_k:k\in\mathbb{N}\} $ is a $ \sigma $-algebra sequence with $ \mathcal{F}_{k-1} \subseteq \mathcal{F}_k $ for all $ k $;
		\item $ \{\mathtt{W}_k,\mathcal{F}_k:k\in\mathbb{N}\} $ is a sequence of adaptive random variables satisfying $ \sum_{k=1}^\infty \Abs{\mE\left[ \mathtt{W}_k \middle|  \mathcal{F}_{k-1} \right]} < \infty $ and $ \mE \left[\Absl{\mathtt{W}_k-\mE\left[ \mathtt{W}_k \middle|  \mathcal{F}_{k-1} \right]}^{\rho} \middle| \mathcal{F}_{k-1} \right] = O\left( \beta_k^\rho \right) $ almost surely for some $ \rho > 2 $;
		\item $ \{\mathtt{U}_k:k\in\mathbb{N}\} $ is a sequence with $ \sum_{k=1}^{\infty} \alpha_k^2 \Abs{\mathtt{U}_k}^2 < \infty $. And, $ \mathtt{U}_k $ is $ \mathcal{F}_{k-1} $-measurable;
		\item $ \mathtt{U}_k + \mathtt{U}_k^\top \geq 2 a I_n $ for some $ p \in \mathbb{N} $, $ a > 0 $ and all $ k \in \mathbb{N} $ almost surely; 
		\item $ \{\mathtt{X}_k,\mathcal{F}_k:k\in\mathbb{N}\} $ is a sequence of adaptive random variables with
		\begin{equation}\label{eq:lemma/state}
			\mathtt{X}_k = \left(I_n - \alpha_k \mathtt{U}_k + O(\gamma_k) \right) \mathtt{X}_{k-1} + \mathtt{W}_k, \ \as
		\end{equation}
	\end{enumerate}
	Then, $ \mathtt{X}_k $ converges to $ 0 $ almost surely.
\end{lemmax}

\begin{proof}
	Consider $ \mathtt{X}_k^\prime = \mathtt{X}_k - \mathtt{Y}_k $, where $ Y_0 = 0 $ and $ \mathtt{Y}_k = \left(I_n - \alpha_k \mathtt{U}_k + O(\gamma_k) \right) \mathtt{Y}_{k-1} + \mE\left[ \mathtt{W}_k \middle|  \mathcal{F}_{k-1} \right] $. Since $ \Abs{\mathtt{Y}_k} \leq (1-a\alpha_k+ O(\alpha_k^2\Abs{\mathtt{U}_k}^2 + \gamma_k)) \Abs{\mathtt{Y}_{k-1}} + \Abs{\mE\left[ \mathtt{W}_k \middle|  \mathcal{F}_{k-1} \right]} $, by Lemma 2 of \cite{WangYQ2024Tailoring}, $ \mathtt{Y}_k $ converges to $ 0 $ almost surely. Therefore, it suffices to prove the convergence of $ \mathtt{X}_k^\prime $, which satisfies
	\begin{align}\label{eq:Xprime}
		\!\!\!\!\mathtt{X}_k^\prime = \left(I_n - \alpha_k \mathtt{U}_k + O(\gamma_k) \right) \mathtt{X}_{k-1}^\prime + \mathtt{W}_k - \mE\left[ \mathtt{W}_k \middle|  \mathcal{F}_{k-1} \right]. 
	\end{align}
	
	By \eqref{eq:Xprime}, one can get
	\begin{align}\label{ineq:condiE/lemma_conv}
		& \mE\left[ \Absl{\mathtt{X}_k^\prime}^2 \middle| \mathcal{F}_{k-1} \right] \nonumber\\
		\leq &  \left(  1 - 2 a \alpha_k  + \alpha_k^2  \Abs{\mathtt{U}_k}^2  + O\left( \gamma_k \right)  \right)  \Absl{\mathtt{X}_{k-1}^\prime}^2  + O\left( \beta_k^2 \right).   
	\end{align}
	Then by Lemma 2 of \cite{WangYQ2024Tailoring}, $\mathtt{X}_{k}$ converges to $0$ almost surely.
\end{proof}

\begin{corollaryx}\label{coro:conv}
	If i)-iv) and vi) in \cref{lemma:conv} hold, $ \alpha_k = O\left( \alpha_{k-1} \right) $ and
	\begin{align}\label{condi:PE/coro}
		\frac{1}{p}\sum_{t=k-p+1}^{k} \left(\mathtt{U}_t + \mathtt{U}_t\tr\right) \geq 2 a I_n 
	\end{align}
	for some $ p \in \mathbb{N} $, $ a > 0 $ and all $ k \in \mathbb{N} $ almost surely, 
	then $ \mathtt{X}_k $ converges to $ 0 $ almost surely.
\end{corollaryx}

\begin{proof}
	By \eqref{eq:lemma/state}, $ \mathtt{X}_k =  \prod_{t=k-p+1}^{k} \left( I_n - \alpha_t \mathtt{U}_t + O(\gamma_t) \right) \mathtt{X}_{k-p}
	 + \sum_{t=k-p+1}^{k} \prod_{l=t+1}^{k} \left( I_n - \alpha_l \mathtt{U}_l \right) \mathtt{W}_t$. 
	In this recursive function, $ \prod_{t=k-p+1}^{k} \left( I_n - \alpha_t \mathtt{U}_t + O(\gamma_t) \right) = I_n -  \sum_{t=k-p+1}^{k} \alpha_t \mathtt{U}_t + O\left( \sum_{t=k-p+1}^{k} \left(\gamma_t + \alpha_k^2 \Abs{\mathtt{U}_k}^2\right) \right) $. 
	Note that by \eqref{condi:PE/coro}, 
	\begin{align*}
		& \frac{1}{p}\sum_{t=k-p+1}^{k} \alpha_t \left(\mathtt{U}_t + \mathtt{U}_t\tr\right) \\
		\geq & \left(\min_{k-p < t\leq k} \alpha_t\right) \frac{1}{p}\sum_{t=k-p+1}^{k} \left(\mathtt{U}_t + \mathtt{U}_t\tr\right)
		\geq 2 a \left(\min_{k-p < t\leq k} \alpha_t\right)  I_n,  
	\end{align*}
	and by Lemma A.2 of \cite{Ke2024SCDE}, $ \sum_{k=p+1}^{\infty} \min_{k-p < t \leq k} \alpha_t = \infty $. Then, the corollary can be proved by \cref{lemma:conv}. 
\end{proof}

\begin{lemmax}\label{lemma:1/k}
	If an adaptive sequence $ \{\mathtt{V}_k,\mathcal{F}_k: k \in \mathbb{N} \} $ satisfies $ \mE\left[ \mathtt{V}_k \middle| \mathcal{F}_{k-1} \right] \leq \left( 1 - \frac{a}{k} + \gamma_k \right) \mathtt{V}_{k-1} + O\left( \frac{1}{k^b} \right) $
	with $ a > 0 $, $ b > 1 $ and $ \sum_{k=1}^{\infty}\gamma_k < \infty $, then
	\begin{align*}
		\mathtt{V}_k = \begin{cases}
			O\left(\frac{1}{k^{a}}\right), & \text{if}\ b - a > 1 ;\\
			O\left(\frac{(\ln k)^2}{k^{b-1}}\right), & \text{if}\ b - a \leq 1.
		\end{cases}
	\end{align*}
\end{lemmax}

\begin{proof}
	If $ b-a > 1 $, then
	\begin{align*}
		& \mE\left[ k^a \mathtt{V}_k \middle| \mathcal{F}_{k-1} \right] \\
		\leq & \left( 1 - \frac{a}{k} + \gamma_k \right)\left( 1 + \frac{a}{k} + O\left(\frac{1}{k^2}\right) \right) (k-1)^a \mathtt{V}_{k-1} \\
		& + O\left( \frac{1}{k^{b-a}} \right) \\
		\leq & \left( 1 + \gamma_k + O\left(\frac{1}{k^2}\right) \right) (k-1)^a \mathtt{V}_{k-1} + O\left( \frac{1}{k^{b-a}} \right),
	\end{align*}
	which together with Theorem 1 of \cite{robbins1971convergence} implies that $ \mathtt{V}_k = O\left( \frac{1}{k^a} \right) $ almost surely.
	
	If $ b - a \leq 1 $, then
	\begin{align*}
		& \mE\left[ \frac{k^{b-1}}{(\ln k)^2} \mathtt{V}_k \middle| \mathcal{F}_{k-1}  \right] \\
		\leq & \left( 1 - \frac{a}{k} + \gamma_k \right)\left( 1 + \frac{b-1}{k} + O\left(\frac{1}{k^2}\right) \right) \frac{(b-1)^{b-1}}{(\ln (k-1))^2} \mathtt{V}_{k-1} \\
		& + O\left( \frac{1}{k (\ln k)^2} \right) \\
		\leq & \left( 1 + \gamma_k + O\left(\frac{1}{k^2}\right) \right)\frac{(b-1)^{b-1}}{(\ln (k-1))^2} \mathtt{V}_{k-1} + O\left( \frac{1}{k (\ln k)^2} \right), 
	\end{align*}
	which together with Theorem 1 of \cite{robbins1971convergence} implies that $ \mathtt{V}_k = O\left( \frac{(\ln k)^2}{k^{b-1}} \right) $ almost surely.  The lemma is thereby proved.  
\end{proof}

\begin{lemmax}\label{lemma:1/k^c}
	If sequences $ \{V_k: k \in \mathbb{N}\} $, $ \{\xi_k: k \in \mathbb{N}\} $, $ \{\eta_k: k \in \mathbb{N}\} $ and $ \{\gamma_k: k \in \mathbb{N}\} $ satisfy
	\begin{enumerate}[label={\roman*)},leftmargin=1.4em]
		\item $ \xi_k\geq 0 $, $ \varlimsup_{k\to\infty} \xi_k < 1 $;
		\item $ \sum_{k=1}^{\infty} \eta_k < \infty $, $ \sum_{k=1}^{\infty} \abs{\gamma_k} < \infty $;
		\item $ V_k \leq \left( 1 - \xi_k + \gamma_k \right) V_{k-1} + \eta_k + O\left( \xi_k \right) $,
	\end{enumerate}
	then $ V_k $ is uniformly upper bounded. 	
\end{lemmax}

\begin{proof}
	Without loss of generality, assume $ \gamma_k \geq 0 $. Besides, by $ \sum_{k=1}^{\infty} \gamma_k < \infty $, there exists $ k_0 $ such that $ \gamma_k < \frac{1}{3} $ and $ \xi_k < 1 + \gamma_k $ for all $ k \geq k_0 $. Set $ U_k = \prod_{t=k_0}^{k} \left( 1 - \gamma_t - \frac{\abs{\gamma_t}}{2}\right) \left( V_k - \sum_{t=k_0}^{k} \eta_t \right) $. Then, there exists $ M > 0 $ such that
	\begin{align*}
		U_k =& \prod_{t=k_0}^{k} \left( 1 - \gamma_t - \frac{\abs{\gamma_t}}{2}\right) \left( V_k - \sum_{t=k_0}^{k} \eta_t \right) \\
		\leq & \prod_{t=k_0}^{k} \left( 1 - \gamma_t - \frac{\abs{\gamma_t}}{2}\right) \\
		&\cdot \left( \left( 1 - \xi_k + \gamma_k \right) \left(V_{k-1} - \sum_{t=k_0}^{k-1} \eta_t\right) + O\left( \xi_k + \abs{\gamma_k} \right) \right) \\
		= & \left( 1 - \gamma_k - \frac{\abs{\gamma_k}}{2}\right) \left( 1 - \xi_k + \gamma_k \right) U_{k-1}  + M \left( \xi_k + \abs{\gamma_k} \right). 
	\end{align*}
	If $ U_{k-1} < 2M $, then 
	\begin{align*}
		U_k < & \left( 1 - \gamma_k - \frac{\abs{\gamma_k}}{2}\right) \left( 1 - \xi_k + \gamma_k \right) 2M + M \left( \xi_k + \abs{\gamma_k} \right) \\
		\leq &  \left( 1 - \frac{1}{2}\left( \xi_k + \abs{\gamma_k} \right) \right) 2M + M \left( \xi_k + \abs{\gamma_k} \right) \leq 2M. 
	\end{align*}
	If $ U_{k-1} \geq 2M $, then 
	\begin{align*}
		U_k \leq \left( 1 - \frac{1}{2}\left( \xi_k + \abs{\gamma_k} \right) \right) \mathtt{U}_{k-1} + M \left( \xi_k + \abs{\gamma_k} \right) \leq \mathtt{U}_{k-1}.
	\end{align*}
	Therefore, $ U_{k} \leq \max\{U_{k-1}, 2M\} $, which implies the uniformly boundedness of $ U_k $ and further $ V_{k} $. 
\end{proof}

\begin{lemmax}\label{lemma:conv_rate}
	If i)-vi) in \cref{lemma:conv} hold, $ \rho > 4 $, $ \alpha_k = \frac{1}{k^c} $, $ \beta_k = \frac{1}{k^b} $ for $ c \in (\frac{1}{2},1] $ and $ b > 1 $, and $ \Abs{\mE\left[ \mathtt{W}_k \middle| \mathcal{F}_{k-1} \right]} \leq \lambda^k $ for some $ \lambda\in(0,1) $, then
	\begin{align}\label{concl:rate/lemma}
		\mathtt{X}_k = \begin{cases}
			O\left( \frac{1}{k^{a}} \right), & \text{if}\ c = 1, 2b - 2a > 1; \\
			O\left( \frac{\ln k}{k^{b-1/2}} \right), & \text{if}\ c = 1, 2b - 2a \leq 1; \\
			O\left( \frac{1}{k^{b-c/2}} \right), & \text{if}\ c \in (\frac{1}{2},1),
		\end{cases}
		\ \text{a.s.}
	\end{align}
\end{lemmax}

\begin{proof}
	Consider $ \mathtt{X}_k^\prime $ and $ \mathtt{Y}_k $ in the proof of \cref{lemma:conv}. One can get
	\begin{align*}
		& \frac{\Abs{\mathtt{Y}_k}}{\prod_{t=1}^{k}\left( 1-\frac{a}{k^c} \right)} \\ 
		\leq & (1 + O(\alpha_k^2\Abs{\mathtt{U}_k}^2 + \gamma_k)) \frac{\Abs{\mathtt{Y}_{k-1}}}{\prod_{t=1}^{k-1}\left( 1-\frac{a}{k^c} \right)} + \frac{\lambda^k}{\prod_{t=1}^{k}\left( 1-\frac{a}{k^c} \right)}. 
	\end{align*}
	By Lemma A.2 of \cite{WangJM2024bipartite} and \cref{lemma:1/k^c}, 
	\begin{align}\label{eq:Y_rate}
		\mathtt{Y}_k = \begin{cases}
			O\left( \frac{1}{k^a} \right),& \text{if}\ c = 1;\\
			O\left( \exp\left( \frac{a}{1-c} k^{1-c} \right) \right),& \text{if}\ c \in (\frac{1}{2},1). 
		\end{cases} 
	\end{align}
	Therefore, it suffices to calculate the convergence rate of $ \mathtt{X}_k^\prime $. 	
	
	If $ c = 1 $, then the lemma can be proved by \eqref{ineq:condiE/lemma_conv} and \cref{lemma:1/k}. Then, it suffices to analyze the case of $ c < 1 $. 
	
	For convenience, denote $ \tilde{\mathtt{W}}_k = \mathtt{W}_k - \mE\left[ \mathtt{W}_k \middle| \mathcal{F}_{k-1} \right] $. 
	By \eqref{eq:Xprime}, one can get
	\begin{align}\label{ineq:kX2/c<1}
		& k^{2b-c} \Abs{\mathtt{X}_k^\prime}^2 \nonumber \\
		\leq & \left( 1 - \frac{a}{k^c} + \alpha_k^2 \Abs{\mathtt{U}_k}^2 + O\left( \gamma_k \right) \right) (k-1)^{2b-c} \Absl{\mathtt{X}_{k-1}^\prime}^2 \cr
		& + 2 k^{2b-c} \tilde{\mathtt{W}}_k^\top \left(I_n - \alpha_k \mathtt{U}_k + O(\gamma_k) \right) \mathtt{X}_{k-1}^\prime \nonumber\\
		& + k^{2b-c} \Absl{\tilde{\mathtt{W}}_k}^2,\ \as 
	\end{align}
	Then, by Lemma 2 of \cite{wei1985asymptotic},  
	\begin{align}\label{ineq:sum_WX}
		& \sum_{t=1}^{k} 2 t^{2b-c} \tilde{\mathtt{W}}_t^\top \left(I_n - \alpha_t \mathtt{U}_t + O(\gamma_t) \right) \mathtt{X}_{t-1}^\prime \cr
		\leq & \sum_{t=1}^{k} (2 t^{b} \tilde{\mathtt{W}}_t)^\top \left( t^{b-c} \left(I_n - \alpha_t \mathtt{U}_t + O(\gamma_t) \right) \mathtt{X}_{t-1}^\prime \right) \cr
		= & O\left( 1 \right) + o\left( \sum_{t=1}^{k} t^{2b-2c} \Abs{\mathtt{X}_{t-1}^\prime}^2 \right),\ \as,
	\end{align}
	and
	\begin{align*}
		& \sum_{t=1}^{k} t^{2b-c} \left(\Absl{\tilde{\mathtt{W}}_t}^2 - \mE\left[ \Absl{\tilde{\mathtt{W}}_t}^2 \middle| \mathcal{F}_{k-1} \right]\right) \\
		\leq & \sum_{t=1}^{k} t^{2b} \left(\Absl{\tilde{\mathtt{W}}_t}^2 - \mE\left[ \Absl{\tilde{\mathtt{W}}_t}^2 \middle| \mathcal{F}_{k-1} \right]\right) \cdot \frac{1}{k^c} 
		= O(1),\ \as 
	\end{align*}
	Therefore, $ \mathtt{X}_k = O\left( \frac{1}{k^{b}} \right) $ almost surely, which together with \eqref{ineq:sum_WX} implies $ \sum_{t=1}^{k} 2 t^{2b-c} \tilde{\mathtt{W}}_t^\top \left(I_n - \alpha_t \mathtt{U}_t + O(\gamma_t) \right) \mathtt{X}_{t-1}^\prime = O(1) $.
	Then, the lemma can be proved by \eqref{eq:Y_rate}, \eqref{ineq:kX2/c<1} and \cref{lemma:1/k^c}. 
\end{proof}

\begin{corollaryx}\label{coro:rate}
	Suppose i)-iv), vi) in \cref{lemma:conv} and \eqref{condi:PE/coro} in \cref{coro:conv} hold. $ \rho $, $ \alpha_k $ and $ \beta_k $ are set as \cref{lemma:conv_rate}. Then, $ \mathtt{X}_k $ achieves the almost sure convergence rate as \eqref{concl:rate/lemma}. 
\end{corollaryx}

The proof of \cref{coro:rate} is similar to \cref{coro:conv}, and thereby omitted here. 

\section{Lemmas and propositions on Gaussian, Laplacian and Cauchy distributions}\label[appen]{appen2}

\setcounter{equation}{0}
\renewcommand{\theequation}{B.\arabic{equation}}

	Following lemmas provide some useful properties on Gaussian, Laplacian and Cauchy distributions. 
	
\begin{lemmax}\label{lemma:zeta}
	If the noise $ \mathtt{d}_{ij,k} $ obeys the distribution $ \mathcal{N}(0,\sigma_{ij,k}^2) $ with $ \inf_k \sigma_{ij,k} > 0 $ (resp., $ \text{\textit{Lap}}(0,b_{ij,k}) $ with $ \inf_k b_{ij,k} > 0 $, then $ \zeta_{ij,k} $, $ \text{\textit{Cauchy}}(0,r_{ij,k}) $ with $ \inf_k r_{ij,k} > 0 $), then $ \zeta_{ij,k} $ in iii) of \cref{assum:privacy_noise} can be $ \frac{\sigma_{ij,1}}{\sigma_{ij,k}} $ (resp., $ \frac{b_{ij,1}}{b_{ij,k}} $, $ \frac{r_{ij,1}}{r_{ij,k}} $).
\end{lemmax}

\begin{proof}
	For the Gaussian distribution case, denote $ f_G^\star(\cdot) $ as the density function of the standard Gaussian distribution. Then, $ f_{ij,k} (x) = \frac{1}{\sigma_{ij,k}} f_G^\star\left( \frac{x}{\sigma_{ij,k}} \right) $.
	Since $ \inf_k \sigma_{ij,k} > 0 $, there exists a compact set $ \mathcal{X}^\prime $ such that $ \frac{x}{\sigma_{ij,k}} \in \mathcal{X}^\prime $ for all $ (i,j)\in\mathcal{E} $, $ k\in\mathbb{N} $ and $ x \in \mathcal{X}^\prime $. Therefore, when $ \zeta_{ij,k} = \frac{1}{\sigma_{ij,k}} $, $ \inf_{(i,j)\in\mathcal{E},k\in\mathbb{N},x\in\mathcal{X}} \frac{f_{ij,k}(x)}{\zeta_{ij,k}}
	\geq \inf_{z\in\mathcal{X}^\prime} \frac{f_G^\star(z)}{\sigma_{ij,1}}>0 $, which implies the lemma. The proofs for Laplacian and Cauchy distribution cases are similar the Gaussian one, and hence, omitted here. 
%
\end{proof}

\begin{lemmax}\label{lemma:Lap}
	Given $ b > 0 $, if $ F_{L}(\cdot;b) $ and $ f_L(\cdot;b) $ are the distribution function and the density function of the distribution $ \text{\textit{Lap}}(0,b) $, respectively, then 
	\begin{align*}
		\sup_{x\in\mathbb{R}} \frac{f_L^2(x;b)}{F_{L}(x;b) (1-F_{L}(x;b))} = \frac{1}{b^2}.
	\end{align*}
\end{lemmax}

\begin{proof}
	By $ f_L(x;b) = \frac{1}{2b} \exp\left( - \frac{\abs{x}}{b} \right) $, $ F_L(x;b) = \frac{1}{2} \exp\left( \frac{x}{b} \right) $ if $ x < 0 $; and $ 1 - \frac{1}{2} \exp\left( - \frac{x}{b} \right) $, otherwise. 
	
	By symmetry, 
	\begin{align*}
		& \sup_{x\in\mathbb{R}} \frac{f_L^2(x;b)}{F_{L}(x;b) (1-F_{L}(x;b))} 
		= \sup_{x\geq 0} \frac{f_L^2(x;b)}{F_{L}(x;b) (1-F_{L}(x;b))} \\
		= & \sup_{x\geq 0} \frac{1}{2b^2 \left( \exp\left( \frac{x}{b} \right) - \frac{1}{2} \right)} 
		= \frac{1}{b^2}.
	\end{align*}
	The lemma is thereby proved. 
\end{proof}

\begin{lemmax}\label{lemma:Cauchy}
	Given $ r > 0 $, if $ F_{C}(\cdot;r) $ and $ f_C(\cdot;r) $ are the distribution function and the density function of the distribution $ \text{\textit{Cauchy}}(0,r) $, respectively, then 
	\begin{align*}
		\sup_{x\in\mathbb{R}} \frac{f_C^2(x;r)}{F_{C}(x;r) (1-F_{C}(x;r))} = \frac{4}{\pi^2 r^2}.
	\end{align*}
\end{lemmax}

\begin{proof}
	Since $ f_C(x;r) = \frac{1}{\pi r \left[ 1 + \left( x \middle/ r \right)^2 \right]} $, one can get $ F_C(x;r) = \frac{1}{2} + \frac{1}{\pi} \arctan\left( \frac{x}{r} \right) $. 
	
	Note that $ F_C(x;r) = F_C\left(\frac{x}{r};1\right),\ f_C(x;r) = \frac{1}{r} f_C\left(\frac{x}{r};1\right) $. Then, it suffices to consider the case of $ r = 1 $. 
	In this case, $ F_C(x;1) $ and and $ f_C(\cdot;r) $ are abbreviated as $ F_C(x) $ and $ f_C(x) $, respectively. Denote 
	\begin{align*}
		h_{C,1}(x) \! = \! \frac{F_{C}(x) (1-F_{C}(x))}{f_C^2(x)} \! = \! (1+x^2)^2 \left(\! \frac{\pi^2}{4} - \arctan^2 x \! \right)\!. 
	\end{align*}
	Then, $ h_{C,1}^\prime(x) = (1+x^2) \left( \pi^2 x - 2 \arctan x - 4 x \arctan^2 x \right) $. 
	
	Furthermore, denote $ h_{C,2} (x) = \pi^2 x - 2 \arctan x - 4 x \arctan^2 x $.
	Then, $ h_{C,1}^\prime(x) = (1+x^2) h_{C,2} (x) $, and
	\begin{align*}
		h_{C,2}^\prime (x) = & \pi^2 - \frac{2}{1+x^2} - 4 \arctan^2 x - \frac{8 x \arctan x}{1+x^2}, \\
		h_{C,2}^{\prime\prime} (x) = & \frac{- 4 x - 16 \arctan x}{(1+x^2)^2}. 
	\end{align*}
	Note that $ h_{C,2}^{\prime\prime} (x) > 0 $ when $ x < 0 $; $ h_{C,2}^{\prime\prime} (x) < 0 $ when $ x > 0 $; and $ \lim_{x\to\infty} h_{C,2}^\prime (x) = \lim_{x\to-\infty} h_{C,2}^\prime (x) = 0 $. Then, $ h_{C,2}^\prime (x) > 0 $, which implies that $ h_{C,2}(x) $ is strictly monotonously increasing. Furthermore, by $ h_{C,2} (0) = 0 $, we have $ h_{C,2} (x) < 0 $ when $ x < 0 $; and $ h_{C,2} (x) > 0 $ when $ x > 0 $. 
	
	Note that $ h_{C,1}^\prime(x) = (1+x^2) h_{C,2} (x) $. Then, $ h_{C,1}^\prime(x) < 0 $ when $ x < 0 $; and $ h_{C,1}^\prime(x) > 0 $ when $ x > 0 $. Therefore, 
	\begin{align*}
		\sup_{x\in\mathbb{R}} \frac{f_C^2(x)}{F_{C}(x) (1-F_{C}(x))} = \frac{1}{\inf_{x\in\mathbb{R}} h_{C,1} (x)} = \frac{1}{h_{C,1}(0)} = \frac{4}{\pi^2}. 
	\end{align*}
	The lemma is thereby proved. 
%
\end{proof}

The following propositions gives sufficient conditions on privacy noises satisfying \cref{assum:privacy_noise,assum:step}, 
when the privacy noises are Gaussian, Laplacian and Cauchy. 

\begin{propositionx}\label{prop:assum3}
	For the noise distribution $ \mathcal{N}(0,\sigma_{ij,k}^2) $ (resp., $ \text{\textit{Lap}}(0,b_{ij,k}) $, $ \text{\textit{Cauchy}}(0,r_{ij,k}) $), \cref{assum:privacy_noise} ii) holds when $ \sigma_{ij,k}>0 $ (resp., $ b_{ij,k} > 0 $, $ r_{ij,k} > 0 $), and \cref{assum:privacy_noise} iii) holds when $ \inf_{k\in\mathbb{N}} \sigma_{ij,k} > 0 $ (resp., $ \inf_{k\in\mathbb{N}} b_{ij,k} > 0 $, $ \inf_{k\in\mathbb{N}} r_{ij,k} > 0 $). \MODIFY{Additionally, if $ \epsilon_{ij} \geq 0 $ and $ \sigma_{ij,k} = \sigma_{ij,1} k^{\epsilon_{ij}} $ (resp., $ b_{ij,k} = b_{ij,1} k ^{\epsilon_{ij}}$, $ r_{ij,k} = r_{ij,1} k ^{\epsilon_{ij}} $), then v) of \cref{thm:privacy/finite} holds.}
\end{propositionx}

\begin{proof}
	Consider Gaussian noise case. 
	By Lemma 5.3 of \cite{WangY2024JSSC}, 
\MODIFY{\begin{align}\label{eq:eta_rate}
		\eta_{ij,k} 
		& = \sup_{x\in\mathbb{R}} \frac{f_{ij,k}^2\left(x\right)}{F_{ij,k}(x)\left( 1 - F_{ij,k}(x) \right)} \nonumber\\ 
		& =  \frac{f_{ij,k}^2\left(0\right)}{F_{ij,k}(0)\left( 1 - F_{ij,k}(0) \right)} 
		= \frac{2}{\pi \sigma_{ij,k}^2}.
	\end{align}}Therefore, when $ \sigma_{ij,k} > 0 $, $ \eta_{ij,k} < \infty $. Besides, \cref{lemma:zeta} implies that $ \inf_{k\in\mathbb{N}} \sigma_{ij,k} > 0 $ is sufficient to achieve \cref{assum:privacy_noise} ii). \MODIFY{Additionally, if $ \epsilon_{ij} \geq 0 $ and $ \sigma_{ij,k} = \sigma_{ij,1} k^{\epsilon_{ij}} $, then \eqref{eq:eta_rate} implies v) of \cref{thm:privacy/finite}.}
	
	The analysis for the Laplacian and Cauchy noise cases is similar, and thereby omitted here. 
\end{proof}

\begin{remarkx}
	By \cref{prop:assum3}, for Gaussian and Laplacian privacy noises, \cref{assum:privacy_noise} ii) and iii) can be replaced with the condition that there is a uniform positive lower bound of the noise variances. The reasons to adopt the assumption are twofold. For the privacy, sufficient privacy noises can ensure the privacy-preserving capability of the algorithm. For the effectiveness, the privacy noises are also necessary dithered signals in the quantizers \cite{Ke2024SCDE}. The lack of sufficient dithered signals in the quantizers will result in the algorithm failing to converge. 
\end{remarkx}


\begin{propositionx}\label{prop:assum4}
	For the noise distribution $ \mathcal{N}(0,\sigma_{ij,k}^2) $ (resp., $ \text{\textit{Lap}}(0,b_{ij,k}) $, $ \text{\textit{Cauchy}}(0,r_{ij,k}) $) with $ \sigma_{ij,k} = \sigma_{ij,1} k ^{\epsilon_{ij}} $ (resp., $ b_{ij,k} = b_{ij,1} k ^{\epsilon_{ij}}$, $ r_{ij,k} = r_{ij,1} k ^{\epsilon_{ij}} $) and $ \zeta_{ij,k} = k^{-\epsilon_{ij}} $,
	there exists step-size sequences $ \{\alpha_{ij,k}:(i,j)\in\mathcal{E}, k\in\mathbb{N}\} $ and $ \{\beta_{i,k}:i\in\mathcal{V},k\in\mathbb{N}\} $ satisfying \cref{assum:step} if \MODIFY{and only if $ \epsilon_{ij} \leq \frac{1}{2} $}. 
\end{propositionx}

\begin{proof}
	\MODIFY{By H{\" o}lder inequality \cite{zorich}, $  \left(\sum_{k=1}^\infty \alpha_{ij,k}\zeta_{ij,k}\right)^2 \leq \left(\sum_{k=1}^\infty \alpha_{ij,k}^2\right)\left(\sum_{k=1}^\infty \zeta_{ij,k}^2\right). $ Then, under \cref{assum:step}, $ \sum_{k=1}^\infty \zeta_{ij,k}^2 = \infty $, which implies $ \epsilon_{ij} \leq \frac{1}{2} $. 
	When $ \epsilon_{ij} \leq \frac{1}{2} $, set $ \alpha_{ij,k} = \frac{\alpha_{ij,1}}{k^{1-\epsilon_{ij}}\ln k} $ and $ \beta_{i,k} = \frac{\beta_{i,1}}{k} $. Then, \cref{assum:step} holds. }
\end{proof}

\begin{remarkx}
	$ \zeta_{ij,k} $ in \cref{prop:assum4} is consistent with the one given in \cref{lemma:zeta}. 
\end{remarkx}

\begin{IEEEbiography}[{\includegraphics[width=1in,height=1.25in,clip,keepaspectratio]{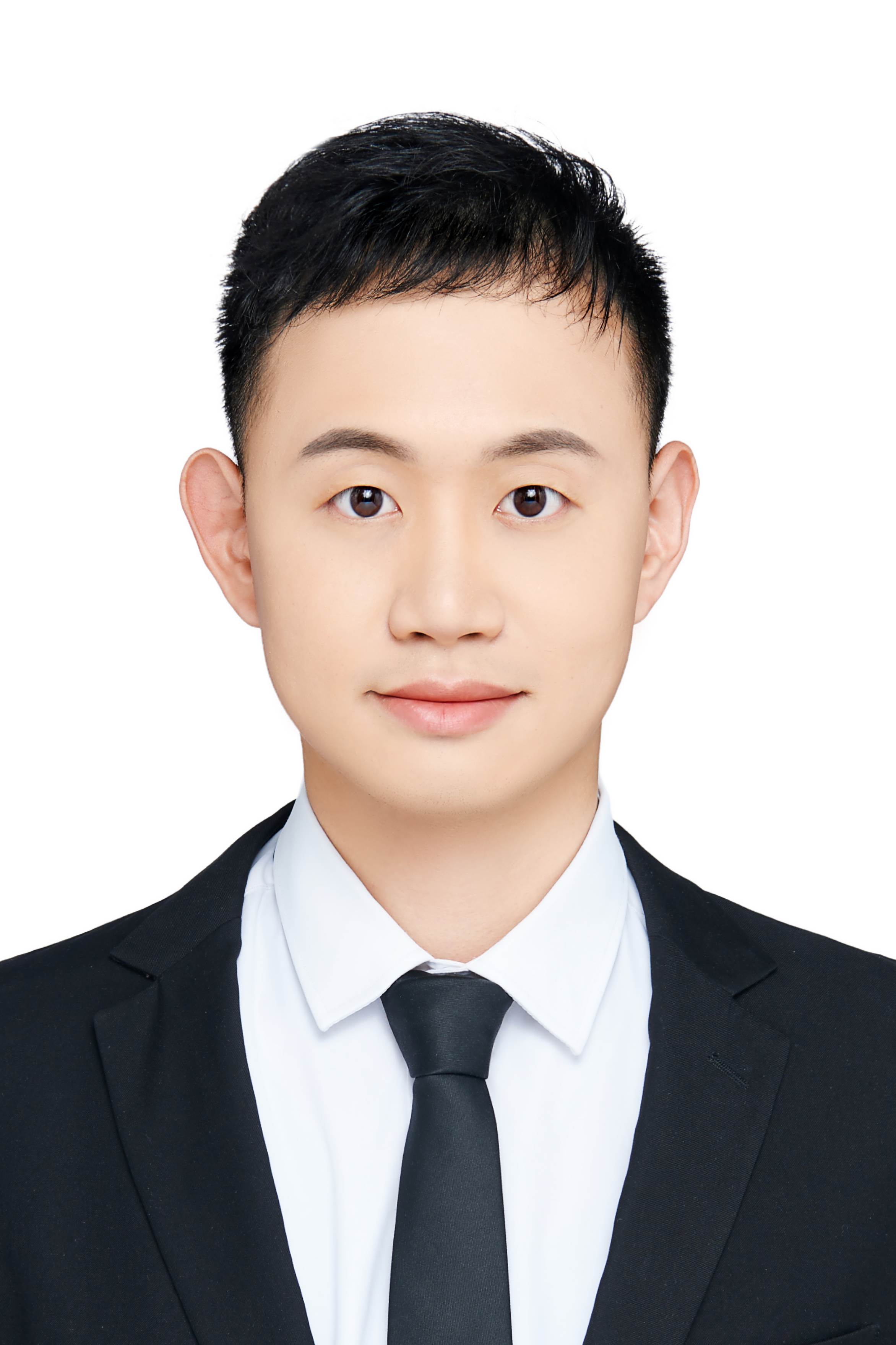}}]{Jieming Ke} received the B.S. degree in mathematics from University of Chinese Academy of Science, Beijing, China, in 2020, and the Ph.D. degree in system theory from the Academy of Mathematics and Systems Science (AMSS), Chinese Academy of Sciences (CAS), Beijing, China, in 2025. 
	
	He is currently a postdoctoral researcher at the Department of Information Engineering, University of Padova, Padova, Italy. His research interests include identification and control of quantized systems, privacy and security in stochastic systems. 
	
\end{IEEEbiography}

\begin{IEEEbiography}[{\includegraphics[width=1in,height=1.25in,clip,keepaspectratio]{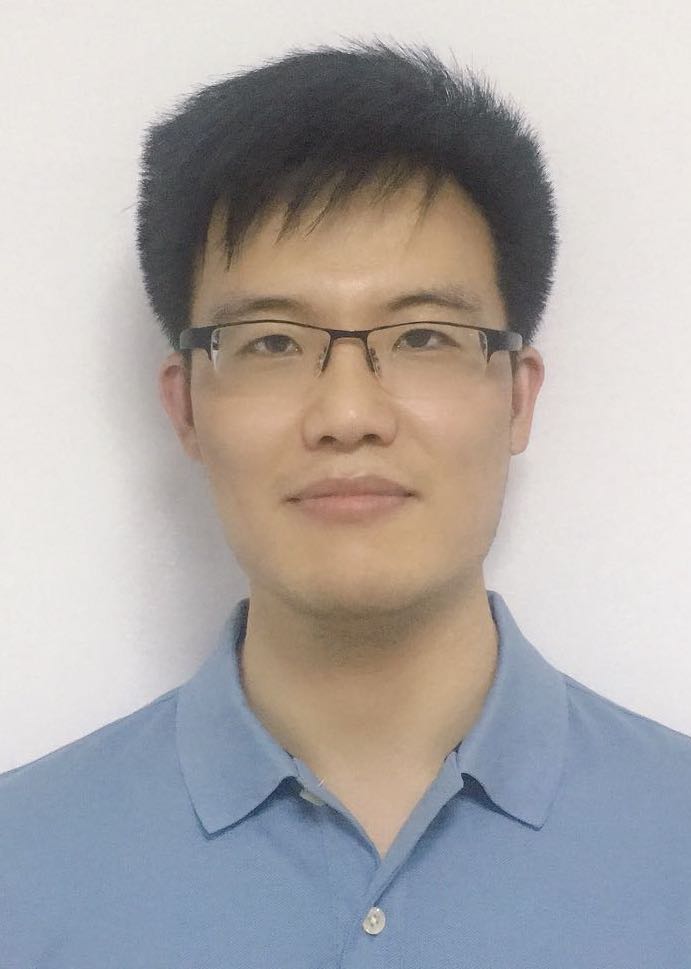}}]{Jimin Wang} (IEEE Member) 
	received the B.S. degree in mathematics from Shandong Normal University, China, in 2012 and the Ph.D. degree from School of Mathematics, Shandong University, China, in 2018. From May 2017 to May 2018, he was a joint Ph.D. student with the School of Electrical Engineering and Computing, The University of Newcastle, Australia. From July 2018 to December 2020, he was a postdoctoral researcher in the Institute of Systems Science, Chinese Academy of Sciences, China. He is currently an associate professor in the School of Automation and Electrical Engineering, University of Science and Technology Beijing. His current research interests include privacy and security in cyber-physical systems, stochastic systems and networked control systems.
	
	He is a member of the IEEE CSS Technical Committee on Security and Privacy, the IEEE CSS Technical Committee on Networks and Communication
	Systems, the IFAC Technical Committee 1.5 on Networked Systems. He was a recipient of Shandong University’s excellent doctoral dissertation. He serves as an Associate Editor for \textit{Systems \& Control Letters} and \textit{Journal of Automation and Intelligence}.
\end{IEEEbiography}

\begin{IEEEbiography}[{\includegraphics[width=1in,height=1.25in,clip,keepaspectratio]{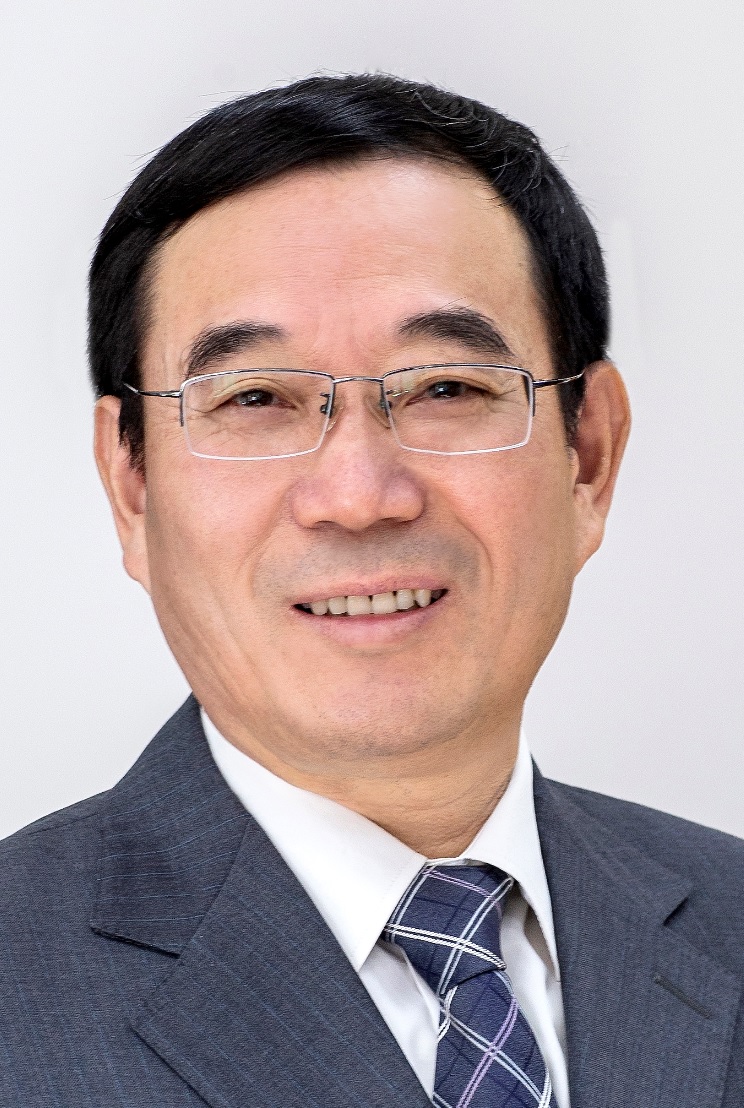}}]{Ji-Feng Zhang} (IEEE Fellow)
	received the B.S. degree in mathematics from Shandong University, China, in 1985, and the Ph.D. degree from the Institute of Systems Science, Chinese Academy of Sciences (CAS), China, in 1991. 
	Now he is with the School of Automation and Electrical Engineering, Zhongyuan University of Technology; and the State Key Laboratory of Mathematical Sciences, Academy of Mathematics and Systems Science, CAS.
	His current research interests include system modeling, adaptive control, stochastic systems, and multi-agent systems.
	
	He is an IEEE Fellow, IFAC Fellow, CAA Fellow, SIAM Fellow, member of the European Academy of Sciences and Arts, and Academician of the International Academy for Systems and Cybernetic Sciences. He received the Second Prize of the State Natural Science Award of China in 2010 and 2015, respectively. He was a Vice-President of the Chinese Association of Automation, the Chinese Mathematical Society and the Systems Engineering Society of China. He was a Vice-Chair of the IFAC Technical Board, member of the Board of Governors, IEEE Control Systems Society; Convenor of Systems Science Discipline, Academic Degree Committee of the State Council of China. He served as Editor-in-Chief, Deputy Editor-in-Chief or Associate Editor for more than 10 journals, including {\em Science China Information Sciences}, {\em IEEE Transactions on Automatic Control} and {\em SIAM Journal on Control and Optimization} etc.
\end{IEEEbiography}

\end{document}